\documentclass[12,reqno]{amsart}

\usepackage{graphicx,psfrag}
\usepackage{amsmath}
\usepackage{amssymb}
\usepackage{float}      
\usepackage{xcolor} 
\usepackage{hyperref,natbib}   
\usepackage[vertfit]{breakurl}  

\setcitestyle{numbers,square,comma}

\newtheorem{theorem}{Theorem}[section]
\newtheorem{lemma}[theorem]{Lemma}
\newtheorem{proposition}[theorem]{Proposition}

\newtheorem{corollary}[theorem]{Corollary}

\theoremstyle{definition}
\newtheorem{definition}[theorem]{Definition}

\newtheorem{assumption}[theorem]{Assumption}

\theoremstyle{remark}
\newtheorem{remark}[theorem]{Remark}

\numberwithin{equation}{section}

\newcommand{\abs}[1]{\lvert#1\rvert}

\renewcommand{\abs}[1]{\ensuremath \left|#1\right|}         

\newcommand{\B}{\ensuremath \mathsf{B}}                    
\newcommand{\Call}{\ensuremath \mathcal{C}}                    
\newcommand{\call}{\ensuremath C}                    
\newcommand{\callg}{\ensuremath \mathsf{C}}                    
\newcommand{\Cbs}{\ensuremath \mathcal{C}_{\mathrm{BS}}}                    
\newcommand{\cbs}{\ensuremath C_{\mathrm{BS}}}                    
\newcommand{\cbsg}{\ensuremath \mathsf{C}_{\mathrm{BS}}}          
\newcommand{\const}{\ensuremath \mathrm{const}}                    
\renewcommand{\d}{\ensuremath \mathrm{d}}                
\newcommand{\EE}{\ensuremath \mathbb{E}}                 
\newcommand{\e}{\ensuremath \mathrm{e}}                     
\DeclareMathOperator{\erfc}{\mathrm{Erfc}}                  
\newcommand{\iv}{\ensuremath \phi}                     
\newcommand{\IV}{\ensuremath \psi}                     

\renewcommand{\O}{\ensuremath \mathrm{O}}           
\newcommand{\PP}{\ensuremath \mathbb{P}}                 
\newcommand{\p}{\ensuremath P}                    
\newcommand{\RR}{\ensuremath \mathbb{R}}
\newcommand{\rcall}{\ensuremath \mathcal{R}}       
\newcommand{\rcalll}{\ensuremath \underline{\mathcal{R}}}       
\newcommand{\rcallu}{\ensuremath \overline{\mathcal{R}}}       
\newcommand{\rcbs}{\ensuremath \mathcal{R}_{\mathrm{BS}}}       
\newcommand{\vhi}{\ensuremath v^\epsilon}         
\newcommand{\z}{\ensuremath Z}                    
\newcommand{\zg}{\ensuremath \mathsf{Z}}                    

\begin{document}

\title[Small/Large Time Implied Volatilities in the MMM]{The Small and Large Time Implied Volatilities in the Minimal Market Model}

\author[Z. Guo]{Zhi Guo}
\address[]{Zhi Guo, School of Computing and Mathematics\\
University of Western Sydney (Parramatta campus) \\
Locked Bag 1797 \\
Penrith, NSW 2751, Australia.
}
\email[]{z.guo@uws.edu.au}

\author[E. Platen]{Eckhard Platen}
\address[]{Eckhard Platen, School of Finance \& Economics and Department of Mathematical Sciences\\
University of Technology, Sydney \\
}
\email[]{eckhard.platen@uts.edu.au}

\curraddr{} \email{}
\thanks{We thank Hardy Hulley for his assistance with the calibration of the MMM parameters.}


\date{\today}



\begin{abstract}
This paper derives explicit formulas for both the small and large time limits of the
implied volatility in the minimal market model. It is shown that interest rates do impact on the implied volatility in the long run even though they are negligible in the short time limit.
\end{abstract}

\maketitle

\section{Introduction}

Proposed by Platen \cite{platen-01, platen-02} to model well diversified stock indices, the minimal market model (\textbf{MMM}) is a flexible one-factor model for capturing real-world price dynamics.
As a local volatility model underpinned by the square-root process, the MMM is not only complete with respect to hedging but also mathematically tractable, with closed form formulas available for forward rates, zero coupon bonds, digital and European options \cite{platen-heath-06}. Further, the MMM has its own volatility feedback mechanism, so unlike stochastic volatility models it does not need an extra volatility process to generate negative correlation between the local volatility and the index, the so-called leverage effect. In addition, the MMM can be extended to model volatility swaps \cite{chan-platen-2011}, commodities and exchange rates in multicurrency markets; and random scaling and jumps can also be embedded in the model to reflect realistic randomness of the market activity \cite[Chapters 13 and 14]{platen-heath-06}.

More importantly, what sets the MMM apart from the other local/stochastic volatility models is its adoption of the benchmark approach \cite{platen-heath-06}, instead of the usual risk-neutral method for derivatives pricing. The benchmark approach does not assume or rely on the
existence of risk-neutral equivalent martingale measures to exclude arbitrage.
Rather, it achieves the elimination of, so called, strong arbitrage by utilizing the growth optimal portfolio of the market as a benchmark for securities and portfolios. Indeed, despite the nonexistence of equivalent risk-neutral measures in the MMM, the benchmark approach accommodates direct arbitrage-free pricing under the original probability measure associated with the underlying asset \cite[Chapters 10, 13]{platen-heath-06}.

In this article we derive both the small and the large time limits of the
implied volatility in the MMM. The derivation of the small time limit takes advantage of an extended Roper--Rutkowski formula \cite{roper-rutkowski-09} for small time implied volatilities. As explained in Section \ref{sec:rr-generalization}, applying a forward price transform can easily extend the model-free Roper--Rutkowski formula to regimes with nonzero interest rates and dividend yields. In contrast, the derivation of the large time limit is based on direct comparisons with the lower and upper bounds of the implied volatility, where the bounds are established by appealing to the asymptotics of the noncentral chi-square distributions.

Following the breakthrough by Berestycki et al. \cite{ber02, ber04}, small time implied volatility asymptotics have been investigated in \cite{alos-et-al-07, henry-labordere-2008, guo-2009, forde-jacquier-2009a, roper-rutkowski-09, forde-et-al-2010, gatheral-et-al, gao-lee-2011}, to name a few studies in the still expanding literature. Whilst covering a diverse range of models, these studies typically assumed zero interest rates, martingale asset prices, or risk neutral regimes. The exception appears to be the paper of Gao and Lee \cite{gao-lee-2011}, of which we learnt after the completion of our work. In \cite{gao-lee-2011}, nonzero interest rates were explicitly allowed and absorbed into forward prices --- a well-known tool that we also use in \eqref{def:xi} below --- and implied volatilities were expanded in terms of option prices in a model-free manner, like that in \cite{roper-rutkowski-09}.  Yet, it does not appear that their zeroth order expansion \cite[Remark 7.4]{gao-lee-2011} implies the Roper--Rutkowski formula or our small time limit. Separately, the article of Gatheral et al. \cite{gatheral-et-al} had also come to our attention. Our small time limit agrees with theirs \cite[(3.21)]{gatheral-et-al}, although we arrived at our result by using a different pricing approach and different techniques.

Comparing to the studies of the small time asymptotics, research in large time implied volatilities has been a more recent event. Rogers and Tehranchi \cite{rogers-tehranchi-09} and Tehranchi \cite{tehranchi-2009-b} examined martingale models. Forde and his coworkers \cite{forde-jacquier-2009b, forde-et-al-2010, forde-2011a, forde-2011b, forde-2011c, figueroa-lopez-et-al-2011} looked at various stochastic volatility models under the assumption of large-time-large-strike, large-time-large-moneyness, and zero interest rate with fixed strike. Besides the aforementioned small time expansion, Gao and Lee \cite{gao-lee-2011} in the same work obtained formulas for large time and extreme strike expansions of the implied volatility in arbitrary order. However, our large time limit complements as much as it is independent of these works. In particular, our explicit formulas for the benchmark approach based limits have demonstrated that interest rates do impact on implied volatilities in the long run, even though they are negligible in the short time limit, see Theorem \ref{thm:small-iv} and \ref{thm:large-iv}. So far, this characterization of the influence of interest rates on implied volatility has not appeared elsewhere.

The organization of this article is as follows. In Section \ref{sec:results} we set up the model and state the main theorems. In Section \ref{sec:calibration} we present a calibrated implied volatility surface and the corresponding small and large time limits.  The extension of the Roper--Rutkowski formula is given in Section \ref{sec:rr-generalization},
and the proofs of the main theorems are in Sections \ref{sec:proof-thm-short-iv} and \ref{sec:proof-large-iv}. Lastly, some auxiliary results and remarks are collected in Section \ref{sec:appendix-a}--\ref{sec:appendix-e}.

\section{Model and main results}
\label{sec:results}

In the stylized MMM \cite[Chapter 13]{platen-heath-06} there exist a savings account and a diversified
accumulation index approximating the growth optimal portfolio of the market. The value of the savings
account $A_t$ grows according to the function
\begin{equation*} 
  A_t = \e^{rt},
  \quad r,t \in [0,\infty),
\end{equation*}
where respectively $r$ and $t$ are the risk-free interest rate and time.
The index price $S_t$ is a square-root process satisfying the equation
\begin{equation} \label{eq:mmm-sde}
  \d S_t =[(r + \sigma^2(S_t,t)] S_t \d t
            + S_t \sigma(S_t,t) \d W_t,
\end{equation}
where $W_t$ is a standard Wiener process on a complete filtered
probability space $(\Omega, \mathcal{A},
\underline{\mathcal{A}},\mathbb{P})$. Dividends for the accumulation index are assumed to be continuously reinvested in the index. The deterministic function $\sigma:
(0,\infty) \times [0,\infty) \to (0,\infty)$ is called the local
volatility; it is defined by
\begin{equation*} 
  \sigma(S,t) = \sqrt{\alpha \e^{(r+\eta)t}/S},
  \quad (S,t) \in (0,\infty)\times [0,\infty).
\end{equation*}
The strictly positive constants $\alpha$ and $\eta$ are, respectively, the initial value and the net growth rate of the growth optimal portfolio of the market. The local volatility $\sigma$ provides volatility feedback to the index price and produces the often observed leverage effects: relatively high (low) asset price leads to relatively low (high) volatility. Figure \ref{fig:sp500-data-call-implied-vol-1} below displays the leverage effect in the SP500 index.

\begin{center}
\begin{figure}[H]
  \includegraphics[width=13cm]{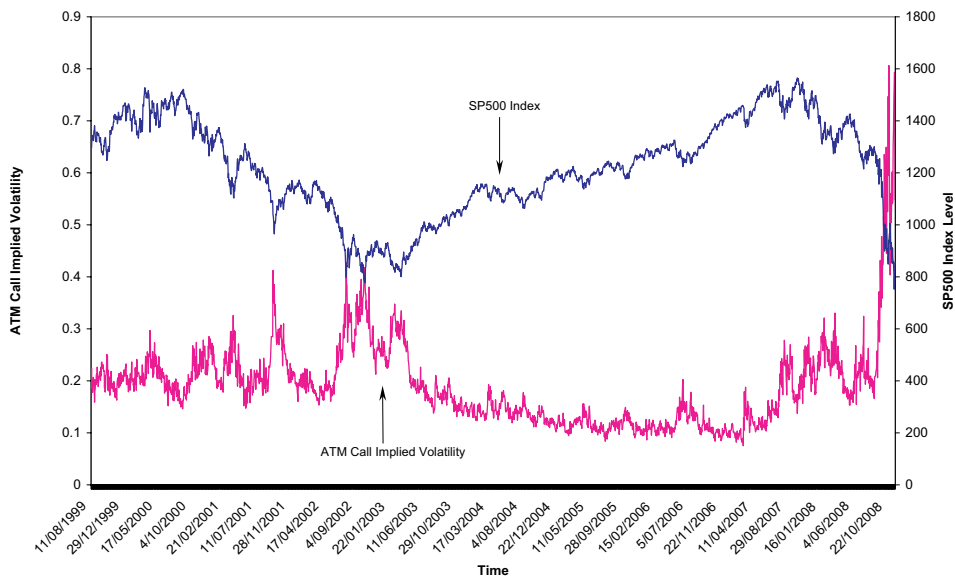} \\
  \caption{Leverage effect: negative correlation between the SP500 Index and its at the money call implied volatility, 11/08/1999--22/10/2008. Data source: Datastream.}\label{fig:sp500-data-call-implied-vol-1}
\end{figure}
\end{center}


\subsection{European option prices under the MMM}

Without loss of generality we will consider European option prices at time $t=0$. In \cite[(10.4.1), (13.3.16)]{platen-heath-06}, it is shown that under the MMM, the European call option price $\call$, denominated in units of the domestic currency, is given by
\begin{equation} \label{eq:call-price-expectation}
  \call(K,T) = S \EE \left[ \left. \frac{(S_T - K)_+}{S_T} \right| S_0 =S \right],
  \qquad 0<S, K < \infty, \quad 0 \le T < \infty,
\end{equation}
where $X_+ = \max(X,0)$, $S$ is the current index price, $K$ the strike, and $T$ the time to expiry. Note that no equivalent risk neutral measure exists in the MMM and $\EE$ is taken directly under the measure $\PP$, see \cite[Chapters 10, 13]{platen-heath-06}.  More explicitly, the MMM call price can be written as
\begin{equation} \label{eq:c-1}
    \call(K,T)
        = S \tilde{\chi}^2(y;4,x)
           - K \e^{-rT} \tilde{\chi}^2(y;0,x),
\end{equation}
where
\begin{equation} \label{eq:c-2}
\left\{
\begin{split}
    x
        & = \frac{S}{\varphi(T)},
            \quad  y = \frac{K \e^{-rT}}{\varphi(T)},
            \quad \varphi(T) = \frac{\alpha}{4\eta} \left(\e^{\eta T}-1\right),
            \quad \alpha, \eta >0, \\
    \tilde{\chi}^2(y;\delta,x)
        & = 1 - \chi^2(y;\delta ,x), \quad \;\; \delta \ge 0, \\
    \chi^2(y;\delta,x)
        & = \int_0^y p(z;\delta,x) \,\d z, \quad \delta>0,
         \quad \chi^2(y;0,x) = \e^{-x/2} + \int_0^y p(z;0,x) \,\d z, \\
    p(y;\delta,x)
        & = \frac{1}{2}
            \left(\frac{y}{x}\right)^{(\delta-2)/4}
            \exp\left(-\frac{x+y}{2}\right)
            I_{(\delta-2)/2}\left(\sqrt{xy}\right), \\
\end{split}
\right.
\end{equation}
and $I_\nu(\cdot)$ is the modified Bessel function of the first
kind with index $\nu$; see \cite[(13.3.17)--(13.3.19)]{platen-heath-06}. For nonnegative $y$, $\delta$, and $x$, the function $\chi^2(y;\delta,x)$ denotes the cumulative distribution function,
evaluated at $y$, of a noncentral chi-square random variable with
$\delta$ degrees of freedom and noncentrality parameter $x$. See e.g. \cite[Chapter
29]{johnson-et-al-95} for details of the distribution with $\delta >0$; see \cite{siegel-79} for the distribution with zero degrees of freedom.

In \cite[(13.3.5), (13.3.20), (13.3.21)]{platen-heath-06}, it is also shown that the European put price $\p$ and zero coupon bond price $\z$ are respectively given by
\begin{eqnarray}
  \p(K,T)
    &=& K \e^{-rT}
            \left(
                    \chi^2(y;0,x)
                - \e^{-x/2}
            \right)
            - S \chi^2(y;4,x), \label{eq:put} \\
  \z(T) &\equiv & \z(0,T) = \e^{-rT} \left(1-\e^{-x/2}\right),  \label{eq:zero-coupond-bond}
\end{eqnarray}
and the following put-call parity relation holds:
\begin{equation} \label{eq:put-call-parity}
  \call(K,T) + K\z(T) = \p(K,T) + S.
\end{equation}

\subsection{The Black--Scholes price and the implied volatility in the MMM}

Assuming a constant dividend yield $\kappa \in \RR$ and a nonincreasing risk-free zero coupon bond price function $T \mapsto \zg(T)\equiv \zg(0,T)$, a general Black--Scholes call price at time $t=0$ can be represented by the formula
\begin{equation} \label{eq:bs-1}
  \cbsg(K,T;v)
    = S\e^{-\kappa T} N(d_1) - K \zg(T) N(d_2),
\end{equation}
where $v$ is the volatility parameter,
\begin{equation} \label{eq:bs-N-d}
\left\{
  \begin{split}
    N(d) & = \int_{-\infty}^d n(\vartheta)\, \d \vartheta,
        \quad n(\vartheta) = \frac{1}{\sqrt{2\pi}} \e^{-\vartheta^2/2}, \\
    d_1(K,T;v)
        & = \frac{\ln (S/K)
        -\ln \zg(T) - \kappa T + v^2 T/2}{v \sqrt{T}},\\
    d_2(K,T;v)
        & = d_1(K,T;v) - v\sqrt{T}.\\
  \end{split}
\right.
\end{equation}
See e.g. \cite[(17.9)]{epps-09}. In Section \ref{sec:rr-generalization} we will use this call price to extend the Roper--Rutkowski formula \cite{roper-rutkowski-09}.
In the MMM, this general formula can be simplified. Firstly, the MMM is concerned with an accumulation index, so $\kappa=0$ in the model.
Secondly, at time $t=0$ the MMM zero coupon bond price $\z$ and the yield-to-maturity of the bond $\hat{r}$ are related by the identities
\begin{equation} \label{eq:rhat-1}
    \hat{r}(T) \equiv \hat{r}(0,T)
        = - \frac{1}{T}
            \ln \bigl(\z(T)\bigr)
        = r - \frac{1}{T} \ln \left(1- \e^{-x/2}\right),
\end{equation}
where $x$ is defined in \eqref{eq:c-2} and the last equality results from \eqref{eq:zero-coupond-bond}; see \cite[(12.2.57) and (13.3.4)]{platen-heath-06}.
(Brigo and Mercurio \cite[p. 6]{brigo-mercurio-01} call $\hat{r}$ the continuously compounded spot interest rate.
Here we follow the terminology of Musiela and Rutkowski \cite[p. 266]{musiela-rutkowski-98} and call it the yield-to-maturity.)
Therefore, in the MMM the Black--Scholes price \eqref{eq:bs-1} can be simplified to
\begin{equation} \label{eq:bs-2}
  \cbs(K,T;v)
    = S N(d_1) - K\e^{-\hat{r}T} N(d_2),
\end{equation}
where $N(\cdot)$ is the same as before and
\begin{equation*}
\left\{
  \begin{split}
    d_1(K,T;v)
        & = \frac{\ln(S/K) +
        (\hat{r}+v^2/2)T}{v\sqrt{T}},\\
    d_2(K,T;v)
        & = \frac{\ln(S/K) +
        (\hat{r}- v^2/2)T}{v\sqrt{T}}.\\
  \end{split}
\right.
\end{equation*}
We are now ready to define implied volatility.

\begin{definition}[Implied volatility in the MMM] \label{def:implied-vol}
Under the MMM, the implied volatility is defined as the unique
nonnegative function $(K,T) \mapsto
\iv(K,T)$ satisfying the equation
\begin{equation} \label{eq:iv-def}
  \call(K,T)
    = \cbs\bigl(K,T;\iv(K,T)\bigr)
\end{equation}
for all $K, T \in (0,\infty)$.
\end{definition}
For $0< T <\infty$, the existence and uniqueness of the
implied volatility $\iv$ is guaranteed by the implicit function
theorem. To see this, let $J = \call(K,T)- \cbs(K,T;v)$. Then the Jacobian
determinant
\begin{equation*}
 \abs{J_v}= \partial_v \cbs(K,T;v) = S n\bigl(d_1(K,T;v)\bigr)\sqrt{T}
\end{equation*}
is strictly positive for all $0 <T <\infty$. However, as $T \to 0$ or $T \to \infty$, the Jacobian
determinant becomes zero. So it is not apparent that the
implied volatility possesses a limit in small time, by which we
mean $\lim_{T \to 0} \iv$, or a limit in large time, by which we mean $\lim_{T \to \infty} \iv$.

\begin{remark} \label{remark:iv-def}
For any finite $T>0$, the existence and uniqueness of the implied volatility can also be deduced by using the general arbitrage bounds for call price and the monotonicity of $\cbsg(K,T;v)$ in $v$; see Section \ref{sec:rr-generalization} below. We omit arbitrage bounds in the definition of the implied volatility because they are automatically satisfied by the MMM call price $\call$; see Step (i) of the proof in Section \ref{sec:proof-thm-short-iv}.
\end{remark}

\subsection{Main results}

For small time asymtotics we have the following theorem:
\begin{theorem} \label{thm:small-iv}
Under the MMM, the implied volatility has the small time limit
\begin{equation} \label{eq:small-iv}
  \lim_{T \to 0} \iv(K,T)
    = \frac{\sqrt{\alpha}\ln(S/K)}{2(\sqrt{S} - \sqrt{K})},
     \qquad \mbox{$K \in (0,\infty)$.}
\end{equation}
\end{theorem}
This theorem makes clear that the risk-free rate does not affect the implied volatility in the small time limit. It confirms the intuition that the time value of money diminishes in infinitesimal time spans and thus has negligible bearing on the option price. The theorem is proved in Section \ref{sec:proof-thm-short-iv}.

For large time asymtotics we have the following theorem:
\begin{theorem} \label{thm:large-iv}
Under the MMM, the implied volatility has the large time limit
\begin{equation} \label{eq:large-iv}
  \lim_{T \to \infty} \iv(K,T)
    = \sqrt{2(3-2\sqrt{2})(r+\eta)}, \qquad \mbox{$K \in (0,\infty)$.}
\end{equation}
\end{theorem}
As a result of this large time limit, the MMM implied volatility in the long run is determined by the risk-free rate $r$ and the net growth rate $\eta$ of the growth optimal portfolio of the market. This is not surprising given that the (long term) increases in the index and option prices are dictated by these two rates.  This theorem is proved in Section
\ref{sec:proof-large-iv}.

\subsection{Notation} \label{sec:notation}
If $F(\cdot)$ is a probability distribution, then $\tilde{F}(\cdot) = 1 - F(\cdot)$ is the complementary distribution function of $F$. Except in the introduction, subscript letters generally denote
partial derivatives, e.g. $x_T = \partial x/\partial T$. Limits and asymptotics have the following meanings:
\begin{equation*}
\begin{split}
    f(T) \xrightarrow{\; T \to l \;} h
    & \quad \Longleftrightarrow  \quad
        \lim_{T \to l} f(T) = h. \\
    f(T) \sim g(T)
    \quad (T \to l)
    & \quad \Longleftrightarrow \quad
        \lim_{T \to l} [f(T)/g(T)] = 1,
        \quad l \in [-\infty,\infty]. \\
    f(T) = \O(g(T))
    \quad (T \to \infty)
    & \qquad \mbox{if } \qquad
       \mbox{$\abs{f(T)/g(T)}$ is bounded in the limit}.
\end{split}
\end{equation*}
Given strike $K$ and expiry $T$, the MMM call price is $\call(K,T)$, and the corresponding Black--Scholes price with implied volatility $\phi$ is $\cbs(K,T,\iv(K,T))$. $\cbs(K,T;v)$ stands for a generic Black--Scholes price with volatility $v \in [0,\infty]$.

\section{Implied volatility calibration} \label{sec:calibration}

To estimate the model parameters we calibrated the MMM on the SP500 total return index (SPX) using data obtained from Datastream for the period 04/01/1988--27/01/2009. A similar calibration procedure has been performed in \cite{hulley-platen-2008}.

On 27/01/2009, the SP500 index value had a value of $1362.18$ and used as a proxy for the (annualized) risk-free rate $r$, the effective 3-month U.S. T-bill rate on the same day was $0.0011154$ per annum. The calibration returned the estimates $\alpha = 43.307$ and $\eta = 0.089896$. These calibrated parameters were then fed into the MMM formula to produce call prices on the SPX. Figure \ref{fig:iv_surface_2011_05_16} shows an implied volatility surface generated from the MMM call prices on 27/01/2009. Also plotted in the graph are the theoretical small and large time implied volatility limits.

From Figures \ref{fig:iv_surface_2011_05_16} and
\ref{fig:iv_vs_strike_2011_05_16} it can be seen that as the maturity shortens, the implied volatility decreases
to the theoretical limit and the skew becomes more pronounced. In comparison, Figures
\ref{fig:iv_surface_2011_05_16}--\ref{fig:iv_vs_time_2011_05_16} illustrate that as the maturity lengthens, the implied volatility
converges to the theoretical large time limit, and the skew flattens at a
decreasing speed. Our observation of the large time asymptotics is consistent with the findings of Rogers and
Tehranchi \cite{rogers-tehranchi-09} and Forde and Jacquier
\cite{forde-jacquier-2009b}, even though these authors work in risk-neutral regimes.

\begin{center}
\begin{figure}[H]
  \includegraphics[width=10cm]{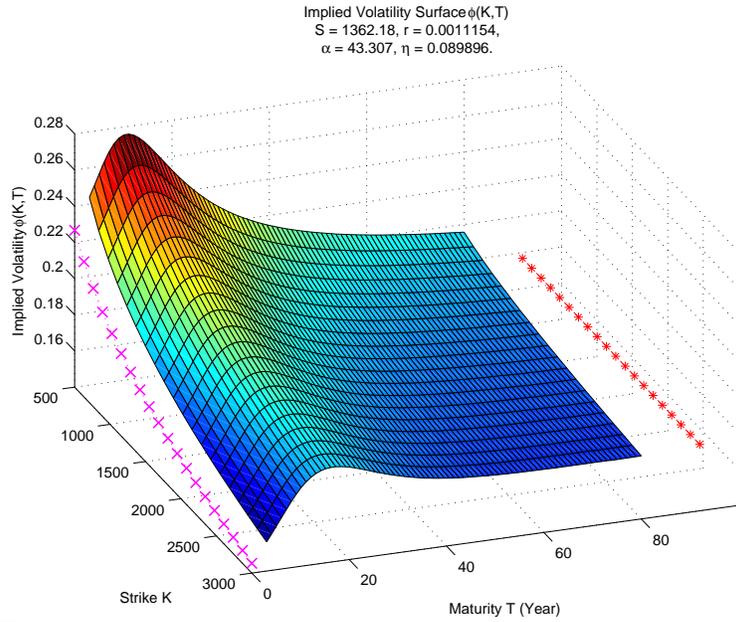} \\
  \caption{SP500 index implied volatility under the MMM and the theoretical small and large time limits on 27/01/2009.}\label{fig:iv_surface_2011_05_16}
\end{figure}
\end{center}

\begin{center}
\begin{figure}[H]
  \includegraphics[width=10cm]{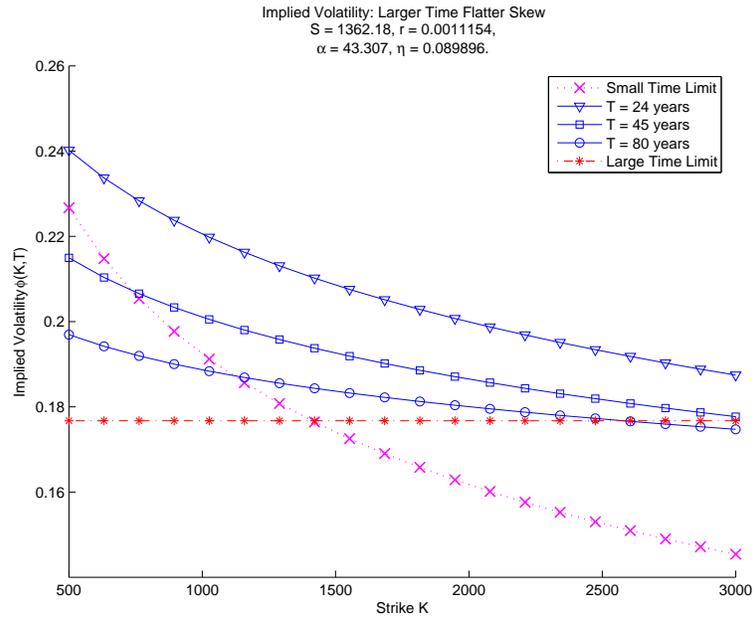} \\
  \caption{The small and large time behavior of the SP500 index implied volatility skew under the MMM on 27/01/2009.}\label{fig:iv_vs_strike_2011_05_16}
\end{figure}
\end{center}

\begin{center}
\begin{figure}[H]
  \includegraphics[width=10cm]{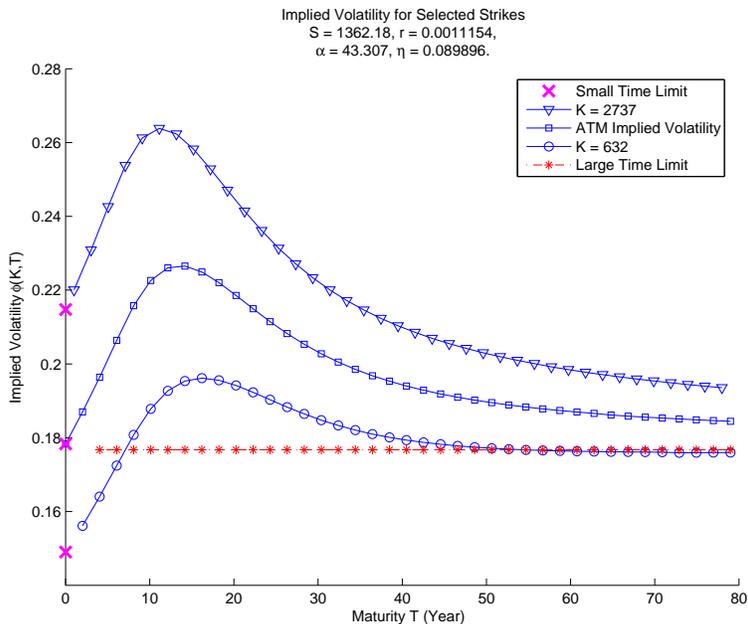} \\
  \caption{The small and large time limits of the SP500 index implied volatility under the MMM on 27/01/2009 for different strikes.}\label{fig:iv_vs_time_2011_05_16}
\end{figure}
\end{center}

\section{An extended Roper--Rutkowski formula} \label{sec:rr-generalization}

Under the assumption of zero risk-free interest rate and some
minimal conditions on the call option prices, Roper and Rutkowski
\cite[Theorem 5.1]{roper-rutkowski-09} derived a model-free zeroth order asymptotic formula
for the implied volatility in small time.

In this section we extend their formula to markets with nonzero dividend yields and interest rates. Since bond prices can be parametrized
by risk-free interest rates, we will, instead of specifying a
risk-free rate, introduce a risk-free zero coupon bond into the
Roper--Rutkowski setup. We will derive the extended formula by applying a well-known forward price transform. After the variable change, it will become clear that the Roper--Rutkowski proof can be repeated here almost line by line. For this reason we will only sketch our proof of the result.

\subsection{The general market model and the extended Roper--Rutkowski formula}

For ease of referencing we shall call our setup a general market model (\textbf{GMM}). Consider a market that has a continuum of zero coupon bond prices and call option
prices for an asset. Without loss of generality we study the market at time $t=0$. Let the constant dividend yield be $\kappa \in \RR$ and current asset price $S>0$.  For the bond price function $T \mapsto \zg(T)$ we have the following assumptions.
\begin{assumption} \label{ass:general-zero-bond}
The bond price $\zg:[0,\infty) \to (0,1]$ satisfies the following
conditions.
\begin{enumerate}
\item[($\zg1$)] No arbitrage bounds:
\begin{equation}
    0<\zg(T) \le 1, \quad \forall \; T \in [0,\infty).
\end{equation}
\item[($\zg2$)] Convergence to payoff:
\begin{equation}
    \lim_{T \to 0} \zg(T) =\zg(0) = 1.
\end{equation}
\item[($\zg3$)] Time value of money:
\begin{equation}
    T \mapsto \zg(T) \quad \mbox{is nonincreasing}.
\end{equation}
\end{enumerate}
\end{assumption}
For the call prices $(K,T) \mapsto \callg(K,T)$ the following conditions are also assumed.
\begin{assumption} \label{ass:general-call}
The call price $\callg: (0,\infty)\times [0,\infty) \to [0,\infty)$ fulfils the
following conditions:
\begin{enumerate}
\item[($\callg 1$)] No arbitrage bounds:
\begin{equation} \label{eq:cond-callg-1}
    (S\e^{-\kappa T} - K \zg(T))_+ \le \callg(K,T) \le S \e^{-\kappa T}, \quad \forall \; S, K >0, \;T
    \ge 0.
\end{equation}
\item[($\callg 2$)] Convergence to payoff:
\begin{equation} \label{eq:cond-callg-2}
    \lim_{T \to 0} \callg(K,T) = \callg(K,0) = (S-K)_+.
\end{equation}
\item[($\callg 3$)] Time value of the option:
\begin{equation} \label{eq:cond-callg-3}
    T \mapsto \callg(K,T) \quad \mbox{is nondecreasing.}
\end{equation}
\end{enumerate}
\end{assumption}
If we set $\kappa = 0$ and $\zg(T)\equiv 1$ in the setup above, then we recover the zero dividend yield and zero interest rate setup of Roper and Rutkowski \cite[Section 2]{roper-rutkowski-09}. With some abuse of notation we can now define implied volatility for the GMM.
\begin{definition}[Implied volatility in the GMM] \label{def:general-implied-vol}
Under the GMM, the implied volatility is defined as the unique
nonnegative function $(K,T) \mapsto
\iv(K,T)$ satisfying the equation
\begin{equation*} 
  \callg(K,T)
    = \cbsg\bigl(K,T;\iv(K,T)\bigr) \qquad \forall \; K, T \in (0,\infty),
\end{equation*}
where $\cbsg$ is defined in \eqref{eq:bs-1}.
\end{definition}
As mentioned earlier in Remark \ref{remark:iv-def}, the existence and uniqueness of the implied volatility is guaranteed by the arbitrage bounds for the call price in \eqref{eq:cond-callg-1} and the monotonicity of $\cbsg(K,T;v)$ in $v$. Here is our result.
\begin{theorem}[Extended Roper--Rutkowski formula] \label{thm:roper-rutkowski}
Under the GMM, if there exists a constant $T_1>0$ such that $\callg(K,T) > (S\e^{-\kappa T}-K\zg(T))_+$ for every
fixed $K>0$ and $T \in (0,T_1)$, then as $T \to 0$,
\begin{equation} \label{eq:rr-formula-extended}
    \lim_{T\to 0} \iv(K,T)
        =
        \left\{%
\begin{array}{ll}
    \displaystyle \lim_{T \to 0} \sqrt{2\pi} \dfrac{\callg(K,T)}{K\sqrt{T}}, & \hbox{$K=S$,} \\
    \displaystyle \lim_{T \to 0}
        \dfrac{\abs{\ln(S/K)}
                }{
                    \left\{
                        -2T\ln[\callg(K,T)-(S\e^{-\kappa T} -K\zg(T))_+]
                    \right\}^{1/2}
                  }, & \hbox{$K \ne S $,} \\
\end{array}
\right.
\end{equation}
where the equality of the limits is understood in the sense that the left-hand side limit exists (is infinite) if the right-hand side limit exists (is infinite).
\end{theorem}

\begin{remark}
Similar to the case discussed in \cite[Section 5.2]{roper-rutkowski-09}, if there exists a $T_0$ such that $\callg(K,T) = (S\e^{-\kappa T}-K\zg(T))_+$ for every
fixed $K>0$ and $T \in (0,T_0)$, then obviously $\iv(K,T)=0$ for every $T \in (0,T_0)$.
\end{remark}

\begin{remark}
If $\kappa =0$ and $\zg(T) \equiv 1$ for all $T \in [0,\infty)$, then \eqref{eq:rr-formula-extended} is reduced to the Roper--Rutkowski formula \cite[Corollary 5.1]{roper-rutkowski-09}.
\end{remark}

\subsection{Proof of the extended Roper--Rutkowski formula}

To prove their formula Roper and Rutkowski rely on a representation formula
for the Black--Scholes call price, which states that
\begin{equation} \label{eq:bs-rep-rr}
    \B(K,T;v)
        = (S-K)_+
            + S \int_0^{v\sqrt{T}}
                N'\left( \frac{\ln(S/K)}{\tau} + \frac{\tau}{2}\right) \, \d \tau,
\end{equation}
where $N'(\cdot) \equiv n(\cdot)$ is the standard normal density in \eqref{eq:bs-N-d}; see \cite[Lemma
3.1]{roper-rutkowski-09}. A variant of this formula had earlier appeared in Carr and Jarrow \cite{carr-jarrow-92}.

However, \eqref{eq:bs-rep-rr} does not hold when the risk-free interest rate is not zero, or equivalently when there is a nontrivial zero coupon bond. Indeed, if in the GMM, $\kappa =0$ and $\zg(T)= \e^{-rT}$ for $T \ge 0$ and some $r>0$, then $\cbsg(K,T;v) \ne \B(K,T;v)$ for $K,T >0$. Yet, using the forward price we can derive a representation formula similar to \eqref{eq:bs-rep-rr} for the Black--Scholes price $\cbsg$ in the GMM. For $S,K > 0$ and $T\ge 0$, the forward price $\xi$ in the MMM is
\begin{equation} \label{def:xi}
    \xi= \ln(S/K) - \ln\bigl(\zg(T)\bigr)-\kappa T.
\end{equation}
In the $(\xi,T)$ coordinates, let $\Call$ be the transformed call price
\begin{equation}
  \Call(\xi,T) = \frac{\callg\bigl(K(\xi,T),T\bigr)}{K(\xi,T)\zg(T)}.
\end{equation}
Then in $(\xi,T)$ the conditions \eqref{eq:cond-callg-1}--\eqref{eq:cond-callg-3} can be written as
\begin{equation} \label{eq:call-conditions-in-xi-T}
\left\{
\begin{split}
    & (\e^{\xi}-1)_+ \le \Call(\xi,T) \le \e^{\xi}, \quad \xi
        \in \RR, \quad T \in [0,\infty), \\
    & \lim_{T \to 0} \Call(\xi,T) = (\e^{\xi} -1)_+,
        \quad \xi \in \RR, \\
    & T \mapsto \Call(\xi,T) \quad \mbox{is nondecreasing}.
\end{split}
\right.
\end{equation}
Note that the corresponding Black--Scholes price in $(\xi,T)$ is
\begin{equation*}
\begin{split}
    \Cbs(\xi,T;v)
        & = \frac{\cbsg\bigl(K(\xi,T),T;v\bigr)}{K(\xi,T)\zg(T)} 
         = \e^{\xi}N\bigl(d_1(\xi,T;v)\bigr) - N\bigl(d_2(\xi,T;v)\bigr), \\
\end{split}
\end{equation*}
where $v$ is the volatility parameter, $d_1(\xi,T;v) = (\xi + v^2T/2)/(v\sqrt{T})$, and $d_2(\xi,T;v) = d_1(\xi,T;v)-v\sqrt{T}$.

Let $\IV$ be the implied volatility in the $(\xi,T)$ coordinates,
i.e.,
\begin{equation} \label{eq:implied-vol-eqaulities}
\begin{split}
    \IV(\xi,T) & = \iv(K(\xi,T),T). \\
\end{split}
\end{equation}
Then by the definition of the implied volatility, in either the $(K,T)$ or $(\xi,T)$ coordinates, we have, for all $\xi \in \RR$ and $T \in (0,\infty)$,
\begin{equation} \label{eq:implied-vol-price-eqaulities}
    \Call(\xi,T)
    = \frac{\callg\bigl(K(\xi,T),T\bigr)}{K(\xi,T)\zg(T)}
    =\frac{\cbsg\bigl(K(\xi,T),T;\iv(K(\xi,T),T)\bigr)}{K(\xi,T)\zg(T)}
    =\Cbs(\xi,T;\IV(\xi,T)).
\end{equation}
Moreover, by following \cite[Lemma 3.1]{roper-rutkowski-09} we deduce the representation formula
\begin{equation} \label{eq:bs-rep}
    \Cbs(\xi,T;v)
        = (\e^{\xi}-1)_+ + \e^{\xi}
            \int_0^{v \sqrt{T}}
                N^\prime
                \left(\frac{\xi}{\tau}+\frac{\tau}{2}\right)
                \, \d \tau,
\end{equation}
where again $N^\prime(\cdot) = n(\cdot)$ is
the standard normal density; c.f. \eqref{eq:bs-rep-rr}.

Now we are ready to present the proof of Theorem \ref{thm:roper-rutkowski}. Since we will largely make use of the results in \cite{roper-rutkowski-09}, our proof will be brief.

\begin{proof}[Proof of Theorem \ref{thm:roper-rutkowski}]
 Let
\begin{equation}
    F(\xi,\theta) =\int_0^\theta
                N^\prime
                \left(\frac{\xi}{\tau}+\frac{\tau}{2}\right)
                \, \d \tau,
                \qquad \xi \in \RR, \quad \theta \ge 0.
\end{equation}
Then by \cite[Lemmas 5.1 and
5.2]{roper-rutkowski-09}, we get, as $\theta \to 0$,
\begin{equation} \label{eq:F-asym}
    F(\xi,\theta)
    \sim
    \left\{
      \begin{array}{ll}
        \displaystyle \frac{\theta}{\sqrt{2\pi}}, & \hbox{$\xi=0$,} \\
        \mbox{} \\
        \displaystyle \frac{\theta^3}{\sqrt{2\pi}\xi^2}
            \exp\left(-\frac{\theta^2 \xi + \xi^2}{2\theta^2}\right), & \hbox{$\xi\ne 0$.}
      \end{array}
    \right.
\end{equation}
By \eqref{eq:bs-rep}, we have
\begin{equation} \label{eq:bs-rep-1}
    \Cbs(\xi,T;\IV(\xi,T))
        = (\e^{\xi}-1)_+ + \e^{\xi}
            \int_0^{\IV(\xi,T) \sqrt{T}}
                N^\prime
                \left(\frac{\xi}{\tau}+\frac{\tau}{2}\right)
                \, \d \tau.
\end{equation}
By Assumption \ref{ass:general-zero-bond}, and Assumption \ref{ass:general-call} in the form expressed by \eqref{eq:call-conditions-in-xi-T}, and by following \cite[Proposition 4.1]{roper-rutkowski-09},
we get
\begin{equation} \label{eq:IV-zero}
    \IV(\xi,T)\sqrt{T} \xrightarrow{\;T \to 0 \;} 0,
    \qquad \xi \in \RR.
\end{equation}
Now from \eqref{eq:implied-vol-price-eqaulities} and \eqref{eq:bs-rep-1} we get
\begin{equation} \label{eq:C-F}
    \frac{\Call(\xi,T) - (\e^{\xi}-1)_+}{\e^{\xi}}=
F\left(\xi,\IV(\xi,T)\sqrt{T}\right).
\end{equation}
Then a combination of
\eqref{eq:C-F}, \eqref{eq:IV-zero}, \eqref{eq:F-asym}, and an application of the
same procedure in \cite[Theorem 5.1, Corollary
5.1]{roper-rutkowski-09} would show that as $T \to 0$,
\begin{equation*} 
    \IV(\xi,T)
        \sim
        \left\{%
\begin{array}{ll}
    \sqrt{2\pi} \dfrac{\Call(\xi,T)}{\sqrt{T}}, & \hbox{$\xi=0$,} \\
    \dfrac{\abs{\xi}}{\Bigl\{-2T\ln[\Call(\xi,T)-(\e^{\xi}-1)_+]\Bigr\}^{1/2}}, & \hbox{$\xi \ne 0$.} \\
\end{array}%
\right.
\end{equation*}
Respectively, $\xi=0$ and $\xi\ne 0$ correspond to the
at the money $(K=S)$ and the not at the money $(K \ne S)$ cases. In fact, by \eqref{eq:implied-vol-eqaulities} and \eqref{eq:implied-vol-price-eqaulities}, a back transformation of the above asymptotic formula to the $(K,T)$ coordinates gives, as $T \to 0$,
\begin{equation*} 
    \iv(K,T)
        \sim
        \left\{%
\begin{array}{ll}
    \sqrt{2\pi} \dfrac{\callg(K,T)}{K\zg(T)\sqrt{T}}, & \hbox{$K=S$,} \\
    \dfrac{\abs{\ln(S/K) - \ln(\zg(T)) - \kappa T}}{\left\{-2T\ln\left[\dfrac{\callg(K,T)}{K\zg(T)}-\left(\dfrac{S}{K\zg(T)\e^{\kappa T}}-1\right)_+\right]\right\}^{1/2}}, & \hbox{$K \ne S$.} \\
\end{array}%
\right.
\end{equation*}
Taking the limits then gives the desired expressions in \eqref{eq:rr-formula-extended}.
\end{proof}

\section{Proof of the small time limit: Theorem \ref{thm:small-iv}} \label{sec:proof-thm-short-iv}

The proof of the small time limit is an application of Theorem \ref{thm:roper-rutkowski}. It takes the following steps:
\begin{enumerate}
\item[(i)] verification of
Assumptions \ref{ass:general-zero-bond} and \ref{ass:general-call};
\item[(ii)] computation of the at the money limit;
\item[(iii)] computation of the out of the money limit;
\item[(iv)] computation of the in the money limit.
\end{enumerate}
Note that the dividend yield $\kappa =0$ in the MMM; see the discussion following \eqref{eq:bs-N-d}.

\begin{proof}[Proof of Theorem \ref{thm:small-iv}] \mbox{}

\bigskip

\noindent \textbf{Step (i): Verification of Assumptions
\ref{ass:general-zero-bond} and \ref{ass:general-call}.} It is easy to verify
that the bond price $\z$ satisfies Assumption \ref{ass:general-zero-bond}. We omit the details.

To verify Assumption \ref{ass:general-call} we will check \eqref{eq:cond-callg-1} first. By \eqref{eq:zero-coupond-bond}, in the MMM \eqref{eq:cond-callg-1} becomes
\begin{equation*}
  \left(S - K\e^{-rT}(1-\e^{-x/2}) \right)_+ \le \call(K,T) \le S, \quad \forall \; S, K >0, \;T
    \ge 0.
\end{equation*}
Since $\chi^2(y;4,x)$ and $\chi^2(y;0,x)$ are distributions, \eqref{eq:c-1} implies that
$\call(K,T) \le S$ for all $K>0$ and $T\ge 0$. This proves
the upper bound for $\call$. To derive the lower bound we will check two cases. When $S \le
K\e^{-rT}(1-\e^{-x/2})$, we need $\call\ge 0$. This is
obviously true considering that in \eqref{eq:call-price-expectation} the payoff function is nonnegative and $S_T$ is a nonnegative process. When $S > K\e^{-rT}(1-\e^{-x/2})$, to derive the lower
bound for $\call$ is to prove the inequality
\begin{equation} \label{eq:proof:fulfill-lower-bound}
  S - K\e^{-rT}(1-\e^{-x/2})
  \le S [1- \chi^2(y;4,x)]
     - K\e^{-rT}[1-\chi^2(y;0,x)],
\end{equation}
which is the same as to prove
\begin{equation}
\begin{split}
    \frac{x}{y}
    =\frac{S}{K\e^{-rT}}
    & \le \frac{\chi^2(y;0,x)-\e^{-x/2}}{\chi^2(y;4,x)}. \\
\end{split}
\end{equation}
This inequality can be proved by using \eqref{eq:c-2} as
\begin{equation}
\begin{split}
    & \frac{\chi^2(y;0,x)-\e^{-x/2}}{\chi^2(y;4,x)} \\
    & = (x/y)
        \left[
        y \left. \int_0^y \frac{1}{\sqrt{z}} \e^{-z/2}I_1(\sqrt{xz}) \,\d z
        \right/ \int_0^y \sqrt{z} \e^{-z/2}I_1(\sqrt{xz}) \,\d z
        \right] \\
    & \ge x/y.
\end{split}
\end{equation}
Note that the last inequality above is valid because
\begin{equation}
\begin{split}
    \int_0^y \sqrt{z} \e^{-z/2}I_1(\sqrt{xz}) \,\d z
    & = \int_0^y \frac{z}{\sqrt{z}} \e^{-z/2}I_1(\sqrt{xz}) \,\d z    \\
    & \le y \int_0^y \frac{1}{\sqrt{z}} \e^{-z/2}I_1(\sqrt{xz}) \,\d  z.
\end{split}
\end{equation}
Thus \eqref{eq:proof:fulfill-lower-bound} holds for all $S,K>0$ and $T \ge 0$. Consequently, $\call$ satisfies \eqref{eq:cond-callg-1}.

Next, $\call$ also satisfies condition
\eqref{eq:cond-callg-2} by Lemma \ref{lem:c-T-0} below. Moreover, using \eqref{eq:cT}, the fact that $x_T/x = - \eta
\e^{\eta T}/(\e^{\eta T} -1)$, and the property that $p(y;\delta,x)$ and
$\tilde{\chi}^2$ are nonnegative, we get $\call_T \ge 0$ for all
$K,T \in (0,\infty)$. Hence $\call$
satisfies \eqref{eq:cond-callg-3} too.

In sum, $\call$ satisfies
the conditions \eqref{eq:cond-callg-1}--\eqref{eq:cond-callg-3} of Assumption \ref{ass:general-call}.
\bigskip

\noindent \textbf{Step (ii): At the money small time limit.}
When $K=S$, \eqref{eq:small-iv} becomes
\begin{equation}
  \left[\lim_{T \to 0} \iv(K,T)\right]_{K=S}
    = \lim_{K \to S} \frac{\sqrt{\alpha}\ln(S/K)}{2(\sqrt{S} - \sqrt{K})}
    = \sqrt{\frac{\alpha}{S}}, \quad S \in
            (0,\infty).
\end{equation}
To get this limit we shall apply \eqref{eq:rr-formula-extended}, which gives,
for $K=S$,
\begin{equation}
  \iv(S,T) \sim \sqrt{2\pi} \frac{\call(S,T)}{S\sqrt{T}}
  \qquad (T \to 0).
\end{equation}
Since $\call(S,T) \xrightarrow{\; T \to 0 \;} 0$, we apply L'Hopital's rule to get
\begin{equation} \label{eq:atm-proof-1}
\begin{split}
  \lim_{T \to 0}
    \sqrt{2\pi} \frac{\call(S,T)}{S \sqrt{T}}
         = \lim_{T \to 0} \sqrt{2\pi}
            \frac{\call_T(S,T)}{S/(2\sqrt{T})}
        & = \lim_{T \to 0} \frac{2\sqrt{2\pi}}{S}
            \sqrt{T}\call_T(S,T),
\end{split}
\end{equation}
provided the last limit exists. Now by \eqref{eq:cT} below,
\begin{equation} \label{eq:atm-proof-2}
  \call_T(K,T)
    = - \frac{2 S}{x} x_T p(y;4,x)
        + r K \e^{-rT} \tilde{\chi}^2(y;0,x).
\end{equation}
Then, by using firstly the fact that
$\sqrt{T}\tilde{\chi}^2(y;0,x) \xrightarrow{\;T\to 0\;} 0$,
and secondly Lemma \ref{lem:atm-prelim-1} (1), we get, when $K=S$,
\begin{equation}
\begin{split}
 \lim_{T \to 0}
    \sqrt{2\pi} \frac{\call(S,T)}{S \sqrt{T}}
     = - 4 \sqrt{2\pi} \lim_{T \to 0}
        \left[
            \sqrt{T}\frac{x_T}{x} p(y;4,x)
        \right]_{K=S}
     = \sqrt{\frac{\alpha}{S}}.
\end{split}
\end{equation}
And this proves the small time limit for the at the money case.
\bigskip

\noindent \textbf{Step (iii): Out of the money small time limit.}
Recall that in the MMM the dividend yield $\kappa =0$. So when $S<K$, \eqref{eq:rr-formula-extended} gives
\begin{equation} \label{eq:ootm-limit-1}
    \lim_{T\to 0} \iv(K,T)
        =
    \displaystyle \lim_{T \to 0}
        \dfrac{\abs{\ln(S/K)}
                }{
                    \left\{
                        -2T\ln[\call(K,T)-(S -K\z(T))_+]
                    \right\}^{1/2}
                  },
\end{equation}
provided the limit on the right exists. As $S<K$, we have $(S -K\z(T))_+=0$ for all sufficiently small $T$. Consequently
\begin{equation} \label{eq:ootm-limit-2}
  \lim_{T \to 0}
    \left\{-2T\ln[\call(K,T)-(S -K\z(T))_+]\right\}
    =  \lim_{T \to 0}
    \left\{-2T\ln[\call(K,T)]\right\}
\end{equation}
if the second limit exists. Since $T\call_T \xrightarrow{\;T \to 0\;} 0$ by Lemma \ref{lem:atm-prelim-1}, applying L'Hopital's rule gives
\begin{equation} \label{eq:lim-denom-1}
\begin{split}
    \lim_{T \to 0}
        \{-2T\ln \call\}
        & = -2 \lim_{T \to 0}
            \frac{\ln \call}{T^{-1}}
         = 2 \lim_{T \to 0}
            \frac{\call_T/\call}{T^{-2}} \\
        & = 2\lim_{T \to 0}
            \frac{T^2 \call_T}{\call}
         = 2 \lim_{T \to 0}
            \frac{2T \call_T + T^2 \call_{TT}}{\call_T} \\
        & = 2 \lim_{T \to 0}
            \left(2T + \frac{T^2 \call_{TT}}{\call_T}\right)
         = 2 \lim_{T\to 0}
                \frac{T^2 \call_{TT}}{\call_T},
\end{split}
\end{equation}
provided the last limit exists. We claim that
\begin{equation} \label{eq:lim-T2cTT-cT}
  \lim_{T \to 0}
    \frac{T^2 \call_{TT}}{\call_T}
    = 2\left(\sqrt{S}-\sqrt{K}\right)^2/\alpha.
\end{equation}
If true, this would give us the desired limit
\begin{equation}
  \iv(K,T)
    \xrightarrow{\; T \to 0 \;}
    \frac{\sqrt{\alpha}\ln(S/K)}{2(\sqrt{S}-\sqrt{K})}
    \quad \mbox{ for $S<K$}.
\end{equation}
To show \eqref{eq:lim-T2cTT-cT}, we use \eqref{eq:cT} and
\eqref{eq:cTT} to get
\begin{equation} \label{eq:R0-R4}
\begin{split}
  T^2 \frac{\call_{TT}}{\call_T}
    & = T^2\frac{1}{\call_T/[-2Sx_T p(y;4,x)/x]}\times \frac{\call_{TT}}{-2Sx_T p(y;4,x)/x} \\
    & = T^2 \frac{1}{R_1}
         (R_2 + R_3 + R_4 + R_5),
\end{split}
\end{equation}
where, $x_T/x = -\eta \e^{\eta T}/(\e^{\eta T}-1)$, and
\begin{equation} \label{eq:R1-R5}
\begin{split}
    R_1 & = 1-\frac{rK\e^{-rT}x \tilde{\chi}^2(y;0,x)}{2Sx_Tp(y;4,x)}, \\
    R_2 & = - \frac{\eta}{\e^{\eta T} -1}, \\
    R_3 & = \frac{1}{2}
            \left[
                \frac{p(y;2,x)}{p(y;4,x)}
                - 1
            \right]y_T
            + \frac{1}{2}
            \left[
                \frac{p(y;6,x)}{p(y;4,x)}
                - 1
            \right]x_T, \\
    R_4 & = \frac{r^2 K\e^{-rT}x\tilde{\chi}^2(y,0,x)}{2S x_T
            p(y;4,x)}, \\
    R_5 & = \frac{rK\e^{-rT}(\e^{\eta T}-1)}{2S \eta}
            \left[
                \frac{-p(y;0,x)}{p(y;4,x)}y_T
                + \frac{-p(y;2,x)}{p(y;4,x)}x_T
            \right].
\end{split}
\end{equation}
By \eqref{eq:x-chi-x-T-p}, $R_1 \xrightarrow{\;T \to 0\;} 1$ and
$R_4 \xrightarrow{\;T \to 0\;} 0$. On the other hand, we also have $T^2
R_2 \xrightarrow{\;T \to 0\;} 0$. So it remains to compute $T^2
R_3$ and $T^2 R_5$. By \eqref{eq:c-2} and \cite[Formula 9.7.1]{abramowitz-stegun-70},
\begin{equation} \label{eq:p-d-2-p-d}
\begin{split}
    \frac{p(y;2,x)}{p(y,4,x)}
        & =
            \frac{
                \e^{-(x+y)/2}
                I_{0}\left(\sqrt{xy}\right)
                }{
                (y/x)^{1/2}
                \e^{-(x+y)/2}
                I_{1}\left(\sqrt{xy}\right)
                }
          = \frac{
                (x/y)^{1/2}
                I_{0}\left(\sqrt{xy}\right)
                }{
                I_{1}\left(\sqrt{xy}\right)
                } \\
        & \sim (x/y)^{1/2}
           \xrightarrow{\;T \to 0 \;} \sqrt{S/K}. \\
\end{split}
\end{equation}
Similarly, as $T \to 0$,
\begin{equation} \label{eq:lim-p-d4-pd}
\begin{split}
    \frac{p(y;6,x)}{p(y,4,x)}
         & \sim (y/x)^{1/2}
            \xrightarrow{\;T \to 0 \;} \sqrt{K/S}, \\
    \frac{p(y;0,x)}{p(y,4,x)}
         & \sim (y/x)^{-1}
            \xrightarrow{\;T \to 0 \;} \sqrt{S/K}. \\
\end{split}
\end{equation}
Moreover, since $y_T \sim -4\eta^2 K/[\alpha(\e^{\eta T}-1)^2]$
and $x_T \sim - 4\eta^2 S/[\alpha(\e^{\eta T}-1)^2]$, we get as $T
\to 0$,
\begin{equation}
  T^2 y_T \sim -\frac{4\eta^2 K}{\alpha(\e^{\eta T}-1)^2} T^2
    = - \frac{4\eta^2 K}{\alpha} \left(\frac{T}{\e^{\eta T}-1}\right)
    \sim -\frac{4\eta^2 K}{\alpha} \frac{1}{\eta^2}
    = - \frac{4K}{\alpha},
\end{equation}
and
\begin{equation} \label{eq:T2-xT}
    T^2 x_T \sim - \frac{4\eta^2 S}{\alpha(\e^{\eta T}-1)^2} T^2
        \sim -\frac{4S}{\alpha}.
\end{equation}
Together these limits show that $T^2 R_5 \xrightarrow{\;T \to 0 \;} 0$ and
\begin{equation} \label{eq:T2-R3}
\begin{split}
    T^2 R_3
         \xrightarrow{\;T \to 0 \;}
         & \frac{1}{2}
            \left[
                \sqrt{S/K} -1
            \right] (-4K/\alpha)
            + \frac{1}{2}
                \left[
                    \sqrt{K/S} -1
                \right](-4S/\alpha) \\
        & = 2 \left(S-2\sqrt{SK}+K\right)/\alpha \\
        & = 2 \left(\sqrt{S}-\sqrt{K}\right)^2/\alpha. \\
\end{split}
\end{equation}
Taking into account \eqref{eq:R0-R4} and the limits of $R_1$ and
$T^2R_i$, $i=1, \ldots, 4$, we get $\frac{T^2 \call_{TT}}{\call_T}
\xrightarrow{\;T\to 0\;}2
\left(\sqrt{S}-\sqrt{K}\right)^2/\alpha$. This proves
\eqref{eq:lim-T2cTT-cT}, from which the small time limit for the out of the
money case follows.

\bigskip

\noindent \textbf{Step (iv): In the money small time limit.}
When $S>K$, \eqref{eq:rr-formula-extended} gives
\begin{equation} \label{eq:rr-itm-limit-1}
    \lim_{T\to 0} \iv(K,T)
        =
    \displaystyle \lim_{T \to 0}
        \dfrac{\abs{\ln(S/K)}
                }{
                    \left\{
                        -2T\ln[\call(K,T)-(S -K\z(T))_+]
                    \right\}^{1/2}
                  }.
\end{equation}
Since $S>K$, $(S -K\z(T))_+=S-K\z(T)$ for all sufficiently small $T$. As a result,
\begin{equation} \label{eq:rr-itm-limit-2}
\begin{split}
  \lim_{T \to 0} \iv(K,T)
    & = \lim_{T \to 0} \frac{\abs{\ln(S/K)}}{\sqrt{-2T\ln[\call(K,T)-S + K\z(T)]}} \\
    & = \lim_{T \to 0} \frac{\abs{\ln(S/K)}}{\sqrt{-2T\ln[\p(K,T)]}},
\end{split}
\end{equation}
where the second equality above results from the put-call parity
\eqref{eq:put-call-parity}. Since $\p (K,T) \xrightarrow{\; T \to 0 \;}
(K-S)_+=0$ when $S>K$, and $T \p_T \xrightarrow{\; T \to 0 \;} 0$,  L'Hopital's rule gives
\begin{equation}  \label{eq:lim-T2-pTT-pT}
\begin{split}
    \lim_{T \to 0}
        \{-2T \ln \p \}
        & = -2 \lim_{T \to 0}
            \frac{\ln \p}{T^{-1}}
         = 2 \lim_{T \to 0}
            \frac{\p_T/\p}{T^{-2}} \\
        & = 2\lim_{T \to 0}
            \frac{T^2 \p_T}{\p}
         = 2 \lim_{T \to 0}
            \frac{2T \p_T + T^2 \p_{TT}}{\p_T} \\
        & = 2 \lim_{T \to 0}
            \left(2T + \frac{T^2 \p_{TT}}{\p_T}\right)
         = 2 \lim_{T\to 0}
                \frac{T^2 \p_{TT}}{\p_T},
\end{split}
\end{equation}
provided the last limit exists. By the put-call parity
\eqref{eq:put-call-parity}, we have
\begin{equation} \label{eq:pT-pTT}
  \p_T = \call_T + K \z_T
  \quad \mbox{ and } \quad
    \p_{TT} = \call_{TT} + K \z_{TT},
\end{equation}
where $\call_T$ and $\call_{TT}$ are given by \eqref{eq:cT} and
\eqref{eq:cTT}; and by \eqref{eq:zero-coupond-bond}
\begin{equation} \label{eq:zT-zTT}
\begin{split}
    \z_T
        & = -r\e^{-rT}(1-\e^{-x/2})
            + \frac{x_T}{2} \e^{-rT-x/2}, \\
    \z_{TT}
        & = r^2 \e^{-rT}(1-\e^{-x/2})
            - rx_T \e^{-rT -x/2}
            + \frac{x_{TT}}{2} \e^{-rT-x/2}
            - \frac{x_T^2}{4}\e^{-rT-x/2}.
\end{split}
\end{equation}
We will prove that
\begin{equation} \label{eq:U-V}
\begin{split}
  & U\equiv \frac{\p_T}{-2Sx_T p(y;4,x)/x}
    \xrightarrow{\;T \to 0 \;} 1, \\
  & V: = \frac{T^2\p_{TT}}{-2Sx_T p(y;4,x)/x}
     \xrightarrow{\;T \to 0 \;} 2 \left(\sqrt{S}-\sqrt{K}\right)^2/\alpha. \\
\end{split}
\end{equation}
We will prove the limit of $U$ first. By \eqref{eq:pT-pTT},
\eqref{eq:zT-zTT}, and \eqref{eq:cT}, we can rewrite $U$ as
\begin{equation} \label{eq:U}
    U  = 1 + U_1 + U_2, \\
\end{equation}
where
\begin{equation}
\begin{split}
      U_1 & = \frac{x_T \e^{-rT-x/2}/2}{-2Sx_T p(y;4,x)/x},     \\
      U_2 & = \frac{rK\e^{-rT}\tilde{\chi}^2(y;0,x)-rK\e^{-rT}(1-\e^{-x/2})
                }{-2Sx_T p(y;4,x)/x }. \\
\end{split}
\end{equation}
Simplifying these expressions gives
\begin{equation} \label{eq:U1U2}
\begin{split}
  U_1 & = - \frac{x \e^{-rT-x/2}}{4Sp(y;4,x)}, \\
  U_2 & = rK\e^{-rT}\frac{x}{x_T}
        \left[
            \frac{\e^{-x/2}}{p(y;4,x)}
            - \frac{\chi^2(y;0,x)}{p(y;4,x)}
        \right]. \\
\end{split}
\end{equation}
By the definition of $x,y,$ and $p$, as $T \to 0$,
\begin{equation}
\begin{split}
    \frac{x \e^{-rT-x/2}}{p(y;4,x)}
    & \sim 2(x/y)^{1/2} \times
        \frac{x\e^{-x/2}}{\e^{-(x+y)/2}I_1(\sqrt{xy})}
     = 2(x/y)^{1/2} \times
        \frac{x}{\e^{-y/2}I_1(\sqrt{xy})} \\
    & \sim 2(S/K)^{1/2} \times
        \frac{x(2\pi)^{1/2}(xy)^{1/4}}{\e^{-y/2}\e^{\sqrt{xy}}}
     = 2(S/K)^{1/2} \times
        \frac{x(2\pi)^{1/2}(xy)^{1/4}}{\e^{-y/2}\e^{\sqrt{xy}}} \\
    & = 2(S/K)^{1/2} \times
        \frac{x(2\pi)^{1/2}(xy)^{1/4}}{\exp\left[\sqrt{xy}\left(1-\frac{1}{2}\sqrt{y/x}\right)\right]}. \\
\end{split}
\end{equation}
Since $S>K$, we have $x>y$ and $(1-\sqrt{y/x}/2)>\const>0$ for all
sufficiently small $T$. By Lemma \ref{lem:atm-prelim-1} (1),
$x,y \xrightarrow{\;T \to 0\;} \infty$. This implies that
\begin{equation} \label{eq:xe-sp}
    \frac{x \e^{-rT-x/2}}{2Sp(y;4,x)}
    \xrightarrow{\; T \to 0 \;} 0.
\end{equation}
Consequently, $U_1 \xrightarrow{\; T \to 0\;} 0$. For $U_2$, we have,
by \eqref{eq:xe-sp}, $\e^{-x/2}/p(y;4,x) \xrightarrow{\;T \to
0\;} 0$. Besides, by L'Hopital's rule, \eqref{eq:p-T}, and
\eqref{eq:pde-chiT},
\begin{equation}
\begin{split}
    & \lim_{T \to 0}
        \frac{\chi^2(y;0,x)}{p(y;4,x)} \\
    & = \lim_{T \to 0}
            \frac{\partial_T \chi^2(y;0,x)}{\partial_T
            p(y;4,x)} \\
    & = \lim_{T \to 0}
            \frac{
                p(y;0,x)y_T - p(y;2,x)x_T
                }{
                [p(y;2,x)-p(y;4,x)]y_T/2
                +
                [p(y;6,x)-p(y;4,x)]x_T/2
                } \\
    & = 2 \lim_{T \to 0}
            \frac{
                \frac{p(y;0,x)}{p(y;4,x)}
                    - \frac{p(y;2,x)}{p(y;4,x)}
                        \times \frac{x_T}{y_T}
                }{
                \frac{p(y;2,x)}{p(y;4,x)}
                -1
                +
                \left[
                    \frac{p(y;6,x)}{p(y;4,x)}-1
                \right] \frac{x_T}{y_T}
                }.  \\
\end{split}
\end{equation}
By \eqref{eq:p-d-2-p-d}, \eqref{eq:lim-p-d4-pd}, and the fact that
$x_T/y_T \xrightarrow{\; T \to 0 \;} S/K$, we have
\begin{equation}
\begin{split}
    & \lim_{T \to 0}
        \frac{\chi^2(y;0,x)}{p(y;4,x)} \\
    & = 2
            \frac{
                \sqrt{S/K}
                    - \sqrt{S/K}
                        \times S/K
                }{
                \sqrt{S/K}
                -1
                +
                \left[
                    \sqrt{K/S}-1
                \right] S/K
                }
        = \frac{\sqrt{S/K}(S-K)}{(\sqrt{S} -\sqrt{K})^2}.
\end{split}
\end{equation}
Because $x/x_T = -\varphi/\varphi_T \xrightarrow{\;T \to 0 \;} 0$,
we get from \eqref{eq:U1U2} the limit $U_2 \xrightarrow{\;T \to 0
\;} 0$. Since $U_1$ and $U_2$ both converge to zero, \eqref{eq:U}
implies the convergence $U \xrightarrow{\;T \to 0 \;} 1$.

Now it remains to prove the limit of $V$ in \eqref{eq:U-V}. By
\eqref{eq:pT-pTT}, \eqref{eq:zT-zTT}, and \eqref{eq:cTT}, we can
rewrite $V$ as
\begin{equation}
\begin{split}
    V
        & = \frac{T^2 \call_{TT} + T^2K\z_{TT}}{-2Sx_T
            p(y;4,x)/x}\\
        & = T^2(R_2 + R_3 + R_5)
            + T^2 V_1 + T^2 V_2
\end{split}
\end{equation}
where $R_2$, $R_3$ and $R_5$ are the same as in \eqref{eq:R1-R5} and
\begin{equation}
\begin{split}
    V_1
        & = \frac{x}{-2 x_T} \times
            \frac{
                r^2K\e^{-rT}\tilde{\chi}^2(y,0,x)
                - r^2 K \e^{-rT}(1-\e^{-x/2})
                }{
                    p(y;4,x)
                }, \\
    V_2
        & = \frac{x}{-2Sx_T} \times
            \frac{
                K
                \left(
                    -rx_T + \frac{x_{TT}}{2} - x_T^2
                \right)\e^{-rT - x/2}
                    }{
                    p(y;4,x)
                }.
\end{split}
\end{equation}
Simplifying $V_1$ gives
\begin{equation}
  V_1 = r^2K\e^{-rT} \frac{x}{-2Sx_T}
        \left[
            \frac{\e^{-x/2}}{p(y;4,x)}
            - \frac{\chi^2(y;0,x)}{p(y;4,x)}
        \right].
\end{equation}
Similar to $U_2$ in \eqref{eq:U1U2}, we found that $V_1
\xrightarrow{\; T \to 0 \;} 0$. Next, because
\begin{equation}
    T^2(-rx_T + x_{TT}/2 - x_T^2)
        \xrightarrow{\; T \to 0 \;}  \const(\alpha,\eta,S,K)
    \quad \mbox{ and } \quad
       x/x_T \xrightarrow{\; T \to 0 \;} 0,
\end{equation}
we also get $T^2 V_2 \xrightarrow{\; T \to 0 \;} 0$. Combining
these limits with the limits of $T^2(R_2+R_3 + R_5)$, see
\eqref{eq:R1-R5}--\eqref{eq:T2-R3}, we get
\begin{equation}
  \lim_{T \to 0} V = 2 \left(\sqrt{S}-\sqrt{K}\right)^2/\alpha,
\end{equation}
as is stated in \eqref{eq:pT-pTT}. Combining
\eqref{eq:lim-T2-pTT-pT} with \eqref{eq:pT-pTT} then gives
\begin{equation}
  \lim_{T \to 0} \{-2T \ln \p \}
        = 2 \lim_{T\to 0}
                \frac{T^2 \p_{TT}}{\p_T}
        = 2 \lim_{T\to 0}
                \frac{V}{U}
        = 4 \left(\sqrt{S}-\sqrt{K}\right)^2/\alpha.
\end{equation}
By \eqref{eq:rr-itm-limit-2}, we then have
\begin{equation}
  \lim_{T \to 0} \iv(K,T)
    = \frac{\sqrt{\alpha}\ln(S/K)}{2(\sqrt{S} - \sqrt{K})}, \qquad S,K \in
            (0,\infty), \quad S>K.
\end{equation}
This proves the small time limit for the in the money case and completes
the proof of the theorem.
\end{proof}

\begin{remark} \label{rem:gao-lee-small-time}
After completing our work we learnt of the results of Gao and Lee \cite{gao-lee-2011}.
They \cite[Remark 7.4]{gao-lee-2011} stated that the small time asymptotics of the time-scaled implied volatility $V$, in their notation, were controlled by $k^2/(2L) \sim V^2$; and they argued that this would imply the Roper--Rutkowski formula \cite[Theorem 5.1]{roper-rutkowski-09}. In our notation, their formula becomes $\iv \sim \abs{\ln(K/S)}/\sqrt{-2T\ln C}$, which clearly does not imply the Roper--Rutkowski formula \cite[Theorem 5.1]{roper-rutkowski-09} or our extended version in \eqref{eq:rr-formula-extended}. See also Remark \ref{rem:gao-lee-large-time} regarding their large time result.
\end{remark}

\begin{remark} \label{rem:gather-et-al}
Separately, we also became aware of the paper by Gatheral et al. \cite{gatheral-et-al}. We noted that for the out of the money case the limit in \eqref{eq:ootm-limit-2} could have been computed by using their Theorem A.2. Indeed, for all in, at, and out of the money cases, our small time limit agrees with theirs \cite[(3.2)]{gatheral-et-al}. However, their pricing approach is different from ours; and it is not entirely clear if on a finite interval in $\RR_+$ away from the origin the Yoshida heat kernel expansion can be generalized to diffusions with degenerate coefficients like the CEV or CIR/square-root process, which is what seems to have been suggested in \cite[Remark 3.5, Appendix A]{gatheral-et-al}.
\end{remark}


\section{Proof of the large time limit: Theorem \ref{thm:large-iv}} \label{sec:proof-large-iv}

We shall prove the large time limit in two steps:
\begin{itemize}
\item[(I)] Proof of the convergence and the upper bound $\limsup_{T \to \infty} \iv(K,T) \le \sqrt{2(r+\eta)}$.
\item[(II)] Proof of $\lim_{T \to \infty} \iv(K,T) =v_*\equiv \sqrt{2(3-2\sqrt{2})(r+\eta)}$.
\end{itemize}
By definition and by the properties of the Black--Scholes formula, the implied volatility is bounded below by zero. Consequently, the upper bound in Step (I) implies that a large time limit exists in the interval $[0,\sqrt{2(r+\eta)}]$. In turn, the existence of the limit allows Step (II) to take place, where we shall show that as $T \to \infty$, the implied volatility $\iv$ is bounded below by any $v \in (0,v_*)$ and above by any $v \in (v_*,\sqrt{2(r+\eta)})$.

\subsection{Step (I): The convergence in large time and the upper bound $\sqrt{2(r+\eta)}$}

We shall prove the following proposition.

\begin{proposition} \label{prop:limsup-bound}
Assume the MMM. Then
\begin{equation*} 
  \limsup_{T \to \infty} \iv(K,T)
  \le \sqrt{2(r+\eta)}
\end{equation*}
and the implied volatility $\iv$ converges to a limit in $[0,\sqrt{2(r+\eta)}]$ as $T \to \infty$.
\end{proposition}
The proof of this proposition requires some preliminary results. Let $0< \epsilon \ll 1$ and define
\begin{equation} \label{eq:vhi-definition}
  \vhi = \sqrt{2(r+\eta)}/(1-\epsilon).
\end{equation}
Then we have the following lemmas.
\begin{lemma} \label{lem:iv-upper-bound-ratio-1}
Under the MMM, we have
\begin{equation} \label{eq:proof-chi-I-0}
  \lim_{T \to \infty}
    \frac{S \chi^2(y;4,x)}{K \e^{-rT} \tilde{\chi}^2(y;0,x)} = 0.
\end{equation}
\end{lemma}
\begin{proof}
See Section \ref{sec:proof:lem:iv-upper-bound-ratio-1} of Appendix D.
\end{proof}

\begin{lemma} \label{lem:iv-upper-bound-ratio-2}
Under the MMM, we have for each $0<\epsilon \ll 1$,
\begin{equation} \label{eq:lem:Nd1-I-0}
  \lim_{T \to \infty}
    \frac{\tilde{N}\bigl(d_1(K,T;\vhi)\bigr)}{ \e^{-rT}\tilde{\chi}^2(y;0,x) } = 0.
\end{equation}
\end{lemma}
\begin{proof}
See Section \ref{sec:proof:lem:iv-upper-bound-ratio-2}  of Appendix D.
\end{proof}

\begin{lemma} \label{lem:iv-upper-bound-ratio-3}
Under the MMM, we have for each $0<\epsilon \ll 1$
\begin{equation} \label{eq:prop:Nd2-I}
  \lim_{T \to \infty}
    \frac{\e^{-\hat{r}T}N\bigl(d_2(K,T;\vhi)\bigr)}{\e^{-rT}\tilde{\chi}^2(y;0,x)} =
    0.
\end{equation}
\end{lemma}
\begin{proof}
See Section \ref{sec:proof:lem:iv-upper-bound-ratio-3} of Appendix D.
\end{proof}

We now present the proof of Proposition \ref{prop:limsup-bound}.

\begin{proof}[Proof of Proposition \ref{prop:limsup-bound}]
Given any arbitrarily small $0<\epsilon\ll 1$, we have
\begin{equation}
\begin{split}
    & \quad \cbs(K,T,\vhi) - \call(K,T) \\
        & = S N\bigl(d_1(K,T;\vhi)\bigr)
            - K\e^{-\hat{r}T}N\bigl(d_2(K,T;\vhi)\bigr)
            - S [1-\chi^2(y;4,x)]
            + K\e^{-rT}\tilde{\chi}^2(y;0,x) \\
        & = S \chi^2(y;4,x)
            - S [1-N(d_1(K,T;\vhi))]
            - K\e^{-\hat{r}T}N\bigl(d_2(K,T;\vhi)\bigr)
            + K\e^{-rT}\tilde{\chi}^2(y;0,x)    \\
        & = S \chi^2(y;4,x)
            - S \tilde{N}(d_1(K,T;\vhi))
            - K\e^{-\hat{r}T}N(d_2(K,T;\vhi))
            + K\e^{-rT}\tilde{\chi}^2(y;0,x) \\
        & = K\e^{-rT}\tilde{\chi}^2(y;0,x) \\
        & \quad \times
          \left[
                \frac{S \chi^2(y;4,x)}{K \e^{-rT}\tilde{\chi}^2(y;0,x)}
                - \frac{S \tilde{N}(d_1(K,T;\vhi))}{K \e^{-rT}\tilde{\chi}^2(y;0,x)}
                - \frac{K\e^{-\hat{r}T}N(d_2(K,T;\vhi))}{K \e^{-rT}\tilde{\chi}^2(y;0,x)}
                + 1
           \right].
\end{split}
\end{equation}
By Lemmas \ref{lem:iv-upper-bound-ratio-1}--\ref{lem:iv-upper-bound-ratio-3}, the square bracketed term is strictly positive for
all sufficiently large $T$; and it tends to $1$ as $T$ tends to
infinity. Since $\call(K,T) = \cbs(K,T,\iv(K,T))$ for all $K,T \in (0,\infty)$,
we get
\begin{equation}
    \cbs(K,T;\vhi) - \cbs(K,T;\iv(K,T))
        = \cbs(K,T;\vhi) - \call(K,T) > 0.
\end{equation}
for all $T$ greater than some sufficiently large constant $T_\epsilon$. Now note that $\cbs(K,T;v)$ is strictly increasing in $v$, other things being equal. So for any $0<\epsilon \ll 1$,
\begin{equation}
  \vhi \ge \iv(K,T), \quad  K \in (0,\infty),
\end{equation}
when $T$ is sufficiently large. Since the implied volatility $\iv$ is bounded below by zero, for each fixed $K \in (0,\infty)$, $\iv(K,T)$ can be considered as a bounded infinite sequence in the time interval $[0,\infty]$. This implies that $\iv(K,T)$ has a convergent subsequence in $T$ as $T \to \infty$. Further, as $\epsilon$ can be made arbitrarily small, we must have
\begin{equation} \label{eq:proof:limsup-bound}
  \limsup_{T \to \infty} \iv(K,T)
  \le \lim_{\epsilon \to 0} \vhi
  = \sqrt{2(r+\eta)}.
\end{equation}
Consequently, $\iv$ has a large time limit in $[0,\sqrt{2(r+\eta)}]$.
\end{proof}
\begin{remark}
In the proof of the large time limit in Section \ref{subsec:main-proof-large-time-limit} below,  we will show that in fact
\begin{equation}
  \limsup_{T \to \infty} \iv(K,T)
  < \sqrt{2(r+\eta)}.
\end{equation}
Our proof will utilize series expansions of the noncentral chi-square distributions.
\end{remark}

\subsection{Step (II): proof of the large time limit} \label{subsec:main-proof-large-time-limit}

We shall present some preliminary lemmas first. Note that $S$, $\alpha$ and $\eta$ are fixed parameters.

\begin{lemma} \label{lem:cbs-call-R}
Assume the MMM and let $v \in [0,\infty]$. Then for any strike $K$ and expiry $T$, the generic Black--Scholes call price with volatility $v$ and the MMM call price can be written respectively as
\begin{equation} \label{eq:cbs-call-R}
\left\{
\begin{split}
  \cbs(K,T;v)
    & = S - S \rcbs(K,T;v), \\
  \call(K,T)
    & = S - S \rcall(K,T), \\
\end{split}
\right.
\end{equation}
where
\begin{equation} \label{eq:rcbs-rcall-def}
\left\{
\begin{split}
  \rcbs(K,T;v)
    & = \tilde{N}(d_1(K,T;v))
        + \frac{K}{S} \e^{-\hat{r} T}
        - \frac{K}{S} \e^{-\hat{r} T} \tilde{N}(d_2(K,T;v)), \\
  \rcall(K,T)
    & = \frac{K}{S} \e^{-r T}
        + \chi^2(y;4,x)
        - \frac{K}{S} \e^{-r T} \chi^2(y;0,x).
\end{split}
\right.
\end{equation}
\end{lemma}
\begin{proof}
See Section \ref{subsec:proof:lem:cbs-call-R} of Appendix D.
\end{proof}

We need the following lemma for the chi-square distribution.
\begin{lemma} \label{lem:chisq-bounds}
Under the MMM the follow inequalities are valid the noncentral $\chi^2$ functions for all $T \in (0,\infty)$:
\begin{equation} \label{eq:lem:chisq-bounds-1}
\left\{
\begin{split}
    & \e^{-(x+y)/2}
      \frac{y^2}{4}
      \sum_{j=0}^\infty \frac{(xy/4)^j}{j!(2+j)!}
      \le \chi^2(y;4,x)
      \le \e^{-x/2}
      \frac{y^2}{4}
      \sum_{j=0}^\infty \frac{(xy/4)^j}{j!(2+j)!} \\
    & \e^{-(x+y)/2}
      \sum_{j=1}^\infty \frac{(xy/4)^j}{(j!)^2}
      + \e^{-x/2}
      \le \chi^2(y;0,x)
      \le \e^{-x/2}
      \sum_{j=1}^\infty \frac{(xy/4)^j}{(j!)^2}
      + \e^{-x/2}.
\end{split}
\right.
\end{equation}
Moreover, for every fixed $K$ (and $S, \alpha, \eta$) there exists a positive constant $T_1=T_1(K;S,\alpha,\eta)<\infty$ such that for all $T>T_1$,
\begin{equation} \label{eq:lem:chisq-bounds-2}
\left\{
\begin{split}
    & \e^{-(x+y)/2}
            \frac{y^2}{4}
            \left[
                \frac{1}{2}
                + \frac{xy}{24}
                +  \frac{(xy)^2}{768}
            \right]
      \le \chi^2(y;4,x)
      \le \e^{-x/2}
            \frac{y^2}{4}
            \left[
                \frac{1}{2}
                + \frac{xy}{24}
                +  (xy)^2
            \right], \\
    & \e^{-x/2}
            + \e^{-(x+y)/2}
                \left[
                    \frac{xy}{4}
                    + \frac{(xy)^2}{64}
                \right]
      \le \chi^2(y;0,x)
      \le \e^{-x/2}
            + \e^{-x/2}
                \left[
                    \frac{xy}{4}
                    + (xy)^2
                \right]. \\
\end{split}
\right.
\end{equation}
\end{lemma}
\begin{proof}
See Section \ref{subsec:proof:lem:chisq-bounds} of Appendix D.
\end{proof}

Let
\begin{equation}
    \rcalll(K,T)
         = \e^{(r+\eta)T-(x+y)/2}
            \frac{y^2}{4}
            \left[
                \frac{1}{2}
                + \frac{xy}{24}
                +  \frac{(xy)^2}{768}
            \right]
         -\frac{K}{S} \e^{\eta T-x/2}
            \left[
                \frac{xy}{4}
                + (xy)^2
            \right]
\end{equation}
\begin{equation}
    \rcallu(K,T)
        = \e^{(r+\eta)T-x/2}
            \frac{y^2}{4}
            \left[
                \frac{1}{2}
                + \frac{xy}{24}
                +  (xy)^2
            \right]
         -\frac{K}{S} \e^{\eta T-(x+y)/2}
            \left[
                \frac{xy}{4}
                + \frac{(xy)^2}{64}
            \right]
\end{equation}
Then we have the following lemmas

\begin{lemma} \label{lem:R-asymptotics}
Assume the MMM and let $v\in (0,\sqrt{2(r+\eta)})$. Then for sufficiently large $T$,
\begin{equation} \label{eq:lem:R-asymptotics-1}
    \frac{K}{S}e^{-[\hat{r}-(r+\eta)]T} + \rcalll(K,T)
    \le \e^{(r+\eta)T}\rcall(K,T)
    \le \frac{K}{S}e^{-[\hat{r}-(r+\eta)]T} + \rcallu(K,T).
\end{equation}
Moreover, as $T \to \infty$,
\begin{equation} \label{eq:lem:R-asymptotics-2}
  \rcalll(K,T) = -\frac{2\eta^2 K^2}{\alpha^2} \e^{-(r+\eta)T}
    -\frac{760 SK^3\eta^4}{3\alpha^4} \e^{-2(r+\eta)T-\eta T}
    + \O(\e^{-3(r+\eta)T-2\eta T}),
\end{equation}
\begin{equation}  \label{eq:lem:R-asymptotics-3}
  \rcallu(K,T)
  = -\frac{2\eta^2 K^2}{\alpha^2} \e^{-(r+\eta)T}
            -\frac{4 SK^3\eta^4}{3\alpha^4} \e^{-2(r+\eta)T-\eta T}
            + \O(\e^{-3(r+\eta)T-2\eta T}),
\end{equation}
and
\begin{equation} \label{eq:lem:R-asymptotics-4}
\begin{split}
    & \e^{(r+\eta)T}\rcbs(K,T;v) \\
    &     = \frac{K}{S} \e^{-[\hat{r}-(r+\eta)]T}
            - \frac{2K\eta}{\alpha}  \frac{v\sqrt{T}}{(d_2+v\sqrt{T})d_2\sqrt{2\pi}} \e^{-d_2^2/2}
            + \frac{2K\eta}{\alpha} n(d_2) \O(d_2^{-3}), \\
\end{split}
\end{equation}
where $d_2=d_2(K,T;v)$.
\end{lemma}
\begin{proof}
See Section \ref{subsec:proof:lem:R-asymptotics} of Appendix D.
\end{proof}


\begin{lemma} \label{lem:d2-v}
Assume the MMM and let $v_* = \sqrt{2(3-2\sqrt{2})(r+\eta)}$. Then for all sufficiently large $T$,
\begin{equation} \label{eq:d2-v}
\left\{
  \begin{array}{ll}
    \frac{1}{2}d_2^2(K,T;v) - (r+\eta)T > c_1 T, & \hbox{if $v \in (0,v_*)$,} \\
    \frac{1}{2}d_2^2(K,T;v) - (r+\eta)T < - c_2 T, & \hbox{if $v \in (v_*, \sqrt{2(r+\eta)})$,}
  \end{array}
\right.
\end{equation}
where $c_1$ and $c_2$ are some strictly positive constants dependent only on $K,S,v,r,\eta$.
\end{lemma}
\begin{proof}
See Section \ref{subsec:proof:lem:d2-v} of Appendix D.
\end{proof}
\begin{remark}
By virtue of Proposition \ref{prop:limsup-bound}, we only need to investigate the limiting behavior of the implied volatility (and the limiting behavior of $d_2$) in the interval $(0,\sqrt{2(r+\eta)})$.
\end{remark}

We now present the proof of the large time limit of the implied volatility in the MMM.
\begin{proof}[Proof of Theorem \ref{thm:large-iv}]
Lemma \ref{lem:cbs-call-R} shows that for any $v \in (0,\sqrt{2(r+\eta)})$,
\begin{equation} \label{eq:R-definition}
\left\{
\begin{split}
  \cbs(K,T;v)
    & = S - S \rcbs(K,T;v), \\
  \call(K,T)
    & = S - S \rcall(K,T), \\
\end{split}
\right.
\end{equation}
where $\rcbs$ and $\rcall$ are given by \eqref{eq:rcbs-rcall-def}. Next, Lemma \ref{lem:d2-v} implies that for large enough $T$,
\begin{equation}
\left\{
  \begin{array}{ll}
    \exp\left(-\frac{1}{2}d_2^2(K,T;v)\right) > \exp(-(r+\eta)T - c_1 T), & \hbox{if $v \in (0,v_*)$;} \\
    \exp\left(-\frac{1}{2}d_2^2(K,T;v)\right) < \exp(-(r+\eta)T + c_2 T), & \hbox{if $v \in (v_*, \sqrt{2(r+\eta)})$.}
  \end{array}
\right.
\end{equation}
where $c_1$ and $c_2$ are some positive constants and $v_* =\sqrt{2(3-2\sqrt{2})(r+\eta)}$. Then combining these inequalities with the results of Lemma \ref{lem:R-asymptotics} gives the following properties:
\begin{enumerate}
\item[(i)] if $v \in (0,v_*)$, then for sufficiently large $T$,
\begin{equation}
\begin{split}
  \e^{(r+\eta)T} \rcbs(K,T;v)
    & \ge \frac{K}{S} \e^{-[\hat{r}-(r+\eta)]T} + \rcallu(K,T) \\
    & \ge \e^{(r+\eta)T}\rcall(K,T);
\end{split}
\end{equation}
\item[(ii)] if $v \in (v_*, \sqrt{2(r+\eta)})$, then for sufficiently large $T$,
\begin{equation}
\begin{split}
    \e^{(r+\eta)T} \rcbs(K,T;v)
    & \le \frac{K}{S} \e^{-[\hat{r}-(r+\eta)]T} + \rcalll(K,T) \\
    & \le \e^{(r+\eta)T}\rcall(K,T).
\end{split}
\end{equation}
\end{enumerate}
By \eqref{eq:R-definition} and the equality that $\call(K,T) = \cbs(K,T;\iv(K,T))$, we get
\begin{equation}
\left\{
\begin{split}
    & \cbs(K,T;v) \le \call(K,T) = \cbs(K,T;\iv(K,T)),
        \quad \mbox{if $v \in (0,v_*)$;} \\
    & \cbs(K,T;v) \ge \call(K,T) = \cbs(K,T;\iv(K,T)),
        \quad \mbox{if $v \in (v_*, \sqrt{2(r+\eta)})$.} \\
\end{split}
\right.
\end{equation}
As the function $v\mapsto \cbs(K,T;v)$ is monotonically increasing in $v$, we have, for each $K$,
\begin{equation}
\left\{
  \begin{array}{ll}
    \iv(K,T) \ge v, & \hbox{if $v \in (0,v_*)$;} \\
    \iv(K,T) \le v, & \hbox{if $v \in (v_*, \sqrt{2(r+\eta)})$,}
  \end{array}
\right.
\end{equation}
for all sufficiently large $T$. This implies that
\begin{equation}
    \liminf_{T \to \infty} \iv(K,T) \ge v_*
    \quad \mbox{ and } \quad
    \limsup_{T \to \infty} \iv(K,T) \le v_*.
\end{equation}
As a result,
\begin{equation}
  \lim_{T \to \infty} \iv(K,T) = v_* = \sqrt{2(3-2\sqrt{2})(r+\eta)}.
\end{equation}
And the proof is complete.
\end{proof}

\begin{remark} \label{rem:gao-lee-large-time}
The results of Gao and Lee \cite[Corollary 7.8]{gao-lee-2011} do not apply here because the assumptions of their corollary are not satisfied in the MMM. Specifically, for their Case $(+)$ they require, in their notation, $k/L \xrightarrow{\;T \to
\infty\;} \const \in [0,\infty)$; but under the MMM $k/L \xrightarrow{\;T \to
\infty \;} -1$. Similarly, for their Case $(-)$ they need $k/L \xrightarrow{\;T \to
\infty \;} \const \in (0,\infty)$; but under the MMM $k/L \xrightarrow{\;T \to
\infty \;} \infty$.
\end{remark}

\section{Appendix A: properties of the noncentral chi-square distribution} \label{sec:appendix-a}

We list some facts about the noncentral chi-square distribution. Let use make the following remark first.
\begin{remark}
Let $Y_x$ be a chi-square random variable with zero degree of freedom and noncentrality parameter $x$. That is, let $Y_x \sim \chi^2_0(x)$. Then
\begin{equation} \label{eq:chisq-zero-degree-1}
\begin{split}
  \chi^2(y;0,x)
    & = \PP(Y_x \le y)  \qquad (y>0) \\
    & = \PP(Y_x = 0) + \PP(0< Y_x \le y) \\
    & = \e^{-x/2} + \int_0^y p(z;0,x) \, \d z
\end{split}
\end{equation}
because $p(\cdot;0,x)$ is an improper density and
\begin{equation}
  \int_0^\infty p(z;0,x) \,\d z = 1 - \e^{-x/2}.
\end{equation}
See Siegel \cite{siegel-79}. Hence
\begin{equation}
\begin{split}
  \tilde{\chi}^2(y;0,x)
    & = \PP(Y_x > y)  \qquad (y>0) \\
    & = \int_y^\infty p(z;0,x) \, \d z \\
    & = 1- \PP(Y_x \le y) \\
    & = 1 - \PP(Y_x = 0) - \PP(0< Y_x \le y) \\
    & = 1 - \e^{-x/2} - \int_0^y p(z;0,x) \, \d z \\
    & = 1 - \chi^2(y;0,x).
\end{split}
\end{equation}
\end{remark}

\begin{lemma} For the noncentral chi-square distribution $\chi^2$ the
following identities is valid for $x,y>0$ and $\delta \ge 0$:
\begin{eqnarray}
    \frac{\partial p(y;\delta,x)}{\partial x}
        & = & \frac{1}{2} p(y;\delta+2,x)
            - \frac{1}{2} p(y;\delta,x), \label{eq:chisq-property-1} \\
    \frac{\partial p(y;\delta,x)}{\partial x}
        & = & - \frac{\partial p(y;\delta+2,x)}{\partial y}, \label{eq:chisq-property-2} \\
    \frac{\partial \chi^2(y;\delta,x)}{\partial x}
        & = & - p(y;\delta+2,x). \label{eq:chisq-property-3}
\end{eqnarray}
\end{lemma}
\begin{proof}
The first identity is given by Equation (2) of Cohen
\cite{cohen-88}; the second identity results from Equations (2)
and (3) of Cohen \cite{cohen-88};  the third identity is given by
Johnson et al. \cite[(29.23e)]{johnson-et-al-95}.
\end{proof}
\begin{corollary} \label{cor:chisq-density-derivative}
Subject to the MMM and the condition that $\delta \ge 0$,
\begin{multline} \label{eq:p-T}
    \frac{\partial p(y;\delta,x)}{\partial T}
        = \frac{1}{2}
        \left[
            p(y;\delta-2,x)
            - p(y;\delta,x)
        \right]y_T
            + \frac{1}{2}\left[
                    p(y;\delta+2,x)
                    - p(y;\delta,x)
                \right]x_T
\end{multline}
and
\begin{equation} \label{eq:tilde-chisq-T}
    \frac{\partial \tilde{\chi}^2(y;\delta,x)}{\partial T}
        =  - p(y;\delta,x)y_T + p(y;\delta+2,x)x_T.
\end{equation}
\end{corollary}
\begin{proof}
The chain rule gives
\begin{equation} \label{eq:p-T-proof-1}
    \frac{\partial p(y;\delta,x)}{\partial T}
    =  \frac{\partial p(y;\delta,x)}{\partial y}
            \frac{\partial y}{\partial T}
            + \frac{\partial p(y;\delta,x)}{\partial x}
            \frac{\partial x}{\partial T}.
\end{equation}
By \eqref{eq:chisq-property-1} and \eqref{eq:chisq-property-2},
\begin{equation} \label{eq:p-T-proof-2}
\begin{split}
    \frac{\partial p(y;\delta,x)}{\partial x}
        & = \frac{1}{2} p(y;\delta+2,x)
            - \frac{1}{2} p(y;\delta,x), \\
    \frac{\partial p(y;\delta,x)}{\partial y}
        & = - \frac{\partial p(y;\delta-2,x)}{\partial x}
        = \frac{1}{2} p(y;\delta - 2,x)
            - \frac{1}{2} p(y;\delta,x).
\end{split}
\end{equation}
Combining \eqref{eq:p-T-proof-1} and \eqref{eq:p-T-proof-2} gives
\eqref{eq:p-T}. To prove \eqref{eq:tilde-chisq-T}, we apply the
chain rule and the Leibniz formula for differentiation of definite
integrals to $\chi^2(y;\delta,x)$. That gives
\begin{equation} \label{eq:pde-chiT}
\begin{split}
    \frac{\partial \chi^2(y;\delta,x)}{\partial T}
        & = \frac{\partial \chi^2(y;\delta,x)}{\partial y}
            \frac{\partial y}{\partial T}
            + \frac{\partial \chi^2(y;\delta,x)}{\partial x}
            \frac{\partial x}{\partial T} \\
        & = p(y;\delta,x)y_T + x_T
            \int_0^y \frac{\partial p(z;\delta,x)}{\partial x}
            \, \d z \\
        & = p(y;\delta,x)y_T - x_T
            \int_0^y \frac{\partial p(z;\delta+2,x)}{\partial z}
            \, \d z \\
        & = p(y;\delta,x)y_T - p(y;\delta+2,x)x_T,
\end{split}
\end{equation}
where the third equality results from \eqref{eq:chisq-property-2}.
Then the desired result follows from the identity $\tilde{\chi}^2
= 1- \chi^2$.
\end{proof}

\begin{lemma}
Assume the MMM. Then
\begin{eqnarray}
  p(y;4,x)
    &=& \frac{y}{x} p(y;0,x), \label{eq:pde-pp1}  \\
  - \frac{2}{x} p(y;4,x)
    &=& p(y;6,x) - \frac{y}{x} p(y;2,x)  .    \label{eq:pde-pp2}
\end{eqnarray}
\end{lemma}
\begin{proof}
By \eqref{eq:c-2}, the definition of the density $p$,
we get
\begin{equation}
    p(y;4,x) - \frac{y}{x} p(y;0,x)
        = \frac{1}{2}
            (y/x)^{1/2}
            \e^{-(x+y)/2}
            \left[
                I_{1}\left(\sqrt{xy}\right)
                -
                I_{-1}\left(\sqrt{xy}\right)
            \right].
\end{equation}
Then we obtain \eqref{eq:pde-pp1} by noting that $I_1(\cdot) =
I_{-1}(\cdot)$; see, e.g., \cite[Formula
9.6.6]{abramowitz-stegun-70}. Proceeding similarly with the
definitions of the terms, we have
\begin{equation}
\begin{split}
    \quad p(y;6,x) - \frac{y}{x} p(y;2,x)
    & = \frac{1}{2}
            \left(y/x\right)
            \e^{-(x+y)/2}
            \left[
                I_{2}\left(\sqrt{xy}\right)
                -
                I_{0}\left(\sqrt{xy}\right)
            \right] \\
    & = - \frac{1}{2}
            \left(y/x\right)
            \e^{-(x+y)/2}
            \frac{2}{\sqrt{xy}}
                I_{1}\left(\sqrt{xy}\right) \\
    & = - \frac{2}{x}p(y;4,x),
\end{split}
\end{equation}
where the second equality follows from the identity $I_{1-1}(z)
-I_{1+1}(z) = 2I_1(z)/z$; see e.g. \cite[Formula
9.6.26]{abramowitz-stegun-70}. This proves \eqref{eq:pde-pp2}, and
the proof is thus complete.
\end{proof}

\section{Appendix B: auxiliary results for the small time limit} \label{sec:appendix-b}

\subsection{Differential properties of the call option price}

We now list some basic properties of the call option price. Recall
that we have set $t=0$.

\begin{lemma}
Under the MMM we have the following representations:
\begin{equation} \label{eq:formula-x-y-a}
\left\{
\begin{split}
    &   \varphi
         = \frac{\alpha}{4\eta}\left(\e^{\eta T} -1\right), \quad
        \varphi_T
         = \frac{\alpha}{4} \e^{\eta T}, \quad
        \varphi_{TT}
         = \frac{\alpha\eta}{4} \e^{\eta T}, \\
    &   x
         = \frac{S}{\varphi}, \quad
        x_T
         = - S\frac{\varphi_T}{\varphi^2}, \quad
        x_{TT}
         = 2S \frac{\varphi_T^2}{\varphi^3}
            - S \frac{\varphi_{TT}}{\varphi^2}, \\
    &   y
         = \frac{K\e^{-rT}}{\varphi}, \quad
        y_T
         = -\frac{rK\e^{-rT}}{\varphi}
            - \frac{K\e^{-rT}\varphi_T}{\varphi^2}.
\end{split}
\right.
\end{equation}
Moreover,
\begin{equation}
\begin{split}
    x_{TTT}
        & = - 6 S \varphi^{-4} \varphi_T^3
            + 4 S \varphi^{-3} \varphi_T \varphi_{TT}
            + 2 S \varphi^{-3} \varphi_T \varphi_{TT}
            - S \varphi^{-2} \varphi_{TTT} \\
        & = - 6 S \varphi^{-4} \varphi_T^3
            + 6 S \varphi^{-3} \varphi_T \varphi_{TT}
            - S \varphi^{-2} \varphi_{TTT}, \\
\end{split}
\end{equation}
and
\begin{equation} \label{eq:formula-x-y-b}
    xy = \frac{S K\e^{-rT}}{\varphi^2}, \quad
    \frac{y}{x} = \frac{K\e^{-rT}}{S}, \quad
    \frac{y_T}{x_T}
        = \frac{K}{S \e^{rT}}
            \left(\frac{r + \varphi_T/\varphi}{\varphi_T/\varphi}\right).
\end{equation}
\end{lemma}
\begin{proof}
Direct differentiation.
\end{proof}

\begin{lemma} \label{lem:cT-cTT}
Under the MMM, we have
\begin{equation} \label{eq:cT}
  \call_T(K,T)
    = - \frac{2 S}{x} x_T p(y;4,x)
        + r K \e^{-rT} \tilde{\chi}^2(y;0,x),
\end{equation}
and
\begin{equation} \label{eq:cTT}
\begin{split}
        & \quad    \call_{TT}(K,T) \\
        & = - 2S \frac{\eta^2 \e^{\eta T}}{(\e^{\eta T}-1)^2}
            p(y;4,x) \\
        & \quad
            + S \frac{\eta \e^{\eta T}}{\e^{\eta T}-1}
            \Biggl\{
                \left[
                    p(y;2,x)
                    - p(y;4,x)
                \right]y_T 
            + \left[
                    p(y;6,x)
                    - p(y;4,x)
                \right]x_T
                \Biggr\} \\
        & \quad
            - r^2 K \e^{-r T} \tilde{\chi}^2(y;0,x)
            + rK\e^{-rT}
                \left[
                    - p(y;0,x)y_T
                    + p(y;2,x)x_T
                \right].
\end{split}
\end{equation}
for $S,K,T \in (0,\infty)$.
\end{lemma}
\begin{proof}
Differentiating \eqref{eq:c-1} with respect to $T$ gives
\begin{equation}
\begin{split}
    \call_T (K,T)
        & = S \frac{\partial \tilde{\chi}^2(y;4,x)}{\partial  T} 
            + rK\e^{-rT}\tilde{\chi}^2(y;0,x)
            - K\e^{-rT} \frac{\partial \tilde{\chi}^2(y;0,x)}{\partial T}. \\
\end{split}
\end{equation}
Applying \eqref{eq:tilde-chisq-T} to the chi-square terms gives
\begin{equation}
\begin{split}
    \call_T
        & = S [-p(y;4,x)y_T + p(y;6,x)x_T]
            + rK\e^{-rT}\tilde{\chi}^2(y;0,x) \\
        & \quad
            - K\e^{-rT} [-p(y;0,x)y_T + p(y;2,x)x_T] \\
        & = - S \left[p(y;4,x) - \frac{y}{x} p(y;0,x)\right] y_T \\
        & \quad
            + S \left[p(y;6,x) - \frac{y}{x} p(y;2,x)\right] x_T
            + rK\e^{-rT}\tilde{\chi}^2(y;0,x) \\
        & = - 2 S \frac{x_T}{x}p(y;4,x) + rK\e^{-rT}\tilde{\chi}^2(y;0,x), \\
\end{split}
\end{equation}
where the last equality results from \eqref{eq:pde-pp1} and
\eqref{eq:pde-pp2}. This proves \eqref{eq:cT}, the identity for
$\call_T$.
To get the identity for $\call_{TT}$, we differentiate
$\call_T$, which gives
\begin{equation}
\begin{split}
  \call_{TT}(K,T)
    & = -2S \frac{x_{TT} x - x_T^2}{x^2}
        p(y;4,x)
         -2S \frac{x_T}{x}
          \frac{\partial p(y;4,x)}{\partial T} \\
    & \quad
        - r^2 K \e^{-r T} \tilde{\chi}^2(y;0,x)
        + rK\e^{-rT} \frac{\partial \tilde{\chi}^2(y;0,x)}{\partial
        T}.
\end{split}
\end{equation}
We then obtain identity \eqref{eq:cTT} by applying Corollary
\ref{cor:chisq-density-derivative} to $\partial p(y;4,x)/\partial T$
and $\partial \tilde{\chi}^2(y;0,x)/\partial T$ and by
noting the identity
\begin{equation}
  \frac{x_{TT}x -x_T^2}{x^2}
  = \frac{\eta^2 \e^{\eta T}}{\left(\e^{\eta T}-1\right)^2}.
\end{equation}
The proof is thus complete.
\end{proof}

\subsection{Some small time limits}

In this subsection we present some small time limits needed for the
proof of Theorem \ref{thm:small-iv}.

\begin{lemma} \label{lem:atm-prelim-1}
Under the MMM, the following properties hold for $S,K \in
(0,\infty)$:
\begin{enumerate}
\item $x\xrightarrow{\;T \to 0\;} \infty$; $y\xrightarrow{\;T
        \to 0\;} \infty$.
\item $\frac{x_T\sqrt{T}}{x} p(y;4,x)
        \xrightarrow{\;T \to 0\;} - \frac{1}{4\sqrt{2\pi}}
        \sqrt{\frac{\alpha}{K}}$ when $S=K$.
\item $T\frac{x_T}{x} p(y;4,x) \xrightarrow{\;T \to
0\;} 0$. \\
\item $T\call_T(K,T) \xrightarrow{\;T \to 0\;} 0$. \\
\end{enumerate}
\end{lemma}
\begin{proof}
The limits in (1) follow from the definitions of $x$ and $y$. We
now prove (2). Note that $x = s/\varphi$ and $x_T = - S
\varphi_T/\varphi^2$.
From these identities we get
\begin{equation}
  \frac{x_T}{x}
    = - \frac{\varphi_T}{\varphi}
    = - \frac{\eta \e^{\eta T}}{\e^{\eta T} -1}
    \sim -\frac{\eta}{\e^{\eta T} -1}
    \qquad (T \to 0).
\end{equation}
Hence, as $T \to 0$,
\begin{equation} \label{eq:lem:atm-prelim-1:1}
\begin{split}
   \left. \sqrt{T} \frac{x_T}{x} p(y;4, x)\right|_{S=K}
    & \sim - \frac{\eta \sqrt{T}}{\e^{\eta T} -1} p(y;4,x) \\
    & = \frac{-\eta \sqrt{T}}{\e^{\eta T} -1}
            \left(y/x\right)^{1/2}
            \e^{-(x+y)/2}
            I_1\left(\sqrt{xy}\right)/2 \\
    & \sim \frac{-\eta \sqrt{T}}{\e^{\eta T} -1}
        \left(\frac{1}{2}\right)
            \exp\left(- \frac{x+y}{2}\right)
            \exp(\sqrt{xy})
            \frac{1}{\sqrt{2\pi \sqrt{xy}}} \\
    & = \frac{-\eta \sqrt{T}}{\e^{\eta T} -1}
        \left(\frac{1}{2}\right)
             \exp\left(- \frac{1}{2}\left(\sqrt{x}-\sqrt{y}\right)^2\right)
            \frac{1}{\sqrt{2\pi}(xy)^{1/4}} \\
    & \sim \frac{-\eta \sqrt{T}}{\e^{\eta T} -1} \times
            \frac{1}{2\sqrt{2\pi}(xy)^{1/4}}  \\
    & = \frac{-\eta}{2\sqrt{2\pi}}
        \frac{\sqrt{T}}{(\e^{\eta T}-1)(xy)^{1/4}}.
\end{split}
\end{equation}
In the third line of \eqref{eq:lem:atm-prelim-1:1} we have used
the fact that when $S=K$, $y/x \sim 1$ and that $I_v(z) \sim
\e^z/\sqrt{2\pi z}$ as $z \to \infty$; see \cite[Formula
9.7.1]{abramowitz-stegun-70}. In the forth line we have used the
property that $(\sqrt{x}-\sqrt{y}) \to 0$ as $T \to 0$ when $S=K$.

By the definitions of $x$ and $y$, we have
\begin{equation} \label{eq:lem:atm-prelim-1:2}
  xy|_{S =K}
    = \frac{sK \e^{-r T}}{\varphi^2}
    =\frac{K^2 \e^{-r T}}{\varphi^2}
    \sim \frac{K^2}{\varphi^2}
    = \frac{16 \eta^2 K^2}{\alpha^2 (\e^{\eta T}-1)^2}.
\end{equation}
Combining \eqref{eq:lem:atm-prelim-1:1} and
\eqref{eq:lem:atm-prelim-1:2} shows that as $T \to 0$,
\begin{equation}
\begin{split}
  \sqrt{T}\frac{x_T}{x} p(y;4,x)
    & \sim \frac{-\eta}{2\sqrt{2\pi}}
        \frac{\sqrt{T}}{(\e^{\eta T}-1)(xy)^{1/4}} \\
    & \sim \frac{-\eta}{2\sqrt{2\pi}}
        \frac{\sqrt{T}
            }{
                (\e^{\eta T}-1)
                \left[
                    \dfrac{16 \eta^2 K^2}{\alpha^2 (\e^{\eta T}-1)^2}
                \right]^{1/4}
            } \\
    & = \frac{-\eta}{2\sqrt{2\pi}}
        \left(\frac{T}{\e^{\eta T} -1}\right)^{1/2}
        \frac{\sqrt{\alpha}}{2\sqrt{\eta K}} \\
    & \sim \frac{-\sqrt{\alpha \eta}}{4\sqrt{2\pi K}}
            \frac{1}{\sqrt{\eta}} \\
    & = \frac{-\sqrt{\alpha }}{4\sqrt{2\pi K}}.
\end{split}
\end{equation}
This proves (2). The limit in (3) follows from those in (2).
The limits in (4) follows from the definitions of $x$,
$y$, $p$ and $\call_T$.
And the proof of the lemma is thus complete.
\end{proof}

\begin{lemma} \label{lem:chisq-limits}
Under the MMM, we have
\begin{eqnarray}
    \lim_{T \to t} \chi^2(y;4,x)
        & = & \lim_{T \to t} \chi^2(y;0,x) = 0,
        \quad \mbox{if $S>K$}; \label{eq:chisq-limits-1} \\
    \lim_{T \to t} \chi^2(y;4,x)
        & = & \lim_{T \to t} \chi^2(y;0,x) = 1,
        \quad \mbox{if $S<K$}. \label{eq:chisq-limits-2}
\end{eqnarray}
\end{lemma}
\begin{proof}
Without loss of generality we shall prove the lemma for $t=0$. We
will check $\chi^2(y;4,x)$ first. From Temme
\cite[(2.6)-(2.8)]{temme-93}, we have, for sufficiently small $T$,
\begin{equation} \label{eq:temme-93}
  \chi^2(y;4,x)
  =
  \left\{%
    \begin{array}{ll}
        \displaystyle \frac{1}{2}(y/x)
                \left[
                    \sqrt{x/y}F_{2}(\sqrt{xy},\omega)
                    - F_{1}(\sqrt{xy},\omega)
                \right], & \hbox{for $S>K$}, \\
        1 + \displaystyle \frac{1}{2}(y/x)
                \left[
                    \sqrt{x/y}F_{2}(\sqrt{xy},\omega)
                    - F_{1}(\sqrt{xy},\omega)
                \right], & \hbox{for $S<K$}, \\
    \end{array}%
  \right.
\end{equation}
where
\begin{equation}
\begin{split}
  & \omega = \frac{1}{2}(\sqrt{y}-\sqrt{x})^2/\sqrt{yx}, \\
  & F_\mu(\sqrt{xy},\omega)
        = \int_{\sqrt{xy}}^\infty \e^{-(\omega +1)z}I_{\mu}(z)
        \,\d z.
\end{split}
\end{equation}
Since $(y/x) \xrightarrow{\; T \to 0 \;} K/S$, we only need to
show that in \eqref{eq:temme-93} both $F_{2}$ and
$F_{1}$ converge to zero as $T$ tends to zero. Mattner and Roos \cite[(8)]{mattner-roos-07}
have shown that
\begin{equation} \label{eq:mattner-roos-inequality}
  \e^{-z}(I_0(z) + I_1(z))
    < \sqrt{2/(\pi z)},
    \quad z \in (0,\infty).
\end{equation}
This implies that
\begin{equation}
\begin{split}
    F_1(\sqrt{xy},\omega)
    \le \sqrt{2/\pi}
        \int_{\sqrt{xy}}^\infty \frac{\e^{-\omega z}}{\sqrt{z}}
        \,\d z
    = \frac{
        \sqrt{2} \erfc\left(\sqrt{w
            \sqrt{xy}}\right)
            }{
            \sqrt{w}},
\end{split}
\end{equation}
where $\erfc(\cdot)$ is the complementary error function.
Observing that $\omega \xrightarrow{\;T \to 0\;} \omega_0=
\frac{1}{2}\sqrt{K/S}\left(1-\sqrt{S/K}\right)^2$ and that $x,y
\xrightarrow{\;T \to 0\;} \infty$, we get $F_1(\sqrt{xy},\omega)
\xrightarrow{\;T \to 0\;} 0$. Using again the identity
$I_{2}(z)=I_{0}(z) -2I_1(z)/z$, \cite[Formula
9.6.26]{abramowitz-stegun-70}, we can similarly show that
$F_2(\sqrt{xy},\omega) \xrightarrow{\;T \to 0\;} 0$.

We now check the limits of $\chi^2(y;0,x)$. Because the
degree of freedom is zero, the results of Temme \cite{temme-93} do
not apply. We will work with the definition of the chi-square
function instead. We will check firstly the case where $S>K$.
From \eqref{eq:c-2} we get
\begin{equation}
\begin{split}
  \chi^2(y;0,x)
    & = \int_0^y
        \frac{1}{2}
            \left(\frac{z}{x}\right)^{-1/2}
            \e^{-(x+z)/2}
            I_{-1}\left(\sqrt{xz}\right)
            \, \d z \\
    & = \frac{\sqrt{x}}{2}
        \int_0^y
            \frac{1}{\sqrt{z}}
            \e^{-(x+z)/2}
            I_{1}\left(\sqrt{xz}\right)
            \, \d z \\
    & = \frac{\sqrt{x}}{2}
        \int_0^y
            \frac{1}{\sqrt{z}}
            \e^{-(\sqrt{x}-\sqrt{z})^2/2}
            \e^{-\sqrt{xz}}
            I_{1}\left(\sqrt{xz}\right)
            \, \d z \\
    & \sim \frac{\sqrt{x}}{2}
        \int_0^y
            \frac{1}{\sqrt{z}}
            \e^{-(\sqrt{x}-\sqrt{z})^2/2}
            (2/\pi)^{1/2}(xz)^{-1/4}
            \, \d z  \quad (T \to 0) \\
    & = \frac{x^{1/4}}{(2\pi)^{1/2}}
        \int_0^y
            z^{-3/4}
            \e^{-(\sqrt{x}-\sqrt{z})^2/2}
            \, \d z \\
    & \le \frac{x^{1/4}}{(2\pi)^{1/2}}
        \e^{-(\sqrt{x}-\sqrt{y})^2/2}
        \int_0^y
            z^{-3/4}
            \, \d z
     = \frac{4(xy)^{1/4}}{(2\pi)^{1/2}}
        \e^{-(\sqrt{x}-\sqrt{y})^2/2} \\
    & \le \frac{4\sqrt{x}}{\sqrt{2\pi}}
        \e^{-(\sqrt{x}-\sqrt{y})^2/2} \\
    & = \frac{4\sqrt{x}}{\sqrt{2\pi}}
        \e^{-x\left(1-\sqrt{y/x}\right)^2/2}
        \xrightarrow{\;T \to 0 \;} 0,
\end{split}
\end{equation}
where in the step for the asymptotic relation $\sim$ we have used
\eqref{eq:mattner-roos-inequality}, and in the last step the fact
that $x \to \infty$ and $y/x \to K/S$ as $T \to 0$. Finally, we
check the case where $S<K$. In this case
\begin{equation} \label{eq:S<K}
\begin{split}
    \tilde{\chi}^2(y;0,x)
        & = \int_y^\infty
                \frac{1}{2}\left(\frac{z}{x}\right)^{-1/2}
                \e^{-(x+y)/2} I_1(\sqrt{xz}) \, \d z \\
        & = \frac{\sqrt{x}}{2}
            \int_y^\infty
                \frac{1}{\sqrt{z}}
                \e^{-(\sqrt{x}-\sqrt{y})^2/2}
                \e^{-\sqrt{xz}} I_1(\sqrt{xz}) \, \d z \\
        & \sim \frac{\sqrt{x}}{2}
            \int_y^\infty
                \frac{1}{\sqrt{z}}
                \e^{-(\sqrt{x}-\sqrt{y})^2/2}
                (2/\pi)^{1/2} (xz)^{-1/4} \, \d z \quad (T \to 0) \\
        & = \frac{x^{1/4}}{\sqrt{2\pi}}
            \int_y^\infty
                \frac{1}{z^{3/4}}
                \e^{-(\sqrt{x}-\sqrt{y})^2/2} \, \d z \\
        & = \frac{2 x^{1/4}}{\sqrt{2\pi}}
            \int_{\sqrt{y}-\sqrt{x}}^\infty
                \frac{1}{\sqrt{u+\sqrt{x}}}
                \e^{-u^2/2} \, \d u \\
        & \le \frac{2 (x/y)^{1/4}}{\sqrt{2\pi}}
            \int_{\sqrt{y}-\sqrt{x}}^\infty
                \e^{-u^2/2} \, \d u
                \xrightarrow{\;T\to 0 \;} 0,
\end{split}
\end{equation}
where in the fourth equality we have changed the variables by
using $u = \sqrt{z} - \sqrt{x}$. The proof is thus complete.
\end{proof}
\begin{corollary}
Under the MMM, we have
\begin{equation} \label{eq:x-chi-x-T-p}
  x\tilde{\chi}^2(y;0,x)/[x_T p(y;4,x)]
    \xrightarrow{\;T \to 0\;} 0, \quad \mbox{for $S<K$}.
\end{equation}
\end{corollary}
\begin{proof}
By Lemma \ref{lem:atm-prelim-1} (1) and \cite[Formula 9.7.1]{abramowitz-stegun-70}, we get, as $T \to 0$,
\begin{equation}
  p(y;4,x)
    \sim \frac{(y/x)^{1/2}}{2\sqrt{2\pi}(xy)^{1/4}}
            \e^{-(x+y)/2} \e^{\sqrt{xy}}
    = \frac{(y/x)^{1/2}}{2\sqrt{2\pi}(xy)^{1/4}}
            \e^{-(\sqrt{x}-\sqrt{y})^2/2}.
\end{equation}
This, together with the fact that $x_T = -s\varphi_T/\varphi^2$,
gives
\begin{equation} \label{eq:chi-e-frac}
\begin{split}
    \frac{x\tilde{\chi}^2(y;0,x)}{x_T p(y;4,x)}
        & \sim \frac{x\tilde{\chi}^2(y;0,x)
                        }{x_T \frac{(y/x)^{1/2}}{2\sqrt{2\pi}(xy)^{1/4}}
            \e^{-(\sqrt{x}-\sqrt{y})^2/2}} \\
        & = 2\sqrt{2\pi}\frac{\varphi(xy)^{1/4}}{\varphi_T}
            \times \frac{\tilde{\chi}^2(y;0,x)}{\e^{-(\sqrt{x}-\sqrt{y})^2/2}}, \\
\end{split}
\end{equation}
where, as $\varphi \xrightarrow{\;T\to 0\;} 0$ and $\varphi_T
\xrightarrow{\;T\to 0\;} \alpha/4$,
\begin{equation}
\begin{split}
    \frac{\varphi(xy)^{1/4}}{\varphi_T}
    & = \frac{\varphi^{1/2}
    \left(SK\e^{-rT}\right)^{1/4}}{\varphi_T}
    \xrightarrow{\;T\to 0\;} 0.
\end{split}
\end{equation}
The second fraction in \eqref{eq:chi-e-frac} also converges to
zero because the last inequality in \eqref{eq:S<K} implies that
\begin{equation}
\begin{split}
    \frac{\tilde{\chi}^2(y;0,x)}{\e^{-(\sqrt{x}-\sqrt{y})^2/2}}
        & \le
            \frac{2 (x/y)^{1/4}}{\sqrt{2\pi}}
            \times \frac{
                \int_{\sqrt{y}-\sqrt{x}}^\infty
                    \e^{-u^2/2} \, \d u
                }{
                    \e^{-(\sqrt{x}-\sqrt{y})^2/2}
                }
            \xrightarrow{\;T\to 0 \;} 0.
\end{split}
\end{equation}
The proof is now complete.
\end{proof}

\begin{lemma} \label{lem:c-T-0}
Under the MMM, we have, for $S,K,T \in (0,\infty)$, and $t\in
[0,T]$,
\begin{equation}
  \lim_{T \to t} \call(S,t;K,T) = (S-K)_+.
\end{equation}
\end{lemma}
\begin{proof}
Platen and Heath \cite[(12.3.11)]{platen-heath-06} have shown that under the MMM, the benchmarked European call option price $\call$ is
\begin{equation}
  \hat{C}(S,t;K,T)
    = \EE
        \left[
            \left.
                \frac{(S_T - K)_+}{S_T}
            \right| S_t = S
        \right].
\end{equation}
where $S_t$ is the underlying process in \eqref{eq:mmm-sde}, $K$ is the strike, $t$ the current time, $T$ the expiry, and $0\le t \le T <\infty$. By a variant of the Feynman--Kac theorem derived in Janson and Tysk \cite[Theorem 5.5]{janson-tysk-06}, we have
\begin{equation} \label{eq:stmmm:pde-chat}
\left\{
\begin{split}
  & \hat{C}_t
    = - \frac{1}{2} \sigma^2(S,t) S^2 \hat{C}_{SS}
        -  [r + \sigma^2(S,t)]S \hat{C}_S,
        \quad (S,t) \in (0,\infty) \times [0,T), \\
  & \lim_{t \to T} \hat{C}(S,t;K,T)
     = \frac{(S - K)_+}{S}, \quad S \in (0,\infty). \\
\end{split}
\right.
\end{equation}
Let $\tau = T-t$. Then $\tau \in [0,T]$, and by the Markov property of $S_t$,
\begin{equation}
    c(K,\tau)
        \equiv \EE
        \left[
            \left.
                \frac{(S_\tau - K)_+}{S_\tau}
            \right| S_0 = S
        \right]
        =\hat{C}(S,T-\tau;K,T).
\end{equation}
The limit in \eqref{eq:stmmm:pde-chat} implies that $c(K,\tau) \xrightarrow{\; \tau \to 0 \;} (S-K)_+/S$. Now set $C(K,\tau) = S c(K,\tau)$, $\tau \in [0,T]$. Then
\begin{equation}
  \lim_{\tau \to 0 } C(K,\tau) = \lim_{\tau \to 0} S c(K,\tau) = S\frac{(S - K)_+}{S}= (S-K)_+.
\end{equation}
And the proof is complete.
\end{proof}
\begin{remark}
In \eqref{eq:stmmm:pde-chat}, the benchmarked price $\hat{C}$ also satisfies the boundary condition
\begin{equation}
  \lim_{S \to 0} \hat{C}(S,t;K,T)
    = 0, \quad 0 \le t \le T < \infty, \quad 0 < K < \infty.
\end{equation}
This can be verified by using \eqref{eq:call-price-expectation} and \eqref{eq:c-1}.
\end{remark}


\section{Appendix C: auxiliary results for the large time limit} \label{sec:appendix-c}

In this appendix we shall use the following notation:
\begin{equation} \label{eq:rhat-star-definition}
    \hat{r}_* = r+\eta.   \\
\end{equation}
Taking into consideration \eqref{eq:vhi-definition}, we shall sometimes write
\begin{equation} \label{eq:vhi-alt-expression}
  \vhi = \frac{\sqrt{2\hat{r}_*}}{1-\epsilon}, \qquad 0< \epsilon \ll 1.
\end{equation}

\subsection{Properties of the terms in the MMM call price}


\begin{lemma} \label{lem:varphi-x-y-at-infty}
Under the MMM, we have the following limits as $T \to \infty$:
\begin{enumerate}
\item $\displaystyle \varphi, \varphi_T, \varphi_{TT} \longrightarrow \infty $.
\item $\displaystyle x, x_T, x_{TT} \longrightarrow 0 $.
\item $\displaystyle y, y_T \longrightarrow 0 $.
\item $\displaystyle
        \frac{\varphi_T}{\varphi}
        = \frac{\eta \e^{\eta T}}{\e^{\eta T}-1} \longrightarrow \eta $.
\item $\displaystyle    \frac{x_T}{x} \longrightarrow - \eta $.
\item $\displaystyle    \frac{x_{TT}}{x} \longrightarrow \eta^2 $.
\item $\displaystyle    \frac{y}{x} \longrightarrow 0 $.
\item $\displaystyle    xy \longrightarrow 0 $.
\item $\displaystyle
        \frac{y_T}{x_T} = \frac{K}{\e^{rT}} \times \frac{r + \varphi_T/\varphi}{S \varphi_T/\varphi}
            \sim \frac{K}{S\e^{rT}} \cdot \frac{(r+\eta)}{\eta } \longrightarrow 0 $.
\item $\displaystyle    \e^{\eta T} x_T = - S \frac{\e^{\eta T}}{\varphi}
                \frac{\varphi_T}{\varphi} \longrightarrow - \frac{4 S \eta^2}{\alpha}$.
\item $\displaystyle \frac{x_{TTT}}{x_{T}} \longrightarrow \eta^2$.
\end{enumerate}
\end{lemma}
%
%
\begin{proof}
We omit the details because all of the limits in this lemma result from similar calculation based on the definition and/or the L'Hopital rule. We only mention that for (11),
\begin{equation}
\begin{split}
    \frac{x_{TTT}}{x_{T}}
        & =
            \frac{
                - 6 S \varphi^{-4} \varphi_T^3
                + 6 S \varphi^{-3} \varphi_T \varphi_{TT}
                - S \varphi^{-2} \varphi_{TTT}
                }{
                - S \varphi^{-2} \varphi_T
                } \\
        & =
            6 \varphi^{-2} \varphi_T^2
            - 6 \varphi^{-1}\varphi_{TT}
            + \varphi_{TTT} \varphi_T^{-1}
            \xrightarrow{\;T \to \infty\;}
            6 \eta^2 - 6 \eta^2 + \eta^2 = \eta^2.
\end{split}
\end{equation}
And the proof is complete.
\end{proof}

The following lemma holds for the $\chi^2$ related terms.
\begin{lemma} \label{lem:chisq-density-at-infty}
Under the MMM, we have the following asymptotics as $T \to \infty$:
\begin{enumerate}
\item $\displaystyle p(y;4,x) \sim y/4 $.
\item $\displaystyle    p(y;2,x) \sim 1/2 $.
\item $\displaystyle    p(y;0,x) \sim x/4 $.
\item $\displaystyle   \frac{p(y;0,x)}{p(y;4,x)} \frac{y_T}{x_T}
        \longrightarrow \frac{r+\eta}{r} $.
\item $\displaystyle   \frac{p(y;2,x)}{p(y;4,x)} \sim \frac{2}{y} $.
\item $\displaystyle    \frac{p(y;6,x)}{p(y;4,x)} \sim \frac{y}{4} $.
\item $\displaystyle    \frac{\e^{-rT}\tilde{\chi}^2(y;0,x)}{p(y;4,x)} \longrightarrow \frac{2S}{K} $.
\item $\displaystyle    \frac{\e^{-rT}x_T p(y;2,x)}{p(y;4,x)} \sim - \frac{2S}{K} \frac{\varphi_T}{\varphi} \longrightarrow -\frac{2S\eta}{K} $.
\item $\displaystyle    \frac{\chi^2(y;4,x)}{p(y;4,x)} \longrightarrow 0 $.
\item $\displaystyle \chi^2(y;4,x) \longrightarrow 0 $.
\item $\displaystyle \chi^2(y;0,x) \longrightarrow 1 $.
\item $\displaystyle \tilde{\chi}^2(y;0,x) \longrightarrow 0 $.
\end{enumerate}
\end{lemma}
\begin{proof}
Straightforward application of the definitions and if needed of L'Hopital's rule. Nevertheless, the last two limits require some more explanation. As mentioned in Siegel \cite{siegel-79}, $p(y;0,x)$ is an improper density and
\begin{equation}
  \chi^2(y;0,x)
        = \frac{1}{2}\e^{-x/2} + \int_0^y p(z;0,x) \, \d z.
\end{equation}
As $T \to 0$ the integral goes to zero because $x, y \xrightarrow{\; T \to \infty \;} 0$ and
\begin{equation}
\begin{split}
  \int_0^y p(z;0,x)  \, \d z
    & = \int_0^y \frac{1}{2}
            \left(\frac{z}{x}\right)^{-1/2} \e^{-(x+z)/2}
            I_{-1}\left(\sqrt{xz}\right) \, \d z \\
    & = \int_0^y \frac{1}{2}
            \sqrt{\frac{x}{z}} \e^{-(x+z)/2}
            I_{1}\left(\sqrt{xz}\right) \, \d z \\
    & \sim \int_0^y \frac{1}{2}
            \sqrt{\frac{x}{z}} \e^{-(x+z)/2}
            \frac{\sqrt{xz}}{2} \, \d z   \qquad (xz \to 0) \\
    & = \int_0^y \frac{1}{4}
            x \e^{-(x+z)/2}\, \d z \xrightarrow{\; x,y \to 0 \;} 0,   \\
\end{split}
\end{equation}
where the third line above results from \cite[Formula 9.6.10]{abramowitz-stegun-70}. This gives (11). The limit in (12) follows from the definition that $\tilde{\chi}^2(y;0,x) = 1 - \chi^2(y;0,x)$.
\end{proof}

\subsection{Properties of the terms in the Black--Scholes call price}

\begin{lemma} \label{lem:rhat-formula}
Let $\theta = 1 - \e^{-x/2}$. Then we have the following representation formulas:
\begin{equation}
\left\{
\begin{split}
    \hat{r}
        & = r - \frac{1}{T} \ln \theta, \\
    \hat{r}_T
        & = \frac{1}{T^2} \ln \theta
            - \frac{1}{T} \frac{\theta_T}{\theta}, \\
    \hat{r}_{TT}
        & = \frac{1}{T^2}
            \left(
                - \frac{2}{T}\ln\theta
                +2 \frac{\theta_T}{\theta}
                - T \frac{\theta_{TT}\theta - \theta_T^2}{\theta^2}
            \right).
\end{split}
\right.
\end{equation}
\end{lemma}
\begin{proof}
Straightforward differentiation.
\end{proof}

\begin{lemma} \label{lem:theta-at-infty}
As $T \to \infty$, we have the following limits for the $\theta$ related terms:
\begin{enumerate}
\item $\displaystyle \theta_T = \frac{1}{2} x_T \e^{-x/2} \longrightarrow 0.$
\item $\displaystyle \theta_{TT} = \frac{1}{2} x_{TT} \e^{-x/2}
            - \frac{1}{4} x_T^2 \e^{-x/2}.$
\item $\displaystyle \frac{\theta_T}{\theta} = \frac{1}{2} \frac{x_T \e^{-x/2}}{1-\e^{-x/2}}
            \longrightarrow - \eta.$
\item $\displaystyle \frac{\theta_{TT}}{\theta} \longrightarrow \eta^2. $
\item $\displaystyle \lim_{T\to \infty}
        \frac{1}{T}\ln \theta
          = \lim_{T\to \infty} \frac{\theta_T}{\theta}
          = \lim_{T\to \infty}
            \frac{1}{2}
            \left(\frac{x_T \e^{-x/2}}{1-\e^{-x/2}}\right)
            = - \eta.$
\end{enumerate}
\end{lemma}
\begin{proof}
These limits follow directly from the definition of the variables and if necessary from an additional application of  L'Hopital's rule. In particular, we have
\begin{equation}
\begin{split}
    \lim_{T\to \infty}
        \frac{1}{T}\ln \theta
        & = \lim_{T\to \infty} \frac{\theta_T}{\theta}
          = \lim_{T\to \infty}
            \frac{1}{2}
            \left(\frac{x_T \e^{-x/2}}{1-\e^{-x/2}}\right) \\
        & = \lim_{T\to \infty}
            \frac{1}{2}
            \left(\frac{x_T }{\e^{x/2}-1}\right)
          = \lim_{T\to \infty}
            \frac{1}{2}
            \left(\frac{x_{TT} }{(x_T \e^{x/2})/2}\right) \\
        & = \lim_{T\to \infty}
            \frac{x_{TT} }{x_T \e^{x/2}}
          = \lim_{T\to \infty}
            \frac{2S\varphi^{-3}\varphi_T^2-s\varphi^{-2}\varphi_{TT}}{-s\varphi^{-2}\varphi_T \e^{x/2}} \\
        & = \lim_{T\to \infty}
            \left[
                - 2\frac{\varphi_T}{\varphi}
                + \frac{\varphi_{TT}}{\varphi_T}
            \right] \frac{1}{\e^{x/2}}
          = -2\eta + \eta
          = - \eta.
\end{split}
\end{equation}
And this completes the proof.
\end{proof}
We have the following lemma for the limits of $\hat{r}$ and its related terms.
\begin{lemma}\label{lem:rhat-at-infty}
Let $\hat{r}_* = r + \eta$. Then the following limits hold as $T \to \infty$:
\begin{enumerate}
\item   $\hat{r} \longrightarrow \hat{r}_*.$
\item   $\hat{r}_T \longrightarrow 0. $
\item   $\hat{r}_T T \longrightarrow 0.$
\item   $\hat{r}_{TT} \longrightarrow 0.$
\item   $(\hat{r} - \hat{r}_*)T \longrightarrow - \ln \frac{2\eta S}{\alpha}.$
\end{enumerate}
\end{lemma}
\begin{proof}
The limit in (1) follows from Lemma \ref{lem:theta-at-infty}. The limit in (2) holds as
\begin{equation}
\begin{split}
    \hat{r}_T
        & = \left[
                \frac{1}{T} \left(\frac{1}{T} \ln \theta \right) - \frac{1}{T} \left(\frac{\theta_T}{\theta}\right)
            \right]
            \sim
            \left[\frac{1}{T} (-\eta) - \frac{1}{T} (-\eta)\right]
             \xrightarrow{\;T \to \infty\;} 0.
\end{split}
\end{equation}
For the limit in (3), we have
\begin{equation}
\begin{split}
    \hat{r}_T T
        & = \left(\frac{1}{T} \ln \theta  - \frac{\theta_T}{\theta} \right)
        \xrightarrow{\;T \to \infty\;} [-\eta - (-\eta)] = 0.
\end{split}
\end{equation}
The limit in (4) follows from the calculation that
\begin{equation}
\begin{split}
    \hat{r}_{TT}
        & = \left[
            \frac{1}{T^2} \left(-\frac{2}{T} \ln \theta  + 2\frac{\theta_T}{\theta}\right)
             - \frac{1}{T} \left(\frac{\theta_{TT}}{\theta} - \frac{\theta_T^2}{\theta^2} \right)
            \right]
            \xrightarrow{\;T \to \infty\;} 0.
\end{split}
\end{equation}
We now prove (5). By \eqref{eq:rhat-1} and Lemma \ref{lem:rhat-formula},
\begin{equation}
  f \equiv (\hat{r} - \hat{r}_*) T
  = \left[r - \frac{1}{T} \ln \theta -  (r + \eta) \right] T
  = - \ln \theta - \eta T.
\end{equation}
Recall that \eqref{eq:c-2} gives $\varphi = \frac{\alpha}{4\eta} (\e^{\eta T} - 1)$.
This gives
\begin{equation}
  \eta T = \ln \left(\frac{4\eta\varphi}{\alpha} +1 \right), \quad
  T = \frac{1}{\eta} \ln \left(\frac{4\eta\varphi}{\alpha} +1 \right).
\end{equation}
Noting that $\varphi \to \infty$ as $T \to \infty$, we want to express $f$ as a function of $\varphi$ and find its limit accordingly. Now we have
\begin{equation}
     f =  - \ln \theta - \eta T
         = - \ln \theta - \ln \left(\frac{4\eta\varphi}{\alpha} +1 \right) = - \ln \theta \left(\frac{4\eta\varphi}{\alpha} +1 \right).
\end{equation}
So our aim now is to find the following limit:
\begin{equation}
  \lim_{T \to \infty} f
  = \lim_{T \to \infty} \left[- \ln \theta \left(\frac{4\eta\varphi}{\alpha} +1 \right)\right]
  = \lim_{\varphi \to \infty} \left[- \ln \theta \left(\frac{4\eta\varphi}{\alpha} +1 \right)\right].
\end{equation}
From Lemma \ref{lem:rhat-formula} we know that $\theta = 1 - \e^{-x/2}$; and by definition $ x= S/\varphi$, see \eqref{eq:c-2}. Hence
\begin{equation}
    x_\varphi = - \frac{S}{\varphi^2}, \quad
    \theta_\varphi = \frac{x_\varphi}{2} \e^{-x/2} = - \frac{S}{2\varphi^2} \e^{-x/2}.
\end{equation}
Moreover, it can be checked that as a function of $\varphi$, the function $\theta \to 0$ as $\varphi \to 0$; similarly we have $x \to 0$ as $\varphi \to 0$. So by the L'Hopital rule, we get
\begin{equation}
\begin{split}
\lim_{\varphi \to \infty} \theta \left(\frac{4\eta\varphi}{\alpha} +1 \right)
    & \equiv \lim_{\varphi \to \infty} \frac{\theta}{\left(\frac{4\eta\varphi}{\alpha} +1 \right)^{-1}}
     = \lim_{\varphi \to \infty}
        \frac{\partial_\varphi \theta}{\partial_\varphi\left[\left(\frac{4\eta\varphi}{\alpha} +1 \right)^{-1}\right]},
\end{split}
\end{equation}
provided the last limit exists. The last limit does exist as
\begin{equation}
\begin{split}
  \frac{\partial_\varphi \theta}{\partial_\varphi\left[\left(\frac{4\eta\varphi}{\alpha} +1 \right)^{-1}\right]}
    & = \frac{- \frac{S}{2\varphi^2} \e^{-x/2}}{-\left(\frac{4\eta\varphi}{\alpha}+1\right)^{-2} \frac{4\eta}{\alpha}}
      = \frac{\alpha S}{8\eta}\times \frac{\left(\frac{4\eta\varphi}{\alpha}+1\right)^2}{\varphi^2} \e^{-x/2} \\
    & = \frac{\alpha S}{8\eta}\left(\frac{4\eta}{\alpha}+\frac{1}{\varphi}\right)^2 \e^{-x/2}
        \xrightarrow{\;\varphi \to \infty \;} \frac{\alpha S}{8\eta} \times \frac{16 \eta^2}{\alpha^2} = \frac{2\eta S}{\alpha}.
\end{split}
\end{equation}
From this the desired limit follows and the proof is thus complete.
\end{proof}

\subsection{Large time limits associated with $\vhi$.}

\begin{lemma} \label{lem:d1-d2-at-infty}
Under the MMM, we have, for each $K \in (0,\infty)$,
\begin{enumerate}
\item $\displaystyle d_1(K,T;\vhi) \xrightarrow{\;T \to \infty\;} \infty.$
\item $\displaystyle d_2(K,T;\vhi) \xrightarrow{\;T \to \infty\;} - \infty.$
\item $\displaystyle \exp\left[-\left(\frac{d_1^2(K,T;\vhi)}{2} - \hat{r}_*T\right)\right]
  \xrightarrow{\;T \to \infty\;} 0.$
\end{enumerate}
\end{lemma}
\begin{proof}
We shall prove (1) first. Let
\begin{equation}
  \beta_1^\epsilon = \frac{\hat{r}}{\vhi} + \frac{\vhi}{2}.
\end{equation}
The by the definition of $d_1$,
\begin{equation}
\begin{split}
  d_1(K,T;\vhi)
    & = \frac{\ln(S/K) + (\hat{r} + {\vhi}^2/2)T}{\vhi\sqrt{T}}
     = \frac{\ln(S/K)}{\vhi\sqrt{T}}
        +  \frac{(\hat{r} + {\vhi}^2/2)T}{\vhi\sqrt{T}} \\
    & = \frac{\ln(S/K)}{\vhi\sqrt{T}}
        +  \left(\frac{\hat{r}}{\vhi} + \frac{\vhi}{2}\right)\sqrt{T}
     = \frac{\ln(S/K)}{\vhi\sqrt{T}}
        +  \beta_1^\epsilon \sqrt{T}.  \\
\end{split}
\end{equation}
For each $0<\epsilon \ll 1$, $\vhi$ is a positive constant, so
\begin{equation} \label{eq:lem:d1-d2-at-infty-1}
  \frac{\ln(S/K)}{\vhi\sqrt{T}} \xrightarrow{\;T \to \infty\;} 0.
\end{equation}
By Lemma \ref{lem:rhat-at-infty} (1), $\hat{r} \xrightarrow{\;T \to \infty\;} \hat{r}_*>0$. This implies that
\begin{equation}
  \beta_1^\epsilon \xrightarrow{\;T \to \infty\;} \const>0.
\end{equation}
Consequently we get (1).

To prove (2), we let
\begin{equation}
  \beta_2^\epsilon = \frac{\hat{r}}{\vhi} - \frac{\vhi}{2}.
\end{equation}
Then we can write $d_2(K,T;\vhi)$ as
\begin{equation}
\begin{split}
  d_2(K,T;\vhi)
    & = \frac{\ln(S/K) + (\hat{r} - {\vhi}^2/2)T}{\vhi\sqrt{T}}
      = \frac{\ln(S/K)}{\vhi\sqrt{T}}
        +  \frac{(\hat{r} - {\vhi}^2/2)T}{\vhi\sqrt{T}} \\
    & = \frac{\ln(S/K)}{\vhi\sqrt{T}}
        +  \left(\frac{\hat{r}}{\vhi} - \frac{\vhi}{2}\right)\sqrt{T}
      = \frac{\ln(S/K)}{\vhi\sqrt{T}}
        +  \beta_2^\epsilon \sqrt{T}.  \\
\end{split}
\end{equation}
As already noted,
\begin{equation}
  \frac{\ln(S/K)}{\vhi\sqrt{T}} \xrightarrow{\;T \to \infty\;} 0.
\end{equation}
Thus we only need to show that $\beta_2^\epsilon \sqrt{T} \xrightarrow{\;T \to \infty\;} -\infty$. We rewrite $\beta_2^\epsilon$ as
\begin{equation}
\begin{split}
  \beta_2^\epsilon
    & = \frac{\hat{r}}{\sqrt{2\hat{r}_*}/(1-\epsilon)}
        - \frac{\sqrt{2\hat{r}_*}/(1-\epsilon)}{2}
      = \frac{\hat{r}}{\sqrt{2\hat{r}_*}}(1-\epsilon)
        - \frac{\sqrt{2\hat{r}_*}}{2} \frac{1}{1-\epsilon} \\
    & = 
            \frac{\hat{r}}{\sqrt{2\hat{r}_*}}
            - \frac{\sqrt{2\hat{r}_*}}{2}
        + \frac{\sqrt{2\hat{r}_*}}{2}
        - \frac{\sqrt{2\hat{r}_*}}{2} \frac{1}{1-\epsilon}
        - \frac{\hat{r}}{\sqrt{2\hat{r}_*}} \epsilon \\
    & = \left(
            \frac{\hat{r}}{\sqrt{2\hat{r}_*}}
            - \frac{\sqrt{2\hat{r}_*}}{2}
        \right)
        + \frac{\sqrt{2\hat{r}_*}}{2}
            \left(
                1- \frac{1}{1-\epsilon}
            \right)
        - \frac{\hat{r}}{\sqrt{2\hat{r}_*}} \epsilon \\
    & = \left(
            \frac{\hat{r}}{\sqrt{2\hat{r}_*}}
            - \frac{\sqrt{2\hat{r}_*}}{2}
        \right)
        - \frac{\sqrt{2\hat{r}_*}}{2}
                \frac{\epsilon}{1-\epsilon}
        - \frac{\hat{r}}{\sqrt{2\hat{r}_*}} \epsilon \\
    & = \left(
            \frac{\hat{r}}{\sqrt{2\hat{r}_*}}
            - \frac{\sqrt{2\hat{r}_*}}{2}
        \right)
        - \epsilon \frac{\sqrt{2\hat{r}_*}}{2}
            \left(
                \frac{1}{1-\epsilon}
                - \frac{\hat{r}}{\hat{r}_*}
            \right). \\
\end{split}
\end{equation}
Using Lemma \ref{lem:rhat-at-infty} (1) we get
\begin{equation} \label{eq:lem:d1-d2-at-infty-2}
  \left(
    \frac{\hat{r}}{\sqrt{2\hat{r}_*}} - \frac{\sqrt{2\hat{r}_*}}{2}
  \right) \xrightarrow{\;T \to \infty\;} 0
  \quad \mbox{ and } \quad
  \frac{\hat{r}}{\hat{r}_*} \xrightarrow{\;T \to \infty\;} 1.
\end{equation}
As a result,
\begin{equation} \label{eq:lem:d1-d2-at-infty-3}
            \left(
                \frac{1}{1-\epsilon}
                - \frac{\hat{r}}{\hat{r}_*}
            \right)
            \xrightarrow{\;T \to \infty\;} \const(\epsilon) >0.
\end{equation}
This then leads to
\begin{equation}
  \beta_2^\epsilon \xrightarrow{\;T \to \infty\;} \const(\epsilon) <0
  \quad \mbox{ and } \quad
   \beta_2^\epsilon \sqrt{T} \xrightarrow{\;T \to \infty\;} -\infty.
\end{equation}
This proves (2).

We now prove (3). Our strategy is to show that $ [d_1^2 (K,T;\vhi)/2 - \hat{r}_*T] \xrightarrow{\;T \to \infty\;} \infty. $ By definition,
\begin{equation} \label{eq:lem:d1-d2-at-infty-4}
\begin{split}
  \frac{1}{2} d_1^2 (K,T;\vhi) - \hat{r}_*T
    & = \frac{1}{2}
        \left[
            \frac{\ln(S/K)}{\vhi\sqrt{T}}
            + \beta_1^\epsilon \sqrt{T}
        \right]^2
        - \hat{r}_* T
      = \frac{1}{2}
        \left[
            \frac{\ln(S/K)}{\vhi T}
            + \beta_1^\epsilon
        \right]^2 T
        - \frac{\Bigl(\sqrt{2\hat{r}_*}\Bigr)^2 T}{2} \\
    & = \frac{T}{2}
        \left[
            \frac{\ln(S/K)}{\vhi T}
            + \beta_1^\epsilon + \sqrt{2\hat{r}_*}
        \right]
        \left[
            \frac{\ln(S/K)}{\vhi T}
            + \beta_1^\epsilon - \sqrt{2\hat{r}_*}
        \right].
         \\
\end{split}
\end{equation}
Here the first square bracketed term is strictly positive as $T$ tends to infinity. For the second square bracketed term, we note that
\begin{equation}
\begin{split}
  \beta_1^\epsilon
    & = \frac{\hat{r}}{\sqrt{2\hat{r}_*}/(1-\epsilon)}
        + \frac{\sqrt{2\hat{r}_*}/(1-\epsilon)}{2}
      = \frac{\hat{r}}{\sqrt{2\hat{r}_*}}(1-\epsilon)
        + \frac{\sqrt{2\hat{r}_*}}{2} \frac{1}{1-\epsilon} \\
    & = \left(
            \frac{\hat{r}}{\sqrt{2\hat{r}_*}}
            + \frac{\sqrt{2\hat{r}_*}}{2}
        \right)
        + \frac{\sqrt{2\hat{r}_*}}{2} \frac{1}{1-\epsilon}
        - \frac{\sqrt{2\hat{r}_*}}{2}
        - \frac{\hat{r}}{\sqrt{2\hat{r}_*}} \epsilon \\
    & = \left(
            \frac{\hat{r}}{\sqrt{2\hat{r}_*}}
            + \frac{\sqrt{2\hat{r}_*}}{2}
        \right)
        + \frac{\sqrt{2\hat{r}_*}}{2}
            \left(
                \frac{1}{1-\epsilon}
                -  1
            \right)
        - \frac{\hat{r}}{\sqrt{2\hat{r}_*}} \epsilon \\
    & = \left(
            \frac{\hat{r}}{\sqrt{2\hat{r}_*}}
            + \frac{\sqrt{2\hat{r}_*}}{2}
        \right)
        + \frac{\sqrt{2\hat{r}_*}}{2}
                \frac{\epsilon}{1-\epsilon}
        - \frac{\hat{r}}{\sqrt{2\hat{r}_*}} \epsilon \\
    & = \left(
            \frac{\hat{r}}{\sqrt{2\hat{r}_*}}
            + \frac{\sqrt{2\hat{r}_*}}{2}
        \right)
        + \epsilon \frac{\sqrt{2\hat{r}_*}}{2}
            \left(
                \frac{1}{1-\epsilon}
                - \frac{\hat{r}}{\hat{r}_*}
            \right). \\
\end{split}
\end{equation}
Hence
\begin{equation}
\begin{split}
  \beta_1^\epsilon - \sqrt{2\hat{r}_*}
    & = \left(
            \frac{\hat{r}}{\sqrt{2\hat{r}_*}}
            + \frac{\sqrt{2\hat{r}_*}}{2}
        \right)
        - \sqrt{2\hat{r}_*}
        + \epsilon \frac{\sqrt{2\hat{r}_*}}{2}
            \left(
                \frac{1}{1-\epsilon}
                - \frac{\hat{r}}{\hat{r}_*}
            \right) \\
    & = \left(
            \frac{\hat{r}}{\sqrt{2\hat{r}_*}}
            - \frac{\sqrt{2\hat{r}_*}}{2}
        \right)
        + \epsilon \frac{\sqrt{2\hat{r}_*}}{2}
            \left(
                \frac{1}{1-\epsilon}
                - \frac{\hat{r}}{\hat{r}_*}
            \right) \\
\end{split}
\end{equation}
By \eqref{eq:lem:d1-d2-at-infty-2} and \eqref{eq:lem:d1-d2-at-infty-3}, we obtain
\begin{equation}
  \beta_1^\epsilon - \sqrt{2\hat{r}_*} \xrightarrow{\;T \to \infty\;} \const(\epsilon) >0.
\end{equation}
Combining this with \eqref{eq:lem:d1-d2-at-infty-1} gives
\begin{equation}
  \left[
            \frac{\ln(S/K)}{\vhi T}
            + \beta_1^\epsilon - \sqrt{2\hat{r}_*}
  \right]
  \xrightarrow{\;T \to \infty\;} \const(\epsilon) >0.
\end{equation}
This then gives
\begin{equation}
  \frac{1}{2} d_1^2 (K,T;\vhi) - \hat{r}_*T \xrightarrow{\;T \to \infty\;} \infty.
\end{equation}
The desired limit in (3) follows from this. And the proof is hence complete.
\end{proof}

\begin{lemma}
Assume the MMM. Then for any $v \in (0,\infty)$,
\begin{equation}
    d_1(K,T;v) \xrightarrow{\;T \to \infty\;} \infty,
\end{equation}
and
\begin{equation} \label{eq:d2-large-time-limits}
\left\{
  \begin{array}{ll}
    d_2(K,T;v) \xrightarrow{\;T \to \infty\;} - \infty, & \hbox{if $v>\sqrt{2\hat{r}_*}$;} \\
    d_2(K,T;v) \xrightarrow{\;T \to \infty\;} 0, & \hbox{if $v=\sqrt{2\hat{r}_*}$;} \\
    d_2(K,T;v) \xrightarrow{\;T \to \infty\;} \infty, & \hbox{if $v<\sqrt{2\hat{r}_*}$.}
  \end{array}
\right.
\end{equation}
\end{lemma}
\begin{proof}
Let $m=\ln(S/K)$. Recall from \eqref{eq:bs-N-d} that
\begin{equation}
\left\{
\begin{split}
    d_1(K,T;v)
        & = \frac{m + (\hat{r} + v^2/2)T}{v\sqrt{T}}
          = \frac{m}{v\sqrt{T}} + \frac{\hat{r} + v^2/2}{v}\sqrt{T}, \\
    d_2(K,T;v)
        & = \frac{m + (\hat{r} - v^2/2)T}{v\sqrt{T}}
          = \frac{m}{v\sqrt{T}} + \frac{\hat{r} - v^2/2}{v}\sqrt{T}. \\
\end{split}
\right.
\end{equation}
Since $\hat{r}\xrightarrow{\;T \to \infty\;} \hat{r}_* = r+\eta$, we have $d_1(K,T;v) \xrightarrow{\;T \to \infty\;} \infty $ for any $v \in (0,\infty)$. The limits for $d_2$ can be derived from the limits
\begin{equation}
\left\{
  \begin{array}{ll}
    \hat{r}-v^2/2 \xrightarrow{\;T \to \infty\;} \const_1 < 0, & \hbox{if $v>\sqrt{2\hat{r}_*}$;} \\
    \hat{r}-v^2/2 \xrightarrow{\;T \to \infty\;} 0, & \hbox{if $v=\sqrt{2\hat{r}_*}$;} \\
    \hat{r}-v^2/2 \xrightarrow{\;T \to \infty\;} \const_2 >0, & \hbox{if $v<\sqrt{2\hat{r}_*}$.}
  \end{array}
\right.
\end{equation}
And the proof is complete.
\end{proof}

\begin{lemma}
Assume the MMM. Then
\begin{equation} \label{eq:exp-short-expansion}
    \e^{-[\hat{r} - (r+\eta)]T}
    = \frac{2\eta S}{\alpha} + \frac{e_1}{\varphi} + \frac{e_2}{\varphi^2}
        + \O(\varphi^{-3}),
\end{equation}
where
\begin{equation}
  e_1 = \frac{S}{2} - \frac{\eta S^2}{2\alpha}, \qquad
  e_2 = \frac{\eta S^3}{12 \alpha} - \frac{S^2}{8}.
\end{equation}
\end{lemma}
\begin{proof}
It is well known that as $z \to 0$,
\begin{equation}
\begin{split}
  \e^{-z}
    & = 1 -z + \frac{z^2}{2!} - \frac{z^3}{3!}
        + \frac{z^4}{4!} - \frac{z^5}{5!} + \frac{z^6}{6!} + \O(z^7) \\
    & = 1 -z + \frac{z^2}{2} - \frac{z^3}{6}
        + \frac{z^4}{24} - \frac{z^5}{120} + \frac{z^6}{720} + \O(z^7). \\
\end{split}
\end{equation}
This gives
\begin{equation}
    1- \e^{-z}
     = z - \frac{z^2}{2} + \frac{z^3}{6}
         - \frac{z^4}{24} + \frac{z^5}{120} - \frac{z^6}{720} + \O(z^7).
\end{equation}
By Lemma \ref{lem:varphi-x-y-at-infty} (2), $x \xrightarrow{\;T\to \infty\;} 0$. This implies that
\begin{equation}
    1- \e^{-x/2}
     = \frac{x}{2} - \frac{x^2}{8} + \frac{x^3}{48}
         - \frac{x^4}{384} + \frac{x^5}{3840} - \frac{x^6}{46080} + \O(x^7).
\end{equation}
By definition, $\varphi = \frac{\alpha}{4\eta}(\e^{\eta T} -1)$; so $ \e^{\eta T} = \frac{4\eta \varphi}{\alpha} + 1$.
This then gives
\begin{equation}
  \e^{-[\hat{r} - (r+\eta)]T}
    = \e^{\eta T}(1- \e^{-x/2})
    = (1- \e^{-x/2})\left(\frac{4\eta\varphi}{\alpha} + 1\right).
\end{equation}
Noting that $x = S/\varphi$, we get, as $T \to \infty$,
\begin{equation}
\begin{split}
  & \e^{-[\hat{r} - (r+\eta)]T} \\
  &  = \left[
        \frac{S}{2\varphi}
        - \frac{S^2}{8\varphi^2}
        + \frac{S^3}{48\varphi^3}
        - \frac{S^4}{384\varphi^4}
        + \frac{S^5}{3840\varphi^5}
        - \frac{S^6}{46080\varphi^6}
        + \O(\varphi^{-7})
      \right]
      \left(\frac{4\eta\varphi}{\alpha}+1\right).
\end{split}
\end{equation}
From this the desired expansion then follows.
\end{proof}

Recall from \eqref{eq:bs-N-d} the definitions for $d_1$ and $d_2$. Then we have the following lemma.
\begin{lemma}
Under the MMM, the following identity holds:
\begin{equation} \label{eq:d1-squared}
  \frac{1}{2} d_1^2(K,T;v) - (r+\eta)T
    =  \frac{1}{2}d_2^2(K,T;v) + \ln(S/K) + (\hat{r} - (r+\eta))T.
\end{equation}
\end{lemma}
\begin{proof}
We shall suppress the argument in the proof. Since $d_1 = d_2 + v\sqrt{T}$, we have
\begin{equation}
\begin{split}
  d_1^2
    & = (d_2+v\sqrt{T})^2
      = d_2^2 + 2 d_2 v\sqrt{T} + v^2 T
      = d_2^2 + 2\ln(S/K) + 2\hat{r}T,
\end{split}
\end{equation}
where the last equality follows from the identity
\begin{equation}
  v\sqrt{T}d_2
    = v\sqrt{T}
        \frac{\ln(S/K) + (\hat{r} - v^2/2)T}{v\sqrt{T}}
    = \ln(S/K) + (\hat{r} - v^2/2)T.
\end{equation}
This then implies that
\begin{equation}
\begin{split}
  \frac{1}{2}d_1^2 - (r+\eta)T
    & = \frac{1}{2}d_2^2 + \ln(S/K) + (\hat{r} T - (r+\eta))T. \\
\end{split}
\end{equation}
And the proof is complete.
\end{proof}



\section{Appendix D: Proof of the lemmas for the large time limit} \label{sec:appendix-d}

In this section we present the proof of the lemmas needed for the large time limit in Section \ref{sec:proof-large-iv}.

\subsection{Proof of Lemma \ref{lem:iv-upper-bound-ratio-1}} \label{sec:proof:lem:iv-upper-bound-ratio-1}

\begin{proof}[Proof of Lemma \ref{lem:iv-upper-bound-ratio-1}]
It suffices to show that
\begin{equation}
  \lim_{T \to \infty}
        \frac{\chi^2(y;4,x)}{\e^{-rT}\tilde{\chi}^2(y;0,x)} = 0.
\end{equation}
This limit is valid since
\begin{equation}
\begin{split}
  \frac{\chi^2(y;4,x)}{\e^{-rT}\tilde{\chi}^2(y;0,x)}
       & =  \frac{\chi^2(y;4,x)/p(y;4,x)}{\e^{-rT}\tilde{\chi}^2(y;0,x)/p(y;4,x)}
         \xrightarrow{\;T \to \infty\;} \frac{0}{2S/K} \\
\end{split}
\end{equation}
by Lemma \ref{lem:chisq-density-at-infty} (7) and (9).
\end{proof}

\subsection{Proof of Lemma \ref{lem:iv-upper-bound-ratio-2}} \label{sec:proof:lem:iv-upper-bound-ratio-2}

\begin{proof}[Proof of Lemma \ref{lem:iv-upper-bound-ratio-2}]
Our aim is to check if
\begin{equation} 
  \lim_{T \to \infty}
    \frac{\tilde{N}(d_1(K,T;\vhi))}{ \e^{-rT} \tilde{\chi}^2(y;0,x)} = 0.
\end{equation}
We can rewrite the ratio as
\begin{equation}
  \frac{\tilde{N}(d_1(K,T;\vhi))}{ \e^{-rT} \tilde{\chi}^2(y;0,x)}
  = \frac{\tilde{N}(d_1(K,T;\vhi))/p(y;4,x)}{ \e^{-rT} \tilde{\chi}^2(y;0,x)/p(y;4,x)}=\frac{A_2}{B_2}.
\end{equation}
Lemma \ref{lem:chisq-density-at-infty} (7) gives $B_2 \xrightarrow{\;T \to \infty\;} 2S/K$.
By Lemma \ref{lem:d1-d2-at-infty} (1), $d_1(K,T;\vhi) \xrightarrow{\;T \to \infty\;} \infty$. Therefore, as $T \to \infty$,
\begin{equation}
\begin{split}
   A_2 
    & \sim  \frac{\dfrac{n(d_1(K,T;\vhi))}{d_1(K,T;\vhi)}
                }{
                    \dfrac{1}{2}\left(\dfrac{y}{x}\right)^{1/2}
                    \exp\left(- \dfrac{x+y}{2}\right)
                    I_1\left(\sqrt{xy}\right)
                } \qquad
                    \mbox{by \eqref{eq:c-2} and \cite[Formulas 7.1.23, 26.2.12, 9.6.6]{abramowitz-stegun-70}} \\
    & = \frac{\dfrac{1}{\sqrt{2\pi}}
                 \exp\left(-\dfrac{d_1^2(K,T;\vhi)}{2}\right)
                }{
                    \dfrac{1}{2}\left(\dfrac{y}{x}\right)^{1/2}
                    \exp\left(- \dfrac{x+y}{2}\right)
                    I_1\left(\sqrt{xy}\right) d_1(K,T;\vhi)
                } \\
    & \sim \frac{\dfrac{1}{\sqrt{2\pi}}
                 \exp\left(-\dfrac{d_1^2(K,T;\vhi)}{2}\right)
                }{
                    \dfrac{1}{2}\left(\dfrac{y}{x}\right)^{1/2}
                    \dfrac{\sqrt{xy}}{2} d_1(K,T;\vhi)
                }
                    \qquad
                    \begin{array}{l}
                      \mbox{as $x,y \to 0$ by Lemma \ref{lem:varphi-x-y-at-infty} (2),(3)}; \\
                      \mbox{and $I_1(\sqrt{xy}) \sim \sqrt{xy}/2$ by \cite[Formula 9.6.10]{abramowitz-stegun-70}}
                    \end{array}
                \\
    & = \frac{\dfrac{4}{\sqrt{2\pi}}
                 \exp\left(-\dfrac{d_1^2(K,T;\vhi)}{2}\right)
                }{
                    y d_1(K,T;\vhi)
                } \\
    & = \frac{\dfrac{4 \varphi \e^{-rT}}{\sqrt{2\pi}}
                 \exp\left(-d_1^2(K,T;\vhi)/2\right)
                }{
                    K d_1(K,T;\vhi)
                }
                \qquad \mbox{by the definition of $y$, see \eqref{eq:c-2}.}   \\
\end{split}
\end{equation}
Recalling from \eqref{eq:c-2} that $\varphi = \frac{\alpha}{4\eta}\left(\e^{\eta T} -1\right)$, we get, as $T \to \infty$,
\begin{equation} \label{eq:lem-proof-nd1-1}
\begin{split}
    A_2
        & \sim
             \frac{\alpha}{\sqrt{2\pi}\eta K d_1(K,T;\vhi)}
                \left(\e^{\eta T} -1\right)
                \e^{rT}
                \exp\left(-d_1^2(K,T;\vhi)/2\right) \\
        & =
             \frac{\alpha}{\sqrt{2\pi}\eta K d_1(K,T;\vhi)}
                \left(
                    \frac{\e^{\eta T} -1}{\e^{\eta T}}
                \right)
                \exp\left[(r+\eta)T\right]
                \exp\left(-d_1^2(K,T;\vhi)/2\right) \\
        & =
             \frac{\alpha}{\sqrt{2\pi}\eta K d_1(K,T;\vhi)}
                \left(
                    1 - \e^{-\eta T}
                \right)
                \exp\left[-\left(\frac{d_1^2(K,T;\vhi)}{2} - \hat{r}_*T\right)\right]. \\
\end{split}
\end{equation}
By Lemma \ref{lem:d1-d2-at-infty} (1) and (3), $d_1(K,T;\vhi) \xrightarrow{\;T \to \infty\;} \infty$ and
\begin{equation}
  \exp\left[-\left(\frac{d_1^2(K,T;\vhi)}{2} - \hat{r}_*T\right)\right]
  \xrightarrow{\;T \to \infty\;} 0.
\end{equation}
Hence $A_2 \xrightarrow{\;T \to \infty\;} 0$. In turn, $A_2/B_2 \xrightarrow{\;T \to \infty\;} 0$, and the proof is complete.
\end{proof}

\subsection{Proof of Lemma \ref{lem:iv-upper-bound-ratio-3}} \label{sec:proof:lem:iv-upper-bound-ratio-3}

\begin{proof}[Proof of Lemma \ref{lem:iv-upper-bound-ratio-3}]
We shall verify that
\begin{equation} 
  \lim_{T \to \infty}
    \frac{\e^{-\hat{r}T}N\bigl(d_2(K,T;\vhi)\bigr)}{\e^{-rT}\tilde{\chi}^2(y;0,x)} =
    0.
\end{equation}
Rewrite the ratio as
\begin{equation}
  \frac{\e^{-\hat{r}T}N\bigl(d_2(K,T;\vhi)\bigr)}{\e^{-rT}\tilde{\chi}^2(y;0,x)}
  = \frac{\e^{-\hat{r}T}N\bigl(d_2(K,T;\vhi)\bigr)/p(y;4,x)}{\e^{-rT}\tilde{\chi}^2(y;0,x)/p(y;4,x)}
  = \frac{A_3}{B_3}.
\end{equation}
As before, Lemma \ref{lem:chisq-density-at-infty} (7) gives $B_3 \xrightarrow{\;T \to \infty\;} 2S/K$; and Lemma \ref{lem:d1-d2-at-infty} (2) gives $d_2(K,T;\vhi)\xrightarrow{\;T \to \infty\;} - \infty$. Hence $ N(d_2(K,T;\vhi))\xrightarrow{\;T \to \infty\;} 0$. Further, as $T \to \infty$,
\begin{equation}
\begin{split}
    \frac{\e^{-\hat{r}T}}{p(y;4,x)}
        & \sim \frac{\e^{-\hat{r}T}}{y/4}  \qquad \mbox{by Lemma \ref{lem:chisq-density-at-infty} (1)} \\
        & = \frac{4 \e^{-\hat{r}T}}{K\e^{-rT}/\varphi }
            \qquad \mbox{by the definition of $y$, see \eqref{eq:c-2}} \\
        & = \frac{4 \varphi}{K} \e^{-(\hat{r}-r)T}
            \qquad \mbox{by the definition of $\varphi$, see \eqref{eq:c-2}} \\
        & = \frac{\alpha}{K\eta}(\e^{\eta T} -1) \e^{-(\hat{r}-r)T} \\
        & = \frac{\alpha}{K\eta}(1- \e^{-\eta T}) \e^{-[\hat{r}-(r+\eta)]T} \\
        & = \frac{\alpha}{K\eta}(1- \e^{-\eta T}) \e^{-(\hat{r}-\hat{r}_*)T}
            \qquad \mbox{by the definition of $\hat{r}_*$, see \eqref{eq:rhat-star-definition} } \\
        & \xrightarrow{\;T \to \infty\;}
            \frac{\alpha}{K\eta}
            \exp \left(\ln \frac{2\eta S}{\alpha}\right) = \frac{2S}{K},
            \quad \mbox{by Lemma \ref{lem:rhat-at-infty} (5).}
\end{split}
\end{equation}
Hence $A_3 \xrightarrow{\;T \to \infty\;} 0$. As a result, we have $A_3/B_3 \xrightarrow{\;T \to \infty\;} 0$. The proof is thus complete.
\end{proof}

\subsection{Proof of Lemma \ref{lem:cbs-call-R}} \label{subsec:proof:lem:cbs-call-R}

\begin{proof}[Proof of Lemma \ref{lem:cbs-call-R}]
Recall from \eqref{eq:c-1} and \eqref{eq:bs-2} that the generic Black--Scholes call price with volatility $v \in [0,\infty]$ and the MMM call price are respectively given by
\begin{equation}
\left\{
\begin{split}
  \cbs(K,T;v)
    & = S N(d_1(K,T;v)) - K \e^{-\hat{r} T} N(d_2(K,T;v)),  \quad v \in [0,\infty], \\
  \call(K,T)
    & = S \tilde{\chi}^2(y;4,x) - K \e^{-r T} \tilde{\chi}^2(y;0,x).
\end{split}
\right.
\end{equation}
By the definition of the complementary function $\tilde{N}(d) = 1- N(d)$, we can rewrite the Black--Scholes price as
\begin{equation}
\begin{split}
    \cbs(K,T;v)
    & = S [1-\tilde{N}(d_1)]
         - K \e^{-\hat{r} T} [1-\tilde{N}(d_2)] \\
    & = S - S \tilde{N}(d_1)
         - K \e^{-\hat{r} T}  + K \e^{-\hat{r} T} \tilde{N}(d_2). \\
\end{split}
\end{equation}
Similarly, we can rewrite the MMM call price $\call$ as
\begin{equation}
\begin{split}
  \call(K,T)
    & = S [1- \chi^2(y;4,x)] - K \e^{-r T} [1-\chi^2(y;0,x)] \\
    & = S - S \chi^2(y;4,x) - K \e^{-r T} + K \e^{-r T}\chi^2(y;0,x). \\
\end{split}
\end{equation}
Taking into consideration of the definition of $\rcbs$ and $\rcall$, a rearrangement of the terms above then gives the desired expressions for $\cbs$ and $\call$.
\end{proof}

\subsection{Proof of Lemma \ref{lem:chisq-bounds}} \label{subsec:proof:lem:chisq-bounds}

\begin{proof}[Proof of Lemma \ref{lem:chisq-bounds}]
\textbf{Step 1: Proof of \eqref{eq:lem:chisq-bounds-1}.} We shall bound $\chi^2(y;4,x)$ first. To obtain an upper bound for $\chi^2(y;4,x)$ we apply the definition \eqref{eq:c-2} to get
\begin{equation}
\begin{split}
  \chi^2(y;4,x)
    & = \int_0^y
            \frac{1}{2} \left(\frac{z}{x}\right)^{1/2}
            \exp\left(-\frac{x+z}{2}\right)I_1\Bigl(\sqrt{xz}\Bigr) \, \d z \\
    & \le \frac{1}{2} \e^{-x/2}
        \int_0^y
             \left(\frac{z}{x}\right)^{1/2}
            I_1\Bigl(\sqrt{xz}\Bigr) \, \d z \\
    & = \frac{1}{2} \e^{-x/2} \left[2 \frac{y}{x} I_2(\sqrt{xy})\right] \\
    & = \e^{-x/2}\frac{y}{x}I_2(\sqrt{xy}),
\end{split}
\end{equation}
where the penultimate equality is valid because by the change of the variables $u = \sqrt{xz}$ and explicit integration:
\begin{equation}
\begin{split}
    \int_0^y
             \left(\frac{z}{x}\right)^{1/2}
            I_1\Bigl(\sqrt{xz}\Bigr) \, \d z
    & = \int_0^{\sqrt{xy}}
        \frac{1}{2} \sqrt{\frac{u^2/x}{x}} I_1(u) \frac{2u}{x} \,\d u
      = \int_0^{\sqrt{xy}}
        \frac{u^2}{x^2} I_1(u)\,\d u \\
    & = \frac{1}{x^2}
        \int_0^{\sqrt{xy}}
        u^2 I_1(u)\,\d u \\
    & = \frac{1}{x^2} xy I_2(\sqrt{xy})
        \qquad  \mbox{by \cite[Formula 9.6.28]{abramowitz-stegun-70}}    \\
    & = \frac{y}{x}I_2(\sqrt{xy}).
\end{split}
\end{equation}
Then a lower bound for $\chi^2(y;4,x)$ can be obtained by noting that for $0\le z \le y$, $\e^{-(x+z)/2} \ge \e^{-(x+y)/2}$ and that
\begin{equation}
\begin{split}
  \chi^2(y;4,x)
    & = \int_0^y
            \frac{1}{2} \left(\frac{z}{x}\right)^{1/2}
            \exp\left(-\frac{x+z}{2}\right)I_1\Bigl(\sqrt{xz}\Bigr) \, \d z \\
    & \ge \frac{1}{2} \e^{-(x+y)/2}
        \int_0^y
             \left(\frac{z}{x}\right)^{1/2}
            I_1\Bigl(\sqrt{xz}\Bigr) \, \d z \\
    & =  \e^{-(x+y)/2}\frac{y}{x}I_2(\sqrt{xy}).
\end{split}
\end{equation}
Combining the lower and upper bounds for $\chi^2(y;4,x)$ then gives
\begin{equation} \label{eq:ncx2cdf-bound-large-time-1}
  \e^{-(x+y)/2} \frac{y}{x} I_2(\sqrt{xy})
  \le \chi^2(y;4,x)
  \le \e^{-x/2} \frac{y}{x} I_2(\sqrt{xy}) .
\end{equation}
We now bound $\chi^2(y;0,x)$. By \eqref{eq:chisq-zero-degree-1},
\begin{equation}
\begin{split}
  \chi^2(y;0,x)
    & = \e^{-x/2} + \int_0^y p(z;0,x) \, \d z
\end{split}
\end{equation}
By \eqref{eq:c-2} and the Bessel identity $I_{-1}(\cdot) = I_{1}(\cdot)$
\begin{equation}
\begin{split}
\int_0^y p(z;0,x) \, \d z
    & = \int_0^y \frac{1}{2}\left(\frac{z}{x}\right)^{(0-2)/4}
        \exp\left(-\frac{x+z}{2}\right) I_{(0-2)/2}(\sqrt{xz}) \, \d z \\
    & = \int_0^y \frac{1}{2}\left(\frac{z}{x}\right)^{-1/2}
        \exp\left(-\frac{x+z}{2}\right)I_{-1}(\sqrt{xz}) \, \d z \\
    & = \int_0^y \frac{1}{2}\sqrt{\frac{x}{z}}
        \exp\left(-\frac{x+z}{2}\right)I_{1}(\sqrt{xz}) \, \d z. \\
\end{split}
\end{equation}
This implies that
\begin{equation}
  \e^{-(x+y)/2}
  \int_0^y \frac{1}{2} \sqrt{\frac{x}{z}} I_{1}(\sqrt{xz}) \, \d z
  \le \int_0^y p(z;0,x) \, \d z
  \le
  \e^{-x/2}
  \int_0^y \frac{1}{2} \sqrt{\frac{x}{z}} I_{1}(\sqrt{xz}) \, \d z.
\end{equation}
Again, setting $u = \sqrt{xz}$ gives
\begin{equation}
\begin{split}
    \int_0^y \frac{1}{2} \sqrt{\frac{x}{z}} I_{1}(\sqrt{xz}) \, \d z
    & = \int_0^{\sqrt{xy}}
        \frac{1}{2} \sqrt{\frac{x}{u^2/x}} I_1(u) \frac{2u}{x} \,\d u
     = \int_0^{\sqrt{xy}}
         I_1(u)\,\d u \\
    & = I_0(\sqrt{xy}) -1.
\end{split}
\end{equation}
Hence,
\begin{equation}
  \e^{-(x+y)/2}
  [I_0(\sqrt{xy}) -1]
  \le \int_0^y p(z;0,x) \, \d z
  \le
  \e^{-x/2}
  [I_0(\sqrt{xy}) -1].
\end{equation}
Combining this with the definition of $\chi^2(y;0,x)$ gives
\begin{equation}
  \e^{-(x+y)/2}
  [I_0(\sqrt{xy}) -1] + \e^{-x/2}
  \le \chi^2(y;0,x)
  \le
  \e^{-x/2}
  [I_0(\sqrt{xy}) -1]+ \e^{-x/2}.
\end{equation}
Summing up the bounds for $\chi^2(y;4,x)$ and $\chi^2(y;0,x)$ gives, for all $T>0$,
\begin{equation}
\left\{
\begin{split}
& \e^{-(x+y)/2} \frac{y}{x} I_2(\sqrt{xy})
  \le \chi^2(y;4,x)
  \le \e^{-x/2} \frac{y}{x} I_2(\sqrt{xy}), \\
& \e^{-(x+y)/2}
  [I_0(\sqrt{xy}) -1] + \e^{-x/2}
  \le \chi^2(y;0,x)
  \le
  \e^{-x/2}
  [I_0(\sqrt{xy}) -1]+ \e^{-x/2}. \\
\end{split}
\right.
\end{equation}
By \cite[Formula 9.6.10]{abramowitz-stegun-70},
\begin{equation}
  I_2(\sqrt{xy}) = \frac{xy}{4} \sum_{j=0}^\infty \frac{(xy/4)^j}{j!(2+j)!}
  \quad \mbox{ and } \quad
  I_0(\sqrt{xy}) = \sum_{j=0}^\infty \frac{(xy/4)^j}{(j!)^2}.
\end{equation}
Incorporating these series representations into the bounds gives the desired inequalities
\begin{equation}
\left\{
\begin{split}
& \e^{-(x+y)/2} \frac{y^2}{4} \sum_{j=0}^\infty \frac{(xy/4)^j}{j!(2+j)!}
  \le \chi^2(y;4,x)
  \le \e^{-x/2} \frac{y^2}{4} \sum_{j=0}^\infty \frac{(xy/4)^j}{j!(2+j)!}, \\
& \e^{-(x+y)/2}
  \sum_{j=1}^\infty \frac{(xy/4)^j}{(j!)^2} + \e^{-x/2}
  \le \chi^2(y;0,x)
  \le
  \e^{-x/2}
  \sum_{j=0}^\infty \frac{(xy/4)^j}{(j!)^2}+ \e^{-x/2}. \\
\end{split}
\right.
\end{equation}

\textbf{Step 2: Proof of \eqref{eq:lem:chisq-bounds-2}.} Under the MMM, we have $x,y>0$ for all $T>0$. So the lower bounds in \eqref{eq:lem:chisq-bounds-2} result from a truncation of the series at the seconder order term. Note that these lower bounds hold for all $T>0$. To obtain the large time upper bounds in \eqref{eq:lem:chisq-bounds-2}, we use the results of Lemma \ref{lem:varphi-x-y-at-infty} (2) and (3), which state respectively that $x,y \xrightarrow{\;T \to \infty\;} 0$. These results imply that
\begin{equation}
\begin{split}
  \sum_{j=2}^\infty \frac{(xy/4)^j}{j!(2+j)!}
    & \le \sum_{j=2}^\infty \frac{(xy/4)^j}{j!}
     = \sum_{l=0}^\infty \frac{(xy/4)^{l+2}}{(l+2)!}
     = \frac{(xy)^2}{4^2} \sum_{l=0}^\infty \frac{(xy/4)^{l+2}}{(l+2)!} \\
    & \le \frac{(xy)^2}{4^2} \sum_{l=0}^\infty \frac{(xy/4)^{l}}{l!}
      = \frac{(xy)^2}{4^2} \e^{xy/4} \\
    & \le (xy)^2 \qquad \mbox{for sufficiently large $T$;}
\end{split}
\end{equation}
and similarly
\begin{equation}
\begin{split}
  \sum_{j=2}^\infty \frac{(xy/4)^j}{(j!)^2}
    & \le \sum_{j=2}^\infty \frac{(xy/4)^j}{j!}
      \le (xy)^2 \qquad \mbox{for sufficiently large $T$.}
\end{split}
\end{equation}
And the proof is complete.
\end{proof}

\subsection{Proof of Lemma \ref{lem:R-asymptotics}} \label{subsec:proof:lem:R-asymptotics}

\begin{proof}[Proof of Lemma \ref{lem:R-asymptotics}]
We will prove \eqref{eq:lem:R-asymptotics-1} first. By the definition of $\rcall$, we have
\begin{equation}
\begin{split}
  \e^{(r+\eta)T}\rcall(K,T)
    & = \frac{K}{S}e^{-rT + (r+\eta)T} + \e^{(r+\eta)T}\chi^2(y;4,x) - \frac{K}{S} \e^{-rT+(r+\eta)T}\chi^2(y;0,x), \\
    & = \frac{K}{S}e^{\eta T} + \e^{(r+\eta)T}\chi^2(y;4,x) - \frac{K}{S} \e^{\eta T}\chi^2(y;0,x).  \\
\end{split}
\end{equation}
By \eqref{eq:lem:chisq-bounds-2}, we have, for large enough $T$,
\begin{equation} \label{eq:lem:R-asymptotics-1-upper-inequality}
\begin{split}
  \e^{(r+\eta)T}\rcall(K,T)
    & = \frac{K}{S}e^{\eta T} + \e^{(r+\eta)T}\chi^2(y;4,x) - \frac{K}{S} \e^{\eta T}\chi^2(y;0,x)  \\
    & \le
        \frac{K}{S}e^{\eta T}
        +  \e^{(r+\eta)T} \e^{-x/2}
            \frac{y^2}{4}
            \left[
                \frac{1}{2}
                + \frac{xy}{24}
                +  (xy)^2
            \right] \\
    & \quad
         -\frac{K}{S} \e^{\eta T}
            \left\{
                \e^{-x/2}
                + \e^{-(x+y)/2}
                    \left[
                        \frac{xy}{4}
                        + \frac{(xy)^2}{64}
                    \right]
            \right\} \\
    & =
        \frac{K}{S}e^{\eta T}
        - \frac{K}{S}e^{\eta T} \e^{-x/2}
        +  \e^{(r+\eta)T} \e^{-x/2}
            \frac{y^2}{4}
            \left[
                \frac{1}{2}
                + \frac{xy}{24}
                +  (xy)^2
            \right] \\
    & \quad
         -\frac{K}{S} \e^{\eta T-(x+y)/2}
            \left[
                \frac{xy}{4}
                + \frac{(xy)^2}{64}
            \right]  \\
    & =
        \frac{K}{S} e^{\eta T} (1-\e^{-x/2})
        +  \e^{(r+\eta)T} \e^{-x/2}
            \frac{y^2}{4}
            \left[
                \frac{1}{2}
                + \frac{xy}{24}
                +  (xy)^2
            \right] \\
    & \quad
         -\frac{K}{S} \e^{\eta T-(x+y)/2}
            \left[
                \frac{xy}{4}
                + \frac{(xy)^2}{64}
            \right] \\
    & =
        \frac{K}{S} e^{-[\hat{r}-(r+\eta)]T}
        +  \rcallu(K,T), \\
\end{split}
\end{equation}
where in the last equality we have used the identity
\begin{equation}
  \e^{-[\hat{r}-(r+\eta)T]} = \e^{\eta T} (1-\e^{-x/2}).
\end{equation}
This proves the second inequality in \eqref{eq:lem:R-asymptotics-1}. To prove the first inequality, we apply \eqref{eq:lem:chisq-bounds-2} again to get, for sufficiently large $T$,
\begin{equation} \label{eq:lem:R-asymptotics-1-lower-inequality}
\begin{split}
  \e^{(r+\eta)T}\rcall(K,T)
    & = \frac{K}{S}e^{\eta T} + \e^{(r+\eta)T}\chi^2(y;4,x) - \frac{K}{S} \e^{\eta T}\chi^2(y;0,x)  \\
    & \ge
        \frac{K}{S}e^{\eta T}
        +  \e^{(r+\eta)T}
        \e^{-(x+y)/2}
            \frac{y^2}{4}
            \left[
                \frac{1}{2}
                + \frac{xy}{24}
                +  \frac{(xy)^2}{768}
            \right] \\
    & \quad
         -\frac{K}{S} \e^{\eta T}
            \left\{
                \e^{-x/2}
                + \e^{-x/2}
                    \left[
                        \frac{xy}{4}
                        + (xy)^2
                    \right]
            \right\} \\
    & \ge
        \frac{K}{S}e^{\eta T} (1-\e^{-x/2})
        + \e^{(r+\eta)T-(x+y)/2}
            \frac{y^2}{4}
            \left[
                \frac{1}{2}
                + \frac{xy}{24}
                +  \frac{(xy)^2}{768}
            \right] \\
    & \quad
         -\frac{K}{S} \e^{\eta T-x/2}
            \left[
                \frac{xy}{4}
                + (xy)^2
            \right] \\
    & =
        \frac{K}{S}e^{-[\hat{r}-(r+\eta)]T}
        + \rcalll(K,T).
\end{split}
\end{equation}
This proves \eqref{eq:lem:R-asymptotics-1}. Next, we will prove \eqref{eq:lem:R-asymptotics-2}. Put
\begin{equation}
  \lambda = \left[\frac{4\eta K}{\alpha(1-\e^{-\eta T})}\right]^2,
  \quad
  \gamma = \frac{16 \eta^2 SK}{\alpha^2(1-\e^{-\eta T})^2}.
\end{equation}
Then
\begin{equation} \label{eq:gamma-limit}
    xy = \gamma \e^{-(r+\eta)T - \eta T},
    \quad
        \gamma \xrightarrow{\;T\to \infty\;} \frac{16 \eta^2 SK}{\alpha^2},
\end{equation}
and
\begin{equation} \label{eq:lambda-limit}
    y^2 = \lambda \e^{-2(r+\eta)T},
    \quad
    \lambda \xrightarrow{\;T\to \infty\;} \frac{16 \eta^2 K^2}{\alpha^2}.
\end{equation}
Substituting the expressions for $x$ and $y$ into $\rcalll$ gives
\begin{equation}
\begin{split}
    \rcalll(K,T)
         & = \e^{(r+\eta)T-(x+y)/2}
            \frac{y^2}{4}
            \left[
                \frac{1}{2}
                + \frac{xy}{24}
                +  \frac{(xy)^2}{768}
            \right] \\
         & \quad
          -\frac{K}{S} \e^{\eta T-x/2}
            \left[
                \frac{xy}{4}
                + (xy)^2
            \right] \\
         & = \e^{(r+\eta)T-(x+y)/2}
            \frac{\lambda \e^{-2(r+\eta)T}}{4}
            \left[
                \frac{1}{2}
                + \frac{\gamma \e^{-(r+\eta)T-\eta T}}{24}
                +  \frac{\gamma^2 \e^{-2(r+\eta)T-2\eta T}}{768}
            \right] \\
         & \quad
            - \frac{K}{S} \e^{\eta T-x/2}
                \left[
                    \frac{\gamma \e^{-(r+\eta)T - \eta T}}{4}
                    + \gamma^2 \e^{-2(r+\eta)T-2\eta T}
                \right] \\
         & = \e^{-(x+y)/2}
            \frac{\lambda \e^{-(r+\eta)T}}{4}
            \left(
                \frac{1}{2}
                + \frac{\gamma \e^{-(r+\eta)T-\eta T}}{24}
                +  \frac{\gamma^2 \e^{-2(r+\eta)T-2\eta T}}{768}
            \right) \\
         & \quad
            - \frac{K}{S} \e^{-x/2}
                \left(
                    \frac{\gamma \e^{-(r+\eta)T}}{4}
                    + \gamma^2 \e^{-2(r+\eta)T - \eta T}
                \right) \\
         & = \e^{-(x+y)/2}
            \left[
                \frac{\lambda}{8}\e^{-(r+\eta)T}
                + \frac{\lambda\gamma}{96} \e^{-2(r+\eta)T - \eta T}
                + \frac{\lambda\gamma^2 \e^{-3(r+\eta)T-2\eta T}}{3072}
            \right] \\
         & \quad
            - \frac{K}{S} \e^{-x/2}
                \left(
                    \frac{\gamma \e^{-(r+\eta)T}}{4}
                    + \gamma^2 \e^{-2(r+\eta)T - \eta T}
                \right) \\
         & = \e^{-(r+\eta)T}
            \left(
                \frac{\lambda}{8}\e^{-(x+y)/2}
                - \frac{\gamma K}{4S} \e^{-x/2}
            \right) \\
         & \quad
            + \e^{-2(r+\eta)T-\eta T}
            \left(
                \frac{\lambda\gamma}{96} \e^{-(x+y)/2}
                - \frac{\gamma^2 K}{S} \e^{-x/2}
            \right)
            + \O(\e^{-3(r+\eta)T-2\eta T}).
\end{split}
\end{equation}
By \eqref{eq:gamma-limit}, \eqref{eq:lambda-limit}, and the limits that $x,y \xrightarrow{\; T \to \infty\;} 0$, we get
\begin{equation} \label{eq:lambda-gamma-terms-limit-1}
\begin{split}
  \frac{\lambda}{8}\e^{-(x+y)/2}
    - \frac{\gamma K}{4S} \e^{-x/2}
    \xrightarrow{\; T \to \infty\;}
    &
    \frac{1}{8}\cdot \frac{16 \eta^2 K^2}{\alpha^2}
    - \frac{K}{4S} \cdot \frac{16 \eta^2 SK}{\alpha^2} \\
    & = \frac{2\eta^2K^2}{\alpha^2} - \frac{4\eta^2 K^2}{\alpha^2} \\
    & = - \frac{2\eta^2 K^2}{\alpha^2}.
\end{split}
\end{equation}
Also,
\begin{equation} \label{eq:lambda-gamma-terms-limit-2}
\begin{split}
  \frac{\lambda\gamma}{96} \e^{-(x+y)/2}
    - \frac{\gamma^2 K}{S} \e^{-x/2}
    \xrightarrow{\; T \to \infty\;}
    &
     \frac{1}{96} \cdot \frac{16 \eta^2 K^2}{\alpha^2} \cdot \frac{16 \eta^2 SK}{\alpha^2}
     - \frac{K}{S}\left(\frac{16 \eta^2 SK}{\alpha^2}\right)^2  \\
    & = -\frac{760 SK^3\eta^4}{3\alpha^4}.
\end{split}
\end{equation}
So as $T \to \infty$,
\begin{equation} \label{eq:rcalll-large-time-order}
  \rcalll(K,T) = -\frac{2\eta^2 K^2}{\alpha^2} \e^{-(r+\eta)T}
    -\frac{760 SK^3\eta^4}{3\alpha^4} \e^{-2(r+\eta)T-\eta T}
    + \O(\e^{-3(r+\eta)T-2\eta T}).
\end{equation}
This proves \eqref{eq:lem:R-asymptotics-2}.

We will prove \eqref{eq:lem:R-asymptotics-3} in a similar fashion as follows:
\begin{equation} \label{eq:rcallu-large-time-order}
\begin{split}
    \rcallu(K,T)
         &  = \e^{(r+\eta)T-x/2}
            \frac{y^2}{4}
            \left[
                \frac{1}{2}
                + \frac{xy}{24}
                +  (xy)^2
            \right]
            -\frac{K}{S} \e^{\eta T-(x+y)/2}
            \left[
                \frac{xy}{4}
                + \frac{(xy)^2}{64}
            \right] \\
         & = \e^{(r+\eta)T-x/2}
            \frac{\lambda \e^{-2(r+\eta)T}}{4}
            \left[
                \frac{1}{2}
                + \frac{\gamma \e^{-(r+\eta)T-\eta T}}{24}
                +  \gamma^2 \e^{-2(r+\eta)T-2\eta T}
            \right] \\
         & \quad
            - \frac{K}{S} \e^{\eta T-(x+y)/2}
                \left[
                    \frac{\gamma \e^{-(r+\eta)T - \eta T}}{4}
                    + \frac{\gamma^2 \e^{-2(r+\eta)T-2\eta T}}{64}
                \right] \\
         & = \e^{-x/2}
            \frac{\lambda \e^{-(r+\eta)T}}{4}
            \left(
                \frac{1}{2}
                + \frac{\gamma \e^{-(r+\eta)T-\eta T}}{24}
                +  \gamma^2 \e^{-2(r+\eta)T-2\eta T}
            \right) \\
         & \quad
            - \frac{K}{S} \e^{-(x+y)/2}
                \left(
                    \frac{\gamma \e^{-(r+\eta)T}}{4}
                    + \frac{\gamma^2 \e^{-2(r+\eta)T-\eta T}}{64}
                \right) \\
         & = \e^{-x/2}
            \left[
                \frac{\lambda}{8}\e^{-(r+\eta)T}
                + \frac{\lambda\gamma}{96} \e^{-2(r+\eta)T - \eta T}
                + \lambda\gamma^2 \e^{-3(r+\eta)T-2\eta T}
            \right] \\
         & \quad
            - \frac{K}{S} \e^{-(x+y)/2}
                \left(
                    \frac{\gamma \e^{-(r+\eta)T}}{4}
                    + \frac{\gamma^2 \e^{-2(r+\eta)T-\eta T}}{64}
                \right) \\
         & = \e^{-(r+\eta)T}
            \left(
                \frac{\lambda}{8}\e^{-x/2}
                - \frac{\lambda K}{4S} \e^{-(x+y)/2}
            \right) \\
         & \quad
            \e^{-2(r+\eta)T-\eta T}
            \left(
                \frac{\lambda\gamma}{96}\e^{-x/2}
                - \frac{K\gamma^2}{64S}\e^{-(x+y)/2}
            \right)
            + \O(\e^{-3(r+\eta)T-2\eta T}),  \\
         & =  -\frac{2\eta^2 K^2}{\alpha^2} \e^{-(r+\eta)T}
            -\frac{4 SK^3\eta^4}{3\alpha^4} \e^{-2(r+\eta)T-\eta T}
            + \O(\e^{-3(r+\eta)T-2\eta T}),
\end{split}
\end{equation}
as $T \to \infty$. In the last equality above we have used the limits in \eqref{eq:lambda-gamma-terms-limit-1} and \eqref{eq:lambda-gamma-terms-limit-2}.

We now prove \eqref{eq:lem:R-asymptotics-4}, the asymptotics of $\e^{(r+\eta)T}\rcbs$. By \eqref{eq:d2-large-time-limits}, for any fixed $K$ and $0< v <\sqrt{2(r+\eta)}$, $d_1(K,T;v), d_2(K,T;v) \xrightarrow{\;T \to \infty\;} \infty$. Hence, for $0< v <\sqrt{2(r+\eta)}$ we have, as $T \to \infty$,
\begin{equation}
\begin{split}
    \e^{(r+\eta)T}\rcbs(K,T;v)
        & = \e^{(r+\eta)T}
            \left[
            \frac{K}{S} \e^{-\hat{r}T}
            + \tilde{N}(d_1)
            - \frac{K}{S}\e^{-\hat{r}T} \tilde{N}(d_2)
            \right] \\
        & = \frac{K}{S} \e^{-[\hat{r}-(r+\eta)]T}
            + \e^{(r+\eta)T}\tilde{N}(d_1)
            + \e^{-[\hat{r}-(r+\eta)]T}\frac{K}{S}\tilde{N}(d_2) \\
        & = \frac{K}{S} \e^{-[\hat{r}-(r+\eta)]T}
            + \e^{(r+\eta)T} \frac{n(d_1)}{d_1}[ 1 + \O(d_1^{-2})] \\
        & \quad
            + \frac{K}{S}\e^{-[\hat{r}-(r+\eta)]T} \frac{n(d_2)}{d_2}[ 1 + \O(d_2^{-2})]. \\
\end{split}
\end{equation}
Applying the identity $Sn(d_1) = K \e^{-\hat{r}T} n(d_2)$ to the last equation gives
\begin{equation}
\begin{split}
    \e^{(r+\eta)T}\rcbs
        & = \frac{K}{S} \e^{-[\hat{r}-(r+\eta)]T}
            + \frac{K}{S}\e^{-[\hat{r}-(r+\eta)]T} \frac{n(d_2)}{d_1}[ 1 + \O(d_1^{-2})] \\
        & \quad
            + \frac{K}{S}\e^{-[\hat{r}-(r+\eta)]T} \frac{n(d_2)}{d_2}[ 1 + \O(d_2^{-2})]  \\
        & = \frac{K}{S} \e^{-[\hat{r}-(r+\eta)]T}
            + \frac{K}{S} \e^{-[\hat{r}-(r+\eta)]T}
            \left(
                \frac{1}{d_1} - \frac{1}{d_2}
            \right)
            n(d_2) \\
        & \quad
            + \frac{K}{S} \e^{-[\hat{r}-(r+\eta)]T} n(d_2) \O(d_1^{-3} + d_2^{-3}).
\end{split}
\end{equation}
Notice that
\begin{equation}
  \frac{1}{d_1} - \frac{1}{d_2}
  = \frac{d_2 - d_1}{d_1 d_2}
  = \frac{-v\sqrt{T}}{(d_2+v\sqrt{T})d_2}.
\end{equation}
Moreover, $d_1= d_2 + v\sqrt{T}$ implies that $\O(d_1^{-3} + d_2^{-3})= \O(d_2^{-3})$. Consequently
\begin{equation}
\begin{split}
    \e^{(r+\eta)T}\rcbs
        & = \frac{K}{S} \e^{-[\hat{r}-(r+\eta)]T}
            - \frac{K}{S} \e^{-[\hat{r}-(r+\eta)]T}
            \frac{v\sqrt{T}}{(d_2+v\sqrt{T})d_2}
            n(d_2) \\
        & \quad
            + \frac{K}{S} \e^{-[\hat{r}-(r+\eta)]T} n(d_2) \O(d_2^{-3}).
\end{split}
\end{equation}
Under the MMM, $\varphi = \O(e^{\eta T})$. Coupling this property with \eqref{eq:exp-short-expansion} gives
\begin{equation}
\begin{split}
    \e^{(r+\eta)T}\rcbs
        & = \frac{K}{S} \e^{-[\hat{r}-(r+\eta)]T} \\
        & \quad
            - \frac{K}{S}
            \left(
                \frac{2\eta S}{\alpha} + \frac{e_1}{\varphi} + \frac{e_2}{\varphi^2}
            + \O(\varphi^{-3})
            \right)
            \frac{v\sqrt{T}}{(d_2+v\sqrt{T})d_2}
            n(d_2) \\
        & \quad
            + \frac{K}{S}
            \left(
                \frac{2\eta S}{\alpha} + \frac{e_1}{\varphi} + \frac{e_2}{\varphi^2}
            + \O(\varphi^{-3})
            \right)
            n(d_2) \O(d_2^{-3}) \\
        & = \frac{K}{S} \e^{-[\hat{r}-(r+\eta)]T}
            - \frac{2K\eta}{\alpha}  \frac{v\sqrt{T}}{(d_2+v\sqrt{T})d_2} n(d_2)
            + \frac{2K\eta}{\alpha} n(d_2) \O(d_2^{-3}) \\
        & = \frac{K}{S} \e^{-[\hat{r}-(r+\eta)]T}
            - \frac{2K\eta}{\alpha}  \frac{v\sqrt{T}}{(d_2+v\sqrt{T})d_2\sqrt{2\pi}} \e^{-d_2^2/2}
            + \frac{2K\eta}{\alpha} n(d_2) \O(d_2^{-3}). \\
\end{split}
\end{equation}
Note that for any fixed $v \in (0,\sqrt{2(r+\eta)})$, $d_2(K,T;v) = \O(\sqrt{T})$, where the order also depends on $K,S, v$. This proves \eqref{eq:lem:R-asymptotics-4}, and the proof of the lemma is complete.
\end{proof}

\subsection{Proof of Lemma \ref{lem:d2-v}} \label{subsec:proof:lem:d2-v}

\begin{proof}[Proof of Lemma \ref{lem:d2-v}]
For ease of notation we set $m=\ln (S/K)$. Then by \eqref{eq:bs-N-d}, for $v>0$,
\begin{equation}
  d_2(K,T;v)
          = \frac{m + (\hat{r} - v^2/2)T}{v\sqrt{T}}
          = \frac{m}{v\sqrt{T}} + \frac{\hat{r} - v^2/2}{v}\sqrt{T}.
\end{equation}
This gives
\begin{equation}
\begin{split}
    d_2^2(K,T;v)
    & = \frac{m^2}{v^2 T}
        + \frac{2m(\hat{r}-v^2/2)}{v^2}
        +\frac{(\hat{r}-v^2/2)^2}{v^2}T. \\
\end{split}
\end{equation}
Recall that in \eqref{eq:rhat-star-definition} we define $\hat{r}_* = r+\eta$. Then
\begin{equation} \label{eq:proof:lem:d2-v:1}
\begin{split}
    \frac{1}{2}d_2^2(K,T;v) - (r+\eta) T
    & =  \frac{1}{2}d_2^2(K,T;v) - \hat{r}_* T  \\
    & = \frac{m^2}{2v^2 T}
        + \frac{m(\hat{r}-v^2/2)}{v^2}
        +\left[
            \frac{(\hat{r}-v^2/2)^2}{2v^2}
            - \hat{r}_*
         \right] T \\
    & = \frac{m^2}{2v^2 T}
        + \frac{m(\hat{r}-v^2/2)}{v^2}
        +\left[
            \frac{(\hat{r}-v^2/2)^2}{2v^2}
            - \hat{r}_*
         \right] T \\
    & = \frac{m^2}{2v^2 T}
        + \frac{m(\hat{r}-v^2/2)}{v^2}
        +
            \frac{
                Q(K,T;v)
                }{
                    8 v^2
               } T,    \\
\end{split}
\end{equation}
where
\begin{equation}
  Q(K,T;v) = v^4 -(4\hat{r} + 8 \hat{r}_*)v^2 + 4 \hat{r}^2.
\end{equation}
For fixed $K,v>0$,
\begin{equation} \label{eq:proof:lem:d2-v:2}
    \frac{m^2}{2v^2 T} \xrightarrow{\; T \to \infty\;} 0,
    \quad
    \frac{m(\hat{r}-v^2/2)}{v^2}
        \xrightarrow{\; T \to \infty\;} \frac{m(\hat{r}_*-v^2/2)}{v^2},
\end{equation}
since $\hat{r} \xrightarrow{\; T \to \infty\;} \hat{r}_*$ by Lemma \ref{lem:rhat-at-infty} (1). Notice that the second of the two limits above is a constant depending on $K,v, \alpha, \eta$ (and $S$). On the other hand, we also have
\begin{equation}
  Q_\infty(v) \equiv \lim_{T \to \infty} Q(K,T;v) = v^4 - 12 \hat{r}_* v^2 + 4 \hat{r}_*^2;
\end{equation}
and it is obvious that as $T$ tends to infinity the term $Q_\infty T/(8v^2)$ ultimately determines the behavior of $d_2^2/2 - \hat{r}_*T$. For arbitrary $v$, the four roots of the quartic polynomial $Q_\infty$ are
\begin{equation}
  \sqrt{2(3+2\sqrt{2})\hat{r}_*}, \quad
  -\sqrt{2(3+2\sqrt{2})\hat{r}_*}, \quad
  \sqrt{2(3-2\sqrt{2})\hat{r}_*}, \quad
  -\sqrt{2(3-2\sqrt{2})\hat{r}_*}. \quad
\end{equation}
However, for $v \in (0,\sqrt{2\hat{r}_*})$, the only possible root of $Q_\infty$ is $v_* = \sqrt{2(3-2\sqrt{2})\hat{r}_*}$. Now
\begin{equation}
  Q_\infty^\prime(v) \equiv \frac{\d Q_\infty(v)}{\d v}= 4 v^3 - 24 \hat{r}_* v = 4v(v^2 - 6 \hat{r}_*).
\end{equation}
Noting that $v^2 - 6 \hat{r}_*<0$ for all $v \in (0,\sqrt{2\hat{r}_*})$, we have $Q_\infty^\prime(v) <0$
for all $v \in (0,\sqrt{2\hat{r}_*})$. This implies that $Q_\infty$ is strictly decreasing in $(0,\sqrt{2\hat{r}_*})$,
with $Q_\infty(v_*) =0$. In other words, $Q_\infty>0$ in $(0,v_*)$ and $Q_\infty<0$ in $(v_*, \sqrt{2\hat{r}_*})$.
Consequently, for sufficiently large $T$,
\begin{equation} \label{eq:proof:lem:d2-v:3}
\left\{
  \begin{array}{ll}
    \frac{Q(K,T;v)}{8 v^2} T > \zeta_1 T, & \hbox{for any fixed $v \in (0,v_*)$,} \\
        & \\
    \frac{Q(K,T;v)}{8 v^2} T < - \zeta_2 T, & \hbox{for any fixed $v \in (v_*, \sqrt{2(r+\eta)})$,}
  \end{array}
\right.
\end{equation}
where $\zeta_1$ and $\zeta_2$ are some strictly positive constants dependant on $K,S,v,r,\eta$.
By combining \eqref{eq:proof:lem:d2-v:1}, \eqref{eq:proof:lem:d2-v:2}, and \eqref{eq:proof:lem:d2-v:3}, we get, for sufficiently large $T$,
\begin{equation}
\left\{
  \begin{array}{ll}
    \frac{1}{2}d_2^2(K,T;v) - (r+\eta)T > c_1 T, & \hbox{if $v \in (0,v_*)$,} \\
    \frac{1}{2}d_2^2(K,T;v) - (r+\eta)T < - c_2 T, & \hbox{if $v \in (v_*, \sqrt{2(r+\eta)})$,}
  \end{array}
\right.
\end{equation}
where $c_1$ and $c_2$ are some strictly positive constants dependant on $K,S,v,r,\eta$. And the proof is complete.
\end{proof}

\section{Appendix E: Remarks on the results of Gao and Lee \cite{gao-lee-2011}} \label{sec:appendix-e}


\subsection{Notation}

Gao and Lee \cite[(3.1)]{gao-lee-2011} defines $C_{-}:[0,\infty)\times (0,\infty) \to (0,\infty)$ as
\begin{equation}
  C_{-}(k,V)\equiv N(-k/V + V/2) - \e^{k} N(-k/V - V/2).
\end{equation}
This formula can be easily derived from \eqref{eq:bs-2} (or \eqref{eq:bs-1}). Put
\begin{equation}
\left\{
\begin{split}
    f & = S\e^{\hat{r}T} \quad (\mbox{forward price}), \\
    k & =  \ln (K/f)= \ln (K/S)  + \hat{r}T \quad (\mbox{log moneyness}), \\
    V & =v\sqrt{T} \quad (\mbox{time-scaled volatility}).
\end{split}
\right.
\end{equation}
Then \eqref{eq:bs-2} becomes
\begin{equation} \label{eq:bs-gao-lee-1}
  \cbs(k,V;T)
    = S N(d_1) - K \e^{-\hat{r}T} N(d_2),
\end{equation}
where
\begin{equation} \label{eq:bs-gao-lee-2}
\left\{
  \begin{split}
    N(d) & = \int_{-\infty}^d n(\vartheta)\, \d \vartheta,
        \quad n(\vartheta) = \frac{1}{\sqrt{2\pi}} \e^{-\vartheta^2/2}, \\
    d_1
        & = \frac{-k + V^2/2}{V} = -k/V + V/2,\\
    d_2
        & = \frac{k - V^2/2}{V} = k/V + V/2.\\
  \end{split}
\right.
\end{equation}
So in our notation their $C_{-}$ \cite[(3.1)]{gao-lee-2011} is
\begin{equation}
  C_{-}(k,V) = \cbs(k,V;T)/S,
\end{equation}
and their $C_{+}$ \cite[(3.7)]{gao-lee-2011} is
\begin{equation}
  C_{+}(k,V) = 1 - \cbs(k,V;T)/S.
\end{equation}
They \cite[(4.1), (4.2)]{gao-lee-2011} also define
\begin{equation}
  L_+ = -\ln C_{+}, \qquad L_{-} = - \ln C_{-}.
\end{equation}

\subsection{The small time asymptotics}

We will show that the small time result of \cite[Remark 7.4]{gao-lee-2011} do not lead to ours. The formula $V^2 \sim k^2/(2L)\equiv k^2/(2L_{-})$ in \cite[Remark 7.4]{gao-lee-2011} implies that, in our notation,
\begin{equation}
  \iv(K,T)\sqrt{T} \sim \frac{\abs{\ln(K/S) + \hat{r}T}}{\sqrt{-2\ln (\call(K,T)/S)}}
    \qquad (T \to 0).
\end{equation}
Lemma \ref{lem:atm-prelim-1} (1) gives $x \xrightarrow{\; T \to 0\;} \infty$; together with \eqref{eq:rhat-1}, this gives
\begin{equation}
  \hat{r}T = r T - \ln \left(1 - \e^{-x/2} \right)
            \xrightarrow{\; T \to 0\;} 0.
\end{equation}
Hence, in the context of the MMM, Remark 7.4 of \cite{gao-lee-2011} suggests that
\begin{equation}
  \iv(K,T)\sqrt{T} \sim \frac{\abs{\ln(K/S)}}{\sqrt{-2\ln \call(K,T) - 2\ln S}}
        \qquad (T \to 0).
\end{equation}
Rearranging this asymptotic formula then gives
\begin{equation}
  \iv(K,T)\sim \frac{\abs{\ln(K/S)}}{\sqrt{-2T \ln \call(K,T)}} \qquad (T \to 0).
\end{equation}
This is different from the Roper--Rutkowski formula \cite[Theorem 5.1]{roper-rutkowski-09} and it does not lead to our extended version \eqref{eq:rr-formula-extended}.

\subsection{The large time asymptotics}

We now explain why the large time results of \cite[Corollary 7.8]{gao-lee-2011} do not apply in our model either.  Their Case $(+)$ requires $k/L_{+}$ to have a limit in $[0,\infty)$. Translating this condition into our setting gives
\begin{equation} \label{eq:limit-k-l-plus-1}
\begin{split}
  \frac{k}{L_{+}}
    & = \frac{\ln(K/S) - \hat{r}T}{-\ln (1-\cbs/S)}
      = \frac{\ln(K/S) - \hat{r}T}{-\ln (1-\call/S)}
      = \frac{\ln(K/S) - \hat{r}T}{-\ln \mathcal{R}},
\end{split}
\end{equation}
where the last equality follows from \eqref{eq:cbs-call-R}. By Lemma \ref{lem:varphi-x-y-at-infty} (2) and (3), $x,y \xrightarrow{\;T\to\infty\;} 0$. Hence it can be seen from \eqref{eq:rcbs-rcall-def} that $\mathcal{R}\xrightarrow{\;T\to\infty\;} 0$. Now Lemma \ref{lem:rhat-at-infty} (1) gives $\hat{r}  \xrightarrow{\;T\to\infty\;} r+\eta$ and $\hat{r}T  \xrightarrow{\;T\to\infty\;} \infty$. So
\begin{equation} \label{eq:limit-k-l-plus-2}
  \frac{\ln(K/S) - \hat{r}T}{-\ln \mathcal{R}}
    \sim \frac{\hat{r}T}{\ln \mathcal{R}}.
\end{equation}
For ease of notation put
\begin{equation}
\left\{
\begin{split}
  A & = \frac{K}{S} \e^{-\hat{r}T} + \e^{-(r+\eta)T} \underline{\mathcal{R}}, \\
  B & = \frac{K}{S} \e^{-\hat{r}T} + \e^{-(r+\eta)T} \overline{\mathcal{R}}. \\
\end{split}
\right.
\end{equation}
By Lemma \ref{lem:R-asymptotics}, for sufficiently large $T$,
\begin{equation}
  0< A \le \mathcal{R} \le B < 1
    \quad \Longrightarrow \quad
  \frac{\hat{r}T}{\ln B} \le \frac{\hat{r}T}{\ln \mathcal{R}} \le \frac{\hat{r}T}{\ln A}.
\end{equation}
Further, by \eqref{eq:rcalll-large-time-order} and \eqref{eq:rcallu-large-time-order},
\begin{equation}
  \ln A \sim \ln B \sim \ln \left(\frac{K}{S} \e^{-\hat{r}T}\right)  \qquad (T \to \infty).
\end{equation}
Therefore, as $T \to \infty$,
\begin{equation*}
  \frac{\hat{r}T}{\ln A} \sim \frac{\hat{r}T}{\ln(K/S) - \hat{r}T} \xrightarrow{\;T \to \infty \;} -1.
\end{equation*}
Similarly, $\hat{r}T/\ln B \xrightarrow{\;T \to \infty \;} -1.$ Consequently,
$\hat{r}T/\ln \mathcal{R} \xrightarrow{\;T \to \infty \;} -1.$ This implies that in the MMM,
$ k/L_{+} \xrightarrow{\;T \to \infty \;} -1 \not \in [0,\infty).$
%
Hence the condition for Case $(+)$ \cite[Collary 7.8]{gao-lee-2011} is not satisfied in the MMM. We now check the condition for Case $(-)$ \cite[Collary 7.8]{gao-lee-2011}, which requires $k/L_{-}$ to have a limit in $(0,\infty)$. In the MMM,
\begin{equation} \label{eq:limit-k-l-minus-1}
\begin{split}
  \frac{k}{L_{-}}
    & = \frac{\ln(K/S) - \hat{r}T}{-\ln (\cbs/S)}
      = \frac{\ln(K/S) - \hat{r}T}{-\ln (\call/S)}
      = \frac{\ln(K/S) - \hat{r}T}{-\ln( 1 - \mathcal{R})}.
\end{split}
\end{equation}
Invoking \eqref{eq:rcalll-large-time-order} and \eqref{eq:rcallu-large-time-order} again, we have $A, \mathcal{R}, B \xrightarrow{\;T \to \infty \;} 0$; and by Now Lemma \ref{lem:rhat-at-infty} (1), $\hat{r}T  \xrightarrow{\;T\to\infty\;} \infty$. These limits imply $\frac{k}{L_{-}} \xrightarrow{\;T \to \infty \;} \infty \not \in [0,\infty).$
Hence the condition for Case $(-)$ \cite[Collary 7.8]{gao-lee-2011} is also not satisfied in the MMM.


%




\providecommand{\doi}[1]{\discretionary{}{}{}\href{http://dx.doi.org/#1}{DOI:#1}}     




\bibliographystyle{plainnat}   
\bibliography{uts-bib-2009-v04}

%
%
%
%
%
%
%
%

\end{document}